\RequirePackage{silence}
\documentclass[a4paper,onecolumn,accepted=2023-04-22]{quantumarticle}
\pdfoutput=1
\usepackage[numbers,sort&compress]{natbib}
\usepackage{blindtext}
\usepackage{graphicx}
\usepackage{dcolumn}
\usepackage{bm}
\usepackage{amssymb}
\usepackage{color}
\usepackage[T1]{fontenc} 
\usepackage{mathrsfs,amsfonts,dsfont}
\usepackage{nccmath}
\usepackage{amstext}
\usepackage{mathtools}
\usepackage{tablefootnote}
\usepackage[inline]{enumitem}
\graphicspath{{graphics/}}

\usepackage[english]{babel}

\usepackage[table]{xcolor}
\usepackage{footnote}

\usepackage{wrapfig}
\usepackage{lipsum} 
\usepackage{blindtext}

\usepackage{braket}
\usepackage{bbm} 
\usepackage[normalem]{ulem}

\usepackage[printonlyused,withpage]{acronym}
\makeatletter
\AtBeginDocument{%
  \renewcommand*{\AC@hyperlink}[2]{%
    \begingroup
      \hypersetup{hidelinks}%
      \hyperlink{#1}{#2}%
    \endgroup
  }%
}
\makeatother

\usepackage[colorlinks=true,citecolor=blue,linkcolor=blue]
{hyperref}
\RequirePackage{doi}

\usepackage{amsthm}
\newtheorem{theorem}{Theorem}
\newtheorem{Definition}{Definition}
\newtheorem{Lemma}{Lemma}
\newtheorem{Corollary}{Corollary}

\newtheorem*{theorem1}{Theorem 1}

\definecolor{martin}{rgb}{0,.4,1}

\DeclareMathOperator{\Ddiamond}{D_{\diamond}}
\DeclareMathOperator{\Dinfty}{D_{\infty}}
\DeclareMathOperator{\TD}{D_1}
\DeclareMathOperator{\Dell}{D_{\ell_1}}
\DeclareMathOperator{\DP}{D_{\mathrm{p}}}
\DeclareMathOperator{\ETD}{E_1}
\DeclareMathOperator{\Rdiamond}{R_{\diamond}}
\DeclareMathOperator{\RP}{R_{\mathrm{p}}}
\DeclareMathOperator{\Rell}{R_{\ell_1}}
\DeclareMathOperator{\Idiamond}{I_{\diamond}}
\DeclareMathOperator{\Iinfty}{I_{\infty}}
\DeclareMathOperator{\Cdiamond}{C_{\diamond}}
\DeclareMathOperator{\Cinfty}{C_{\infty}}
\DeclareMathOperator{\IFdiamond}{\mathrm{IF}_{\diamond}}
\DeclareMathOperator{\IFinfty}{\mathrm{IF}_{\infty}}

 \hypersetup{pdftitle = {Distance-based resource quantification for sets of quantum measurements},
       pdfauthor = {Lucas Tendick, Martin Kliesch, Hermann Kampermann, Dagmar Bruß},
       pdfsubject = {Quantum information theory}, 
       pdfkeywords = {quantum, weighted, assemblage, resource, theory, framework 
                      resource, theory, quantum channel, SDP, semidefinite program, incompatibility, mutually unbiased bases, MUB, steering, EPR, 
                      nonlocality, Bell, 
                      entanglement, 
                      quantification, distance, metric, 
                      resource, monotone, 
                      assemblage, prepare-and-measure, 
                      distinguishability, 
                      informativeness, coherence, jointly measurable, hierarchy, nonlocal game, 
                      inequality, 
                      quantum advantage, convex, optimization, dual
              }
      }

\begin{document}
\title{Distance-based resource quantification for sets of quantum measurements}
\author[1]{Lucas Tendick}
\email{lucas.tendick@hhu.de}
\author[1,2]{Martin Kliesch}
\author[1]{Hermann Kampermann}
\author[1]{Dagmar Bru\ss}
\affiliation[1]{Institute for Theoretical Physics, Heinrich Heine University D\"usseldorf, 
D-40225 D\"usseldorf, Germany}
\affiliation[2]{Institute for Quantum-Inspired and Quantum Optimization, Hamburg University of Technology, D-21079 Hamburg, Germany}

\begin{abstract}
The advantage that quantum systems provide for certain quantum information processing tasks over their classical counterparts can be quantified within the general framework of resource theories. Certain distance functions between quantum states have successfully been used to quantify resources like entanglement and coherence. Surprisingly, such a distance-based approach has not been adopted to study resources of quantum measurements, where other geometric quantifiers are used instead. Here, we define distance functions between sets of quantum measurements and show that they naturally induce resource monotones for convex resource theories of measurements. By focusing on a distance based on the diamond norm, we establish a hierarchy of measurement resources and derive analytical bounds on the incompatibility of any set of measurements. We show that these bounds are tight for certain projective measurements based on mutually unbiased bases and identify scenarios where different measurement resources attain the same value when quantified by our resource monotone. Our results provide a general framework to compare distance-based resources for sets of measurements and allow us to obtain limitations on Bell-type experiments.
\end{abstract}

\maketitle

\section{Introduction}
It is arguably one of the most remarkable features of quantum theory that certain quantum systems exhibit behaviors without any classical analog. While these quantum phenomena were first just regarded as a strange feature of nature which led to many philosophical questions~\cite{EPR_paper,Bell_seminal,PhysRev.34.163}, it has later been realized that these phenomena can actually be used as a resource in real-world applications such as computation~\cite{2106.10522}, sensing~\cite{RevModPhys.89.035002}, or cryptography~\cite{Pirandola:20}. 
To understand the potential of these upcoming applications, it is essential to characterize the advantages quantum technologies can provide over classical information processing technologies and which physical phenomena enable them. To achieve this advantage, properties of both quantum states and measurements are relevant. Together, quantum states and measurements give rise to purely quantum phenomena that cannot be explained by classical physics, such as entanglement~\cite{RevModPhys.81.865} and its detection \cite{Ghne2009}, EPR-steering~\cite{Steering_resource,Cavalcanti2016,RevModPhys.92.015001}, and Bell nonlocality~\cite{Nonlocality_review,Nonlocality_resource}. \\
\indent The latter two are similar in the sense that both can be seen as resources that require one out of several judiciously chosen quantum measurements to be performed on a resourceful quantum state in each round of an experiment. In particular, it is well-known that entangled states and incompatible measurements are necessary to witness steering or nonlocality~\cite{Cavalcanti2016}. 
Moreover, if appropriate quantifiers are chosen, the amount of incompatibility and entanglement provide upper bounds for the possible amount of steering and nonlocality~\cite{PhysRevA.93.052112,PhysRevLett.116.240401,PhysRevResearch.4.L012002}. 
Therefore, entanglement and incompatibility can be thought of as a resource for quantum advantages. In general, states and measurements may possess various other resources responsible for quantum advantages as well~\cite{Purity_resource,RevModPhys_coherence,Bera2017,PhysRevLett.126.090401,2112.06784,PhysRevLett.119.190501,Guerini2017,PhysRevLett.122.140403,Baek2020}. \\
\indent \Acp{QRT}~\cite{RevModPhys.91.025001} allow us to identify, study, and quantify quantum resources for certain quantum information processing tasks in a general framework. Moreover, this allows us to identify similarities among different resources, adapt concepts and quantification methods~\cite{PhysRevLett.122.130404,PhysRevLett.126.220404,PhysRevX.9.031053,PhysRevLett.125.110401} from one to another, and establish relations between different resources~\cite{Purity_resource,PhysRevLett.115.230402,PhysRevA.93.052112,PhysRevResearch.4.L012002}. Any \ac{QRT} aims to answer at least the following three questions: \begin{enumerate*}[label=(\roman*)] \item Which objects (e.g. states or measurements) are resources for a certain task, and which ones are free, i.e., do not provide any advantage? \item Which transformations are free, i.e., cannot create resources from free objects? \item How can we quantify the amount of the resource?
\end{enumerate*}
A standard approach to quantify quantum resources, illustrated in Figure~\ref{Distance_resource}, asks how far away a given resource is from the set of free objects, as measured by some distance-based function. \\
\begin{figure}
\begin{center}
\includegraphics[scale=1.95]{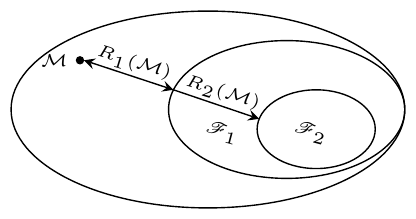}
 \caption{Distance-based resource quantification. A set of measurements $\mathcal{M}$ contains different quantum resources in general. These different quantum resources are associated with their respective sets of free measurements, here denoted by $\mathscr{F}_1$ for \ac{QRT} $Q_1$ and $\mathscr{F}_2$ for \ac{QRT} $Q_2$. The amount of resource in $\mathcal{M}$ associated to $Q_1$ and $Q_2$ is quantified by its distance $\mathrm{R}_1(\mathcal{M})$ to the set $\mathscr{F}_1$ and the distance $\mathrm{R}_2(\mathcal{M})$ to the set $\mathscr{F}_2$, respectively. As all free measurements $\mathcal{F} \in \mathscr{F}_2$ are also contained in $\mathscr{F}_1$ it follows that  $\mathrm{R}_1(\mathcal{M}) \leq \mathrm{R}_2(\mathcal{M})$.}
  \label{Distance_resource}
  \end{center}
\end{figure}
\indent Together with  robustness-based~\cite{PhysRevA.59.141,PhysRevA.67.054305,PhysRevLett.114.060404,Incomop_resource,PhysRevLett.122.140403,Designolle2019} and weight-based quantifiers~\cite{Elitzur1992,PhysRevLett.80.2261,PhysRevLett.112.180404,PhysRevLett.125.110401,Quantifying_coherence,PhysRevLett.125.110402}, distance-based~\cite{PhysRevLett.78.2275,PhysRevA.68.042307,PhysRevResearch.2.012035,PhysRevLett.105.190502} resource quantifiers form the class of so-called geometric quantifiers. One main advantage of these quantifiers is that they generally can be defined for any convex \ac{QRT} \cite{Regula2017,PhysRevLett.122.130404,PhysRevX.9.031053,PhysRevLett.125.110401,Oszmaniec2019,PhysRevLett.122.140402,PhysRevLett.125.110402} (i.e., a \ac{QRT} where the set of free objects is convex), which for instance, allows for a practical way to compare certain resources with each other. The combined insights from all three classes of geometric quantifiers usually give a detailed picture of any convex \ac{QRT}. \\ 
\indent While the distance-based approach has been employed successfully to resources of quantum states like entanglement~\cite{RevModPhys.81.865} and coherence~\cite{RevModPhys_coherence} and correlations like steering \cite{PhysRevA.97.022338} and nonlocality \cite{PhysRevA.97.022111}, it is widely unexplored for quantum measurements, not to mention for sets of measurements. The reason for this is most likely that only recently some strategies for discrimination of single measurements on the basis of distances between them have been proposed \cite{PhysRevA.98.042103,PhysRevA.90.052312}. Robustness-based and weight-based quantifiers have been used as mathematical tools to circumvent this problem.
Later, it was realized that these quantifiers are also linked to operational advantages \cite{Oszmaniec2019,PhysRevLett.122.130404,PhysRevLett.122.130403,PhysRevA.98.012126,PhysRevLett.122.140402,PhysRevLett.122.140403,PhysRevLett.122.140404,PhysRevLett.116.150502,PhysRevLett.114.060404,PhysRevLett.125.110401,2002.03504,PhysRevLett.125.110402}. However, the question about the existence of an operationally meaningful distance between sets of quantum measurements remained open. The interest in detailed studies on distance-based resource quantification goes beyond the question of their existence. Having access to an additional formalism for resource quantification offers new tools to study measurement resources and allows us to better understand their operational significance. \\
\indent In this work, we answer the question of whether distance-based resource quantification for sets of measurements is possible in the affirmative. Hence, we 
extend the class of geometric quantifiers for convex \acp{QRT} for sets of measurements (so-called assemblages) by introducing distance-based resource quantifiers for any convex resource theory of measurement assemblages. \\
\indent First, we discuss the necessary properties any distance between sets of measurements has to fulfil. Then, we show that every such distance induces a resource monotone. We propose one particular quantifier, which is based on the diamond norm~\cite{Kitaev2002} between different measure-and-prepare channels and is specially tailored to Bell-type experiments, as it captures the idea that only one particular measurement out of a given collection is applied at a time in a round-by-round protocol. Based on this quantifier, we establish a hierarchy of measurement resources, including recently introduced steering~\cite{PhysRevA.97.022338} and nonlocality monotones~\cite{PhysRevA.97.022111}. See Table \ref{Resource_table} for an overview of the resources and the quantities we analyze in this work. We show that our quantifier can be computed efficiently by means of a \ac{SDP} which we use to obtain analytical upper and lower bounds on the incompatibility (i.e., the non-joint measurability)~\cite{2112.06784} for any set of measurements. Finally, we show that these bounds are tight for special instances of projective measurements based on \ac{MUB}~\cite{DURT2010}, which also play a special role in cases when different measurement resources attain the same value when quantified with our proposed quantifier. 

\begin{widetext}
\begin{center}

\begin{table}
\scalebox{0.65}{
    \renewcommand{\arraystretch}{1.5}
    \hspace{-1em}
    \begin{minipage}{\textwidth}
    \begin{tabular}{|l|l|l|l|l|l|}
    \hline
    Resource & Monotone & Free objects & Free operations  & Optimization & Type \\  \hline 
      General & $\bullet$ $\Rdiamond(\mathcal{M}_\mathbf{p})$ \eqref{diamond_monotone} & $\mathcal{F} \in \mathscr{F}$ & $\Lambda^{\dagger} \in \mathds{F}$, simulations $\xi \in \mathds{S}$ & \ac{SDP} \eqref{SDP_general_primal},\eqref{SDP_general_dual} & Set\textbackslash Average 
    \\ \hline
    Informativeness & $\bullet$ $\IFdiamond(\mathcal{M}_\mathbf{p})$ \eqref{Informativeness} & $F_{a \vert x} = q(a \vert x) \mathds{1}_d$ \eqref{informativeness_free} 
      &  Unital maps $\Lambda^{\dagger}$  \footnote{Even though it is not discussed in \cite{PhysRevLett.122.140403}, it follows directly from the definition of unitality, that no quantum channel ${\Lambda}^{\dagger}$ can create informativeness from uninformative measurements.}, simulations $\xi$ \cite{PhysRevLett.122.140403} &  \ac{SDP} \eqref{Info_SDP_primal}, \eqref{Info_SDP_dual}& Average  
    \\ \hline
    Coherence & $\bullet$ $\Cdiamond(\mathcal{M}_\mathbf{p})$ \eqref{Coherence} & $F_{a \vert x} = \sum_i \alpha_{i \vert (a,x)} \vert i \rangle \langle i \vert $ \eqref{coherence_free} & SI-operations $\Lambda^{\dagger}_{\mathrm{SIO}}$ \footnote{SIO stands for \emph{strictly incoherent operations}.} \cite{Baek2020}, simulations $\xi$ \footnote{Even though it is not discussed in \cite{Baek2020}, it follows directly from the definition of the classical simulations $\xi$ that they cannot create coherence, as linear combinations of diagonal matricies are diagonal.} & \ac{SDP} \eqref{Coherence_SDP_primal}, \eqref{Coherence_SDP_dual} & Average  \\ \hline Incompatibility & $\bullet$ $\Idiamond(\mathcal{M}_\mathbf{p})$ \eqref{Incompatibility}  & $F_{a \vert x} = \sum_{\lambda} v(a \vert x, \lambda) G_{\lambda}$ \eqref{incompatibility_free} & Unital maps $\Lambda^{\dagger}$ \cite{Incomop_resource}, simulations $\xi$ \cite{Guerini2017} & \ac{SDP} \eqref{Incomp_SDP_primal}, \eqref{Incomp_SDP_dual}& Set \\ \hline
    Steering & $\mathrm{S}(\Vec{\sigma}_\mathbf{p})$ \eqref{Steering_quantifier_distance}, \cite{PhysRevA.97.022338} & $\tau_{a \vert x} = \sum_{\lambda} v(a \vert x, \lambda) \sigma_{\lambda}$ \eqref{LHS_def} & (Restricted) 1W-LOCC \footnote{1W-LOCC stands for \emph{one-way local operations and classical communication}} \cite{PhysRevA.97.022338,PhysRevA.96.022332}& SDP \eqref{Steering_primal}, \eqref{steering_dual_sdp} & Set 
    \\ \hline
    Nonlocality &$\mathrm{N}(\mathbf{q}_\mathbf{p})$ \eqref{nonlocality_distance}, \cite{PhysRevA.97.022111}& 
    \begin{tabular}{@{}l@{}}
      $t(a,b \vert x,y)$  \hspace{2.1cm} \eqref{LHV_model} 
      \\
      $= \sum_{\lambda} \pi(\lambda) v_A(a \vert x, \lambda) v_B(b \vert y, \lambda)$ 
    \end{tabular}    
    & WCCPI \footnote{WCCPI stands for \emph{wirings and classical communication prior to the inputs}.}\cite{PhysRevLett.109.070401} &  Linear \eqref{optimal_Bell_dual}, \eqref{optimal_Bell_primal} & Set 
    \\ \hline
    \end{tabular}
    \end{minipage}
   }
    
    \caption{Overview over the resources analyzed in this work. The different resources are presented in terms of the monotones we consider, the respective free objects, and the set of free operations associated to the considered \ac{QRT}. The monotones we introduce in this work are marked with a bullet point $\bullet$. Furthermore, we present by which kind of optimization the respective monotone can be computed and whether the resources are genuine properties of a set of objects or an average over single object properties. The free operations for steering and nonlocality are listed for completeness here and we refer to the references in the table for more details.}
    
    \label{Resource_table}
    \end{table}
  
\end{center}
\end{widetext}

\section{Distance-based resource quantification}  

Consider the canonical example of the trace distance~\cite{nielsen_chuang_2010}. The trace distance between two quantum states $\rho, \tau \in \mathcal{S}(\mathcal{H})$, where $\mathcal{S}(\mathcal{H})$ is the set of density matrices acting on a Hilbert space $\mathcal{H} \cong \mathds{C}^d$ of finite dimension $d$, is given by
\begin{align}
\TD(\rho,\tau) = \dfrac{1}{2}\lVert \rho - \tau \rVert_1 \geq 0,    
\end{align}
where $\lVert X \rVert_1 = \mathrm{Tr}[\sqrt{X^{\dagger} X}]$ is the trace norm of $X$. The trace distance is a useful tool to distinguish $\rho$ and $\tau$, as it fulfils all necessary properties of a metric between quantum states. Consider $\rho, \tau, \chi \in \mathcal{S}(\mathcal{H})$
and any \ac{CPT} map $\Lambda$ also known as quantum channel. It holds that
\begin{align}
\label{distance_conditions_states}
&\TD(\rho,\tau) = 0 \iff \rho = \tau,  \\
&\TD(\rho,\tau) = \TD(\tau,\rho), \nonumber \\
&\TD(\rho,\tau) \leq \TD(\rho,\chi) + \TD(\chi,\tau), \nonumber \\
&\TD(\rho,\tau) \geq \TD(\Lambda(\rho), \Lambda(\tau)), \nonumber
\end{align}
i.e., $\TD(\rho, \tau)$ is a faithful and symmetric function that obeys the triangle inequality and monotonicity (i.e. it does not increase) under arbitrary \ac{CPT} maps $\Lambda$. In addition to these minimal requirements, it is well known that $\TD(\rho, \tau)$ has an operational interpretation in terms of the optimal probability to distinguish $\rho$ and $\tau$ in a single-shot experiment~\cite{nielsen_chuang_2010}. That is, the optimal guessing probability is given by $p^{(\rho,\tau)}_{1,\text{guess}} = \dfrac{1}{2}(1+\TD(\rho, \tau))$. These properties make the trace distance a viable tool to quantify (convex) resources. 

Let us consider the prime example of a resource, the entanglement of a bipartite state $\rho \in \mathcal{S}(\mathcal{H} \otimes \mathcal{H})$. One can quantify the entanglement of $\rho$ by its distance to the set $\mathrm{Sep}(\mathcal{H} \otimes \mathcal{H})$ of separable quantum states~\cite{PhysRevLett.78.2275} given as
\begin{align}
\ETD(\rho) = \min\limits_{\rho_{S} \in\mathrm{Sep}(\mathcal{H} \otimes \mathcal{H})} \TD(\rho,\rho_S).  \label{entanglement}
\end{align}
It is now readily verified that $\ETD(\rho)$ is a non-negative, convex function with $\ETD(\rho) = 0 \iff \rho \in \mathrm{Sep}(\mathcal{H} \otimes \mathcal{H})$ obeying the monotonicity $\ETD(\rho) \geq \ETD(\Lambda_{\mathrm{LOCC}}(\rho))$ under any \ac{LOCC} \cite{RevModPhys.81.865} map $\Lambda_{\mathrm{LOCC}}$. That is, $\ETD(\rho)$ is a faithful (i.e. $\ETD(\rho) = 0 \iff \rho \in \mathrm{Sep}(\mathcal{H} \otimes \mathcal{H})$) convex resource monotone. Note that the monotonicity $\ETD(\rho) \geq \ETD(\Lambda_{\mathrm{LOCC}}(\rho))$ captures the fact that \ac{LOCC} maps cannot create entanglement. The monotonicity of resources under these so-called free operations is a crucial property of any resource theory and also reflects that these types of operations cannot create resources from free objects \cite{RevModPhys.91.025001}. \\
\indent We can use the insights for distances and resources of quantum states to define distance-based resource monotones for sets of quantum measurements in the following. We describe a quantum measurement most generally by a \ac{POVM} i.e., a finite set $\lbrace M_a \rbrace_a$ of effect operators $0 \leq M_a \leq \mathds{1}_d$, acting on a finite $d$-dimensional Hilbert space $\mathcal{H}$ such that $\sum_a M_{a} = \mathds{1}_d$, where $\mathds{1}_d$ is the identity operator on $\mathcal{H}\cong\mathbb C^d$. Note that $X \leq Y \iff Y-X \geq 0 $ for any Hermitian operators $X,Y$ means that the operator $Y-X$ is positive semidefinite. A set of \acp{POVM} with outcomes $a$ for different settings $x$ is known as measurement assemblage $\mathcal{M} = \lbrace M_{a \vert x} \rbrace_{a,x}$. Note that we will omit in the following the set-indices and simply write $\mathcal{M} = \lbrace M_{a \vert x} \rbrace$ when there is no risk of confusion. If we talk about a specific element of the assemblage $\mathcal{M}$, for instance the \ac{POVM} corresponding to setting $x$, we will write $\mathcal{M}_x = \lbrace M_{a \vert x} \rbrace_a$. Here, we consider assemblages with $m$ measurement settings and $o$ outcomes in each setting, i.e. $x = 1, \cdots, m$ and $a = 0, \cdots, o-1$. The outcome statistics of a measurement on any state $\rho$ is given by  $p(a,x) = p(x)p(a \vert x) = p(x)\mathrm{Tr}[M_{a \vert x} \rho]$, where $p(x)$ is the probability to choose the setting $x$. \\
\indent A measurement assemblage can be converted by two different processes to another assemblage. First, as any quantum state $\rho$ can be transformed via any \ac{CPT} map $\Lambda$ to another state $\Lambda(\rho)$, it follows from $\mathrm{Tr}[M_{a \vert x} \Lambda(\rho)] = \mathrm{Tr}[\Lambda^{\dagger}(M_{a \vert x}) \rho]$ that an assemblage $\mathcal{M}$ can be transformed via the Hilbert-Schmidt adjoint (unital) map $\Lambda^{\dagger}$ to another assemblage $\Lambda^{\dagger}(\mathcal{M})$. Second, classical simulations (via mixtures and classical post-processing) maps $\mathcal{M}' = \xi(\mathcal{M})$ with $ M'_{b \vert y} = \sum_x p(x \vert y) \sum_a q(b \vert y,x,a) M_{a \vert x}$ can be used to simulate~\cite{Guerini2017} the assemblage $\mathcal{M}'$ from $\mathcal{M}$  via the conditional probabilities $p(x \vert y)$ and $q(b \vert y,x,a)$ for all $y$, respectively for all $y,x,a$. Note that as $p(x) = \sum_y q(y) p(x \vert y)$ one also obtains the allowed probabilities $q(y)$ to perform setting $y$. See also \cite{Pusey2015} for an approach to simulability that combines quantum pre-processing and classical post-processing. \\
\indent We use the probability distribution $\mathbf{p} = \lbrace p(x) \rbrace$ to capture the fact that typically only one quantum measurement can be performed at a time and it is also natural to assume that the likelihood of the settings $x$ influences the capabilities of $\mathcal{M}$ in experiments. Note that we consider only the case $p(x)>0 \ \forall \ x$, as measurements that are never performed can be discarded trivially. We define a distance between sets of measurements weighted with the distribution $\mathbf{p}$ as follows. 
\begin{Definition}
\label{Def_Distance}
Let $\mathcal{M}$ be a measurement assemblage containing $m$ \acp{POVM} and let $\mathbf{p}$ be a probability distribution with $p(x) > 0 \ \forall \ x =1,\cdots,m$.\,We call the tuple $\mathcal{M}_{\mathbf{p}} \coloneqq (\mathcal{M},\mathbf{p})$ a \emph{\ac{WMA}}. Let $\mathcal{M}_\mathbf{p}$, $\mathcal{N}_\mathbf{p}$, and $\mathcal{K}_\mathbf{p}$ be any \acp{WMA}. Further, let $\Lambda^{\dagger}$ be any \ac{CP} unital map and $\xi$ any classical simulation map. Any non-negative function $\mathrm{D}(\mathcal{M}_\mathbf{p}, \mathcal{N}_\mathbf{p})$ that fulfils the conditions
\begin{align}
&\mathrm{D}(\mathcal{M}_\mathbf{p},\mathcal{N}_\mathbf{p}) = 0 \iff \mathcal{M} = \mathcal{N}, \label{Distance_conditions} \\
&\mathrm{D}(\mathcal{M}_\mathbf{p} , \mathcal{N}_\mathbf{p}) = \mathrm{D}(\mathcal{N}_\mathbf{p} , \mathcal{M}_\mathbf{p}), \nonumber \\
&\mathrm{D}(\mathcal{M}_\mathbf{p}, \mathcal{N}_\mathbf{p}) \leq \mathrm{D}(\mathcal{M}_\mathbf{p}, \mathcal{K}_\mathbf{p}) + \mathrm{D}(\mathcal{K}_\mathbf{p}, \mathcal{N}_\mathbf{p}), \nonumber \\
&\mathrm{D}(\mathcal{M}_\mathbf{p}, \mathcal{N}_\mathbf{p}) \geq  \mathrm{D}(\Lambda^{\dagger}(\mathcal{M})_\mathbf{p}, \Lambda^{\dagger}(\mathcal{N})_\mathbf{p}), \nonumber \\
&\mathrm{D}(\mathcal{M}_\mathbf{p}, \mathcal{N}_\mathbf{p}) \geq \mathrm{D}(\xi(\mathcal{M}_\mathbf{p})_\mathbf{q} , \xi(\mathcal{N}_\mathbf{p})_\mathbf{q}). \nonumber
\end{align}
is a \emph{distance} between $\mathcal{M}_\mathbf{p}$ and $\mathcal{N}_\mathbf{p}$.
\end{Definition}
Note that all conditions are in direct correspondence to the conditions in Eq.\,\eqref{distance_conditions_states} for quantum states. Any distance that fulfills the conditions in Definition~\ref{Def_Distance} can be used to define a faithful resource monotone for convex \acp{QRT} of measurement assemblages.
\begin{Definition}
\label{QRT}
Let $\mathscr{F}$ be a convex and compact set of measurement assemblages, $\mathds{F}$ the (maximal) set of free quantum maps $\Lambda^{\dagger}$ such that $\Lambda^{\dagger}(\mathcal{F}) \in \mathscr{F}$ for any $\mathcal{F} \in \mathscr{F}$, and let $\mathds{S}$ be the set of classical simulations $\xi$ such that $\xi(\mathcal{F}) \in \mathscr{F}$ for any $\mathcal{F} \in \mathscr{F}$. The tuple $Q \coloneqq (\mathscr{F},\mathds{F}, \mathds{S})$ is called a \emph{\ac{QRT} of measurement assemblages}. \end{Definition}

We want to emphasize that all of our considerations hold for the maximal set $\mathds{F}$ of free quantum maps. Therefore, they also hold for any subset of free operations. In some situations, not considering the maximal set of free operations might be physically more motivated, as it is the case for \ac{LOCC} in the resource theory of entanglement. 

\begin{Definition}
Let $Q = (\mathscr{F},\mathds{F}, \mathds{S})$ be a \ac{QRT} of \acp{WMA} $\mathcal{M}_\mathbf{p}$, $\mathcal{N}_\mathbf{p}$. Any non-negative function $\mathrm{R}(\mathcal{M}_\mathbf{p})$ that fulfils 
\begin{align}
&\mathrm{R}(\mathcal{M}_\mathbf{p}) = 0 \iff \mathcal{M} \in \mathscr{F},  \\
&\mathrm{R}(\mathcal{M}_\mathbf{p}) \geq \mathrm{R}(\Lambda^{\dagger}(\mathcal{M})_\mathbf{p}), \ \forall \ \Lambda^{\dagger} \in \mathds{F}, \nonumber \\
&\mathrm{R}(\mathcal{M}_\mathbf{p}) \geq \mathrm{R}(\xi(\mathcal{M}_\mathbf{p})_\mathbf{q}), \ \forall \ \xi \in \mathds{S}, \nonumber 
\end{align}
 is a \emph{faithful resource monotone} of \acp{WMA}. Moreover, if $\mathrm{R}(\mathcal{M}_\mathbf{p})$ fulfills 
\begin{align}
&\mathrm{R}(\eta\mathcal{M}_\mathbf{p}+(1-\eta)\mathcal{N}_\mathbf{p}) \leq \eta \mathrm{R}(\mathcal{M}_\mathbf{p}) +(1-\eta)\mathrm{R}(\mathcal{N}_\mathbf{p}),    
\end{align}
for any $ \eta \in [0,1] $ it is is a \emph{faithful convex resource monotone} of \acp{WMA}. 
\end{Definition}
With these definitions we obtain the following lemma, showing that every (jointly-convex) distance between measurement assemblages induces a faithful (convex) resource monotone.
\begin{Lemma}
\label{lem1}
Let $Q = (\mathscr{F},\mathds{F},\mathds{S})$ be any \ac{QRT} of \acp{WMA} $\mathcal{M}_\mathbf{p}$ and $\mathrm{D}(\mathcal{M}_\mathbf{p}, \mathcal{F}_\mathbf{p})$ a (jointly convex) distance function. The distance of $\mathcal{M}$ to the set $\mathscr{F}$ weighted with the probability $\mathbf{p}$ given by 
\begin{align}
\mathrm{R}(\mathcal{M}_\mathbf{p}) \coloneqq  \min\limits_{\mathcal{F} \in \mathscr{F}} \mathrm{D}(\mathcal{M}_\mathbf{p}, \mathcal{F}_\mathbf{p}),
\end{align}
is a faithful (convex) resource monotone. 
\end{Lemma}
\begin{proof}
The proof relies mainly on the conditions in Definition~\ref{Def_Distance}. The non-negativity and faithfulness (i.e., $\mathrm{R}(\mathcal{M}_\mathbf{p}) = 0 \iff \mathcal{M} \in \mathscr{F}$) follow directly, and the monotonicity conditions follow from
\begin{align}
\mathrm{R}(\mathcal{M}_\mathbf{p}) =  \min\limits_{\mathcal{F} \in \mathscr{F}} \mathrm{D}(\mathcal{M}_\mathbf{p}, \mathcal{F}_\mathbf{p}) \geq  \min\limits_{\mathcal{F} \in \mathscr{F}} \mathrm{D}(\Lambda^{\dagger}(\mathcal{M})_\mathbf{p}, \Lambda^{\dagger}(\mathcal{F})_\mathbf{p}) \geq  \min\limits_{\mathcal{F}' \in \mathscr{F}} \mathrm{D}(\Lambda^{\dagger}(\mathcal{M})_\mathbf{p}, \mathcal{F}'_\mathbf{p}) = \mathrm{R}(\Lambda^{\dagger}(\mathcal{M})_\mathbf{p}), 
\end{align}
where we used the monotonicity of the distance and the fact that free operations $\Lambda^{\dagger} \in \mathds{F}$ map free assemblages to free assemblages. An analogous calculation follows for the simulations $\xi \in \mathds{S}$. In addition, if $\mathrm{D}(\mathcal{M}_\mathbf{p}, \mathcal{F}_\mathbf{p})$ is jointly convex, i.e., obeys 
\begin{align}
\mathrm{D}(\eta \mathcal{M}^{(1)}_\mathbf{p} + (1-\eta) \mathcal{M}^{(2)}_\mathbf{p}, \eta \mathcal{F}^{(1)}_\mathbf{p} + (1-\eta) \mathcal{F}^{(2)}_\mathbf{p}) \leq \eta \mathrm{D}(\mathcal{M}^{(1)}_\mathbf{p}, \mathcal{F}^{(1)}_\mathbf{p}) + (1-\eta) \mathrm{D}(\mathcal{M}^{(2)}_\mathbf{p}, \mathcal{F}^{(2)}_\mathbf{p}),
\end{align}
for any $ \eta \in [0,1] $ and any \acp{WMA} $\mathcal{M}^{(1)}_\mathbf{p},\mathcal{M}^{(2)}_\mathbf{p},\mathcal{F}^{(1)}_\mathbf{p},\mathcal{F}^{(2)}_\mathbf{p}$ it holds that $\mathrm{R}(\mathcal{M}_\mathbf{p})$ is convex. This follows from
\begin{align}
\mathrm{R}(\eta \mathcal{M}^{(1)}_\mathbf{p} + (1-\eta)  \mathcal{M}^{(2)}_\mathbf{p}) &\coloneqq \min\limits_{\mathcal{F} \in \mathscr{F}} \mathrm{D}(\eta \mathcal{M}^{(1)}_\mathbf{p} + (1-\eta)  \mathcal{M}^{(2)}_\mathbf{p}, \mathcal{F}_{\mathbf{p}}) \\
&\leq \mathrm{D}(\eta \mathcal{M}^{(1)}_\mathbf{p} + (1-\eta)  \mathcal{M}^{(2)}_\mathbf{p}, \eta \mathcal{F}^{(1) *}_\mathbf{p} + (1-\eta)  \mathcal{F}^{(2) *}_\mathbf{p}) \nonumber \\
&\leq \eta \mathrm{D}(\mathcal{M}^{(1)}_\mathbf{p}, \mathcal{F}^{(1) *}_\mathbf{p}) + (1-\eta) \mathrm{D}(\mathcal{M}^{(2)}_\mathbf{p}, \mathcal{F}^{(2) *}_\mathbf{p}) \nonumber \\
&= \eta \mathrm{R}(\mathcal{M}^{(1)}_\mathbf{p}) + (1-\eta) \mathrm{R}(\mathcal{M}^{(2)}_\mathbf{p}), \nonumber 
\end{align}
where we used that $\mathcal{F}^{(1) *}, \mathcal{F}^{(2) *} \in \mathscr{F}$ are the closest free assemblages to $\mathcal{M}^{(1)}$ and $\mathcal{M}^{(2)}$, respectively. Furthermore, we used that $\eta \mathcal{F}^{(1) *}_\mathbf{p} + (1-\eta)  \mathcal{F}^{(2) *}_\mathbf{p} $ is free as well, by the convexity of $\mathscr{F}$. Note that the arguments used here are similar to those for distance-based resource monotones of quantum states.
\end{proof}
We propose in the following a specific distance on which we focus on in the remainder of the work (see however the appendix for alternatives). More specifically, we associate to any \ac{POVM} $\mathcal{M}_x = \lbrace M_{a \vert x} \rbrace_a$ a measure-and-prepare channel defined by \begin{align} \label{measure_prepare}
\Lambda_{\mathcal{M}_x}(\rho) = \sum_a \mathrm{Tr}[M_{a \vert x} \rho] \vert a \rangle \langle a \vert,
\end{align}
where the register states $ \vert a \rangle $ form an orthonormal basis $\lbrace \vert a \rangle \rbrace_{0 \leq a \leq o-1}$. Note that the channel $\Lambda_{\mathcal{M}_x}$ can equivalently be described by its Choi–Jamio\l kowski-matrix (see e.g.\,\cite{Watrous2018}). The Choi–Jamio\l kowski-matrix of a quantum channel is obtained by applying a given channel to the first subsystem of the (unnormalized) maximally entangled state $\vert \tilde{\Phi}^+ \rangle = \sum_{i = 0}^{d-1} \lvert ii \rangle$. More precisely, the Choi–Jamio\l kowski-matrix of a measure-and-prepare channel as described in Eq.\,\eqref{measure_prepare} is given by
\begin{align}
 J(\mathcal{M}_x) = (\Lambda_{\mathcal{M}_x} \otimes \mathds{1})(\lvert \tilde{\Phi}^+ \rangle \langle \tilde{\Phi}^+ \rvert) = \sum_a \lvert a \rangle \langle a \rvert \otimes M_{a \vert x}^T,   \label{Choi–Jamiołkowski_state}
\end{align}
where the transpose is with respect to the computational basis. 
We denote the diamond distance between two quantum channels $\Lambda_1, \Lambda_2$ by
\begin{align}
\Ddiamond(\Lambda_1, \Lambda_2) = \max\limits_{\rho \in \mathcal{S}(\mathcal{H} \otimes \mathcal{H})} \dfrac{1}{2} \lVert ((\Lambda_1-\Lambda_2) \otimes \mathds{1}_{d})\rho \rVert_1.
\end{align}
Due to the connection to the trace distance, the diamond distance determines the optimal single-shot probability $p^{(\Lambda_1,\Lambda_2)}_{\diamond,\mathrm{guess}} =  \dfrac{1}{2}(1+\Ddiamond(\Lambda_1, \Lambda_2))$ to distinguish between $\Lambda_1$ and $\Lambda_2$. We want to make use of this operational relevance in the following, by designing a distance between measurement assemblages that has a similar operational interpretation. Based on the diamond distance, we propose the distance $\Ddiamond(\mathcal{M}_\mathbf{p}, \mathcal{N}_\mathbf{p})$ between the \acp{WMA} defined as
\begin{align}
\Ddiamond(\mathcal{M}_\mathbf{p}, \mathcal{N}_\mathbf{p}) \coloneqq \sum_x p(x)  \Ddiamond(\Lambda_{\mathcal{M}_x}, \Lambda_{\mathcal{N}_x}),  \label{assemblage_distance}
\end{align}
and its induced resource monotone 
\begin{align}
\Rdiamond(\mathcal{M}_\mathbf{p}) \coloneqq \min\limits_{\mathcal{F} \in \mathscr{F}} \sum_{x} p(x)  \Ddiamond(\Lambda_{\mathcal{M}_x}, \Lambda_{\mathcal{F}_x}).   \label{diamond_monotone} 
\end{align} 
Note that the diamond distance between measure-and prepare-channels has also been introduced in the context of single \ac{POVM} discrimination \cite{PhysRevA.98.042103,PhysRevA.90.052312}. To prove that $\Rdiamond(\mathcal{M}_\mathbf{p})$ is a convex resource monotone, we need to show that $\Ddiamond(\mathcal{M}_\mathbf{p}, \mathcal{N}_\mathbf{p})$ is distance function according to the conditions in Definition~\ref{Def_Distance} and that it is a jointly-convex function.
\begin{theorem}
\label{thrm1}
The function $\Ddiamond(\mathcal{M}_\mathbf{p}, \mathcal{N}_\mathbf{p})$ is a distance function between the \acp{WMA} $\mathcal{M}_\mathbf{p}$ and $\mathcal{N}_\mathbf{p}$, i.e., it fulfils all the conditions stated in Definition~\ref{Def_Distance}. Moreover, $\Ddiamond(\mathcal{M}_\mathbf{p}, \mathcal{N}_\mathbf{p})$ is jointly-convex.
\end{theorem}
\begin{proof}
The proof relies mostly on the properties of the diamond distance. It is possible to rewrite 
\begin{align}
&\Ddiamond(\mathcal{M}_\mathbf{p}, \mathcal{N}_\mathbf{p}) = \dfrac{1}{2} \sum_{x} p(x) \max\limits_{\rho} \sum_a \lVert  \sigma_{a \vert x}(\rho) - \tau_{a \vert x}(\rho) \rVert_1,    \label{steering_interpretation_incompatibility}
\end{align}
where we have introduced $\sigma_{a \vert x}(\rho) = \mathrm{Tr}_1[(M_{a \vert x} \otimes \mathds{1}) \rho]$ and $\tau_{a \vert x}(\rho) = \mathrm{Tr}_1[(N_{a \vert x} \otimes \mathds{1}) \rho]$. Note that $\mathrm{Tr}_1[\cdot]$ denotes the trace with respect to the first subsystem and that we omit here and in the following the Hilbert space $\rho$ acts on. All conditions in Definition~\ref{Def_Distance} and the joint-convexity can now be verified by direct computation. See Appendix~\ref{Append_thrm1} for all details. 
\end{proof}
Note that it follows directly from its definition that $\Rdiamond(\mathcal{M}_\mathbf{p})$ is upper bounded by $\Rdiamond(\mathcal{M}_\mathbf{p}) \leq 1$, and that it fulfills the continuity condition 
\begin{align}
\lvert \Rdiamond(\mathcal{M}_\mathbf{p}) - \Rdiamond(\mathcal{N}_\mathbf{p}) \rvert \leq \Ddiamond(\mathcal{M}_\mathbf{p}, \mathcal{N}_\mathbf{p}),    
\end{align}
due to the triangle inequality for the diamond norm. 
Moreover, it can be rewritten as 
\begin{align}
\Rdiamond(\mathcal{M}_\mathbf{p}) = \min\limits_{\mathcal{F} \in \mathscr{F}} 2 \sum_x  p(x) p_{\diamond,\mathrm{guess}}^{(\mathcal{M},\mathcal{F})}(x) -1, \label{first_interpretation}
\end{align}
with $p_{\diamond,\mathrm{guess}}^{(\mathcal{M},\mathcal{F})}(x) = \tfrac{1}{2}(1+\Ddiamond(\Lambda_{\mathcal{M}_x},\Lambda_{\mathcal{F}_x}))$
which is up to normalization the average optimal probability to distinguish the resources $\mathcal{M}$ from the free measurements $\mathcal{F}$ in a single-shot experiment. Hence, Eq.~\eqref{first_interpretation} reveals the desired operational significance of $\Rdiamond(\mathcal{M}_\mathbf{p})$ in terms of an average single-shot distinguishability. See also Figure~\ref{figure_diamond_distance} for an illustration of the operational meaning of $\Rdiamond(\mathcal{M}_\mathbf{p})$.
 \begin{figure}
 \centering
\includegraphics{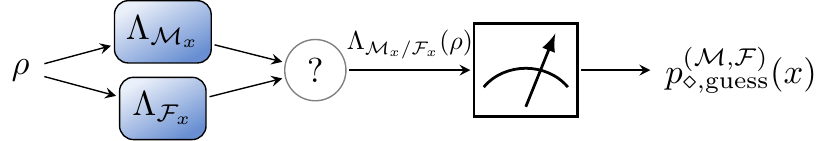}
 \caption{Illustration of the idea to use the diamond distance as resource monotone. Quantum measurements $\mathcal{M}_x,\mathcal{F}_x$ are associated with quantum channels $\Lambda_{\mathcal{M}_x}, \Lambda_{\mathcal{F}_x}$. These are distinguished by applying the channels to an optimal quantum state $\rho$ and performing an ideal dichotomic measurement afterwards to distinguish between the output of the channels $\Lambda_{\mathcal{M}_x}$ and $\Lambda_{\mathcal{F}_x}$. The probability $p_{\diamond,\mathrm{guess}}^{(\mathcal{M},\mathcal{F})}(x)$ tells us how distinguishable the resourceful measurement $\mathcal{M}_x$ is from the free measurements $\mathcal{F}_x$.}
 \label{figure_diamond_distance}
\end{figure}

\section{Hierarchy of measurement resources}  

One main goal while studying \acp{QRT} is to obtain relations between different resources. In particular, we want to understand how one resource limits another quantitatively. This will show one strength of a geometric quantifier, as it can be defined for various resource theories and the discussion often reduces to an analysis of the free sets $\mathscr{F}$. In the following, we will establish a hierarchy of measurement resources based on the newly introduced quantifier $\Rdiamond(\mathcal{M}_\mathbf{p})$. We start by introducing the different resources. \\
\indent The most basic resource of an assemblage is its informativeness~\cite{PhysRevLett.122.140403}. The informativeness of a \ac{WMA} quantifies how valuable it is to actually perform measurements compared to randomly guessing the outcomes in an experiment. An assemblage $\mathcal{M}$ is called \ac{UI} if 
\begin{align}
M_{a \vert x} = q(a \vert x) \mathds{1}_d \ \forall \ a,x, \label{informativeness_free}
\end{align}
where $\lbrace q(a \vert x) \rbrace$ are some probability distributions of $a$ conditioned on setting $x$. These measurements are \ac{UI} as their measurement result does not depend on the quantum state. We denote the set of \ac{UI} assemblages by $\mathscr{F}_{\mathrm{UI}}$ and introduce the informativeness monotone
\begin{align}
\label{Informativeness}
\IFdiamond(\mathcal{M}_\mathbf{p}) = \min\limits_{\mathcal{F} \in \mathscr{F}_{\mathrm{UI}}} \sum_{x} p(x)  \Ddiamond(\Lambda_{\mathcal{M}_x}, \Lambda_{\mathcal{F}_x}).    
\end{align}
Note that measurement informativeness was initially introduced only for a single \ac{POVM} and studied in terms of the generalized robustness~\cite{PhysRevLett.122.140403}. We have extended the notion here by considering the average informativeness of $\mathcal{M}_\mathbf{p}$. \\
\indent A resource that is the foundation for the distinction between \textit{classical} and \textit{quantum} systems is the coherence of measurements~\cite{Baek2020}. An assemblage $\mathcal{M}$ is incoherent (in some predefined orthonormal basis $\lbrace \vert i \rangle  \rbrace$) if 
\begin{align}
M_{a \vert x} = \sum_{i} \alpha_{i \vert (a,x)} \vert i \rangle \langle i \vert \ \forall \ a,x,  \label{coherence_free}
\end{align} 
where $\alpha_{i \vert (a,x)} = \langle i \vert M_{a \vert x} \vert i \rangle$. These measurements cannot distinguish quantum states $\rho$ from their fully dephased versions $\Delta(\rho) = \sum_i \lvert i \rangle \langle i \rvert  \rho \rvert i \rangle \langle i \rvert$, hence they cannot detect coherence. We denote the set of incoherent assemblages by $\mathscr{F}_{\mathrm{IC}}$ and introduce the coherence monotone
\begin{align}
\label{Coherence}
\Cdiamond(\mathcal{M}_\mathbf{p}) = \min\limits_{\mathcal{F} \in \mathscr{F}_{\mathrm{IC}}} \sum_{x} p(x)  \Ddiamond(\Lambda_{\mathcal{M}_x}, \Lambda_{\mathcal{F}_x}).    
\end{align}
Similarly to the informativeness, the coherence of measurements was initially introduced for a single \ac{POVM} and we have extended it here by considering the average coherence of $\mathcal{M}_\mathbf{p}$. See also \cite{PhysRevLett.126.220404} for a different approach to coherence of measurement assemblages. \\
\indent The incompatibility of measurements is probably the best-known example of a \ac{QRT} for measurements and has been studied extensively in recent years~\cite{2112.06784,Heinosaari2016,Incomop_resource,PhysRevLett.122.050402,Designolle2019,PhysRevLett.122.130403}. Contrary to classical physics, different quantum measurements may be incompatible, i.e., they cannot be performed simultaneously and one cannot access their joint measurement statistics as famously illustrated by the Heisenberg-Robertson uncertainty relation~\cite{PhysRev.34.163}. Initially interpreted as a drawback, this phenomenon lies at the heart of Bell-type experiments, as incompatibility is a necessary prerequisite to witness steering and nonlocality.
An assemblage $\mathcal{M}$ is called compatible or \ac{JM} if the statistics of $\mathcal{M}$ can be simulated by a single measurement via some \ac{POVM} $\lbrace G_\lambda \rbrace$ and classical post-processing via the deterministic probability distributions $\lbrace v(a \vert x, \lambda) \rbrace$ such that 
\begin{align}
M_{a \vert x} = \sum_{\lambda} v(a \vert x, \lambda) G_{\lambda} \ \forall \ a,x,    \label{incompatibility_free}
\end{align}
 and is called incompatible otherwise. Note that using deterministic post-processings $\lbrace v(a \vert x, \lambda) \rbrace$ (which represent the vertices of the corresponding probability polytope) is not a restriction as all randomness from non-deterministic distributions can be put inside $G_{\lambda}$. 
 We denote the set of \ac{JM} assemblages by $\mathscr{F}_{\mathrm{JM}}$ and introduce the incompatibility monotone
\begin{align}
\label{Incompatibility}
\Idiamond(\mathcal{M}_\mathbf{p}) = \min\limits_{\mathcal{F} \in \mathscr{F}_{\mathrm{JM}}} \sum_{x} p(x)  \Ddiamond(\Lambda_{\mathcal{M}_x}, \Lambda_{\mathcal{F}_x}).    
\end{align}
It is important to note that the incompatibility in Eq.~\eqref{Incompatibility} is not the average of single \ac{POVM} properties, as incompatibility is always a property of sets of measurements. Therefore, the incompatibility $\Idiamond(\mathcal{M}_\mathbf{p})$ is qualitatively different from the coherence or informativeness. \\
\indent Similar to entanglement, incompatibility can be witnessed in a Bell-type experiment, as both are necessary resources for steering and nonlocality. Consider the \ac{WMA} $\mathcal{M}_\mathbf{p}$ and any bipartite quantum state $\rho$ shared by two-parties, Alice and Bob. By performing the measurements $\mathcal{M}_\mathbf{p}$ on her share of the state, Alice prepares the conditional states 
\begin{align}
\sigma_{a \vert x} = \mathrm{Tr}_1[(M_{a \vert x} \otimes \mathds{1}) \rho],    
\end{align}
for Bob. Here $p(a \vert x) = \mathrm{Tr}[\sigma_{a \vert x}]$ is the probability to obtain $\sigma_{a \vert x}$. We denote the obtained state assemblage by $ \Vec{\sigma} = \lbrace \sigma_{a \vert x} \rbrace$ and its weighted version by $\Vec{\sigma}_\mathbf{p} = (\Vec{\sigma} , \mathbf{p})$.

To make sure that Alice performs incompatible measurements on an entangled state, she can prove that she can demonstrate steering. A state assemblage $\Vec{\sigma}$ is said to be steerable if it cannot be obtained from a \ac{LHS} given by
\begin{align}
\sigma_{a \vert x} = \sum_{\lambda} v(a \vert x, \lambda) \sigma_{\lambda} \ \forall a,x,   \label{LHS_def}
\end{align}
where the $\sigma_{\lambda}$ are operators that satisfy $\sigma_{\lambda} \geq 0 \ \forall \lambda$ and $\mathrm{Tr}[\sum_{\lambda} \sigma_{\lambda}] = 1$. Otherwise we say $\Vec{\sigma} $ is unsteerable which we denote by $\Vec{\sigma}  \in \mathrm{LHS}$. Steering can also be quantified and we use the distance-based monotone introduced by Ku et al.~\cite{PhysRevA.97.022338} as
\begin{align}
\mathrm{S}(\Vec{\sigma}_\mathbf{p}) = \min\limits_{\Vec{\tau} \in \mathrm{LHS}} \dfrac{1}{2}\sum\limits_{a,x} p(x) \lVert \sigma_{a \vert x} - \tau_{a \vert x} \rVert_1. \label{Steering_quantifier_distance}
\end{align}
Note that originally an additional consistency constraint $\sum\limits_a \tau_{a\vert x} = \sum\limits_a \sigma_{a\vert x}$ was introduced~\cite{PhysRevA.97.022338}. However, we do not require this constraint here, as consistency constraints are sometimes introduced for mathematical convenience which provides no advantages in our considerations. For more details on consistent quantifiers, see also \cite{PhysRevA.93.052112}. \\
\indent Consider now that both parties, Alice and Bob, want to prove that they perform incompatible measurements on an entangled state. Let $\mathcal{M}_{\mathbf{p}_A}$ and $\mathcal{N}_{\mathbf{p}_B}$ be the \acp{WMA} of Alice and Bob, respectively, and let $\rho$ be their shared quantum state. Alice and Bob obtain the probability distribution $\mathbf{q} = \lbrace q(a,b \vert x,y) \rbrace $ via $q(a,b\vert x,y) = \mathrm{Tr}[(M_{a \vert x} \otimes N_{b \vert y}) \rho]$. Note that $p(x,y) = p_A(x)p_B(y)$ is the probability to choose setting $x$ for Alice and $y$ for Bob and we introduce the tuple $\mathbf{q}_\mathbf{p} = (\mathbf{q},\mathbf{p})$. To assure themselves that they share an entangled state and perform incompatible measurements, they can check whether they can demonstrate nonlocality. A probability distribution $\mathbf{q}$ is local if it can be obtained from a \ac{LHV} given by
\begin{align}
q(a,b \vert x,y)  = \sum\limits_{\lambda} \pi(\lambda) v_A(a \vert x, \lambda) v_B(b \vert y, \lambda) \ \forall a,b,x,y, \label{LHV_model}
\end{align}
where $\pi(\lambda)$ is the probability distribution of the hidden variable $\lambda$ and $\lbrace v_A(a \vert x, \lambda) \rbrace$ and $\lbrace v_B(b \vert y, \lambda) \rbrace$ are deterministic probability distributions of Alice and Bob, respectively. In this case we denote $\mathbf{q} \in \mathrm{LHV}$ and we say $\mathbf{q}$ is nonlocal otherwise. To quantify the nonlocality, we use the distance-based resource monotone for nonlocality introduced by Brito et al.~\cite{PhysRevA.97.022111} as
\begin{align}
\mathrm{N}(\mathbf{q}_\mathbf{p}) = \dfrac{1}{2} \min\limits_{\mathbf{t} \in \mathrm{LHV}} \sum\limits_{a,b,x,y} p(x,y) \lvert q(a,b \vert x,y) - t(a,b \vert x,y) \rvert.   \label{nonlocality_distance}
\end{align}
\indent Having introduced all these different notions of quantum resources, we can complete our goal to establish relations among them.  
\begin{theorem}
\label{thrm_hierarchy}
Let $\mathcal{M}_{\mathbf{p}_A}$, $\mathcal{N}_{\mathbf{p}_B}$ be any \acp{WMA} and $\rho$ any bipartite quantum state of appropriate dimensions. Let $\Vec{\sigma}_{\mathbf{p}_A}$ be a state assemblage obtained via $\sigma_{a \vert x} = \mathrm{Tr}_1[(M_{a \vert x} \otimes \mathds{1}) \rho]$ and let $\mathbf{q}_{\mathbf{p}} = (\mathbf{q},\mathbf{p})$ be a probability distribution obtained via $q(a,b\vert x,y) = \mathrm{Tr}[N_{b \vert y} \sigma_{a \vert x} ]$ and $p(x,y) = p_A(x)p_B(y)$. The following sequence of inequalities holds:
\begin{align}
\IFdiamond(\mathcal{M}_{\mathbf{p}_A}) \geq \Cdiamond(\mathcal{M}_{\mathbf{p}_A})  \geq  \Idiamond(\mathcal{M}_{\mathbf{p}_A})  \geq \mathrm{S}(\Vec{\sigma}_{\mathbf{p}_A}) \geq \mathrm{N}(\mathbf{q}_{\mathbf{p}}).   \label{full_hierarchy} 
\end{align}
\end{theorem}
\begin{proof}
The inequalities $\IFdiamond(\mathcal{M}_{\mathbf{p}_A}) \geq \Cdiamond(\mathcal{M}_{\mathbf{p}_A}) \geq \Idiamond(\mathcal{M}_{\mathbf{p}_A})$ follow from the nested structure of the sets of free assemblages. More formally, $\mathscr{F}_{\mathrm{UI}} \subset \mathscr{F}_{\mathrm{IC}} \subset \mathscr{F}_{\mathrm{JM}} $ which can be seen by realizing that \ac{POVM} effects that are proportional to the identity are also incoherent (in any basis) and as incoherent \acp{POVM} commute pairwise, they are jointly measurable~\cite{PhysRevLett.126.220404}. Since we are minimizing the distance with respect to these sets, the inequalities hold. To prove that $\Idiamond(\mathcal{M}_{\mathbf{p}_A}) \geq \mathrm{S}(\Vec{\sigma}_{\mathbf{p}_A}) $ holds, we
use that incompatibility is necessary for steering. This allows us to use  $\Vec{\tau} = \lbrace \tau_{a,x} = \mathrm{Tr}_1[(F^*_{a \vert x} \otimes \mathds{1}) \rho] \rbrace$ as an unsteerable assemblage for any state $\rho$, as the closest \ac{JM} measurements $\mathcal{F}^*$ (with respect to the assemblage $\mathcal{M}$) cannot lead to steerable assemblages. It follows,
\begin{align}
\mathrm{S}(\Vec{\sigma}_{\mathbf{p}_A}) &\leq \dfrac{1}{2} \sum_x p_A(x) \sum_a \Vert \mathrm{Tr}_1[(M_{a \vert x} \otimes \mathds{1}) \rho] - \mathrm{Tr}_1[(F^*_{a \vert x} \otimes \mathds{1}) \rho] \Vert_1 \\
&\leq \dfrac{1}{2} \sum_x p_A(x) \max\limits_{\rho} \sum_a \Vert \mathrm{Tr}_1[((M_{a \vert x} - F^*_{a \vert x})  \otimes \mathds{1}) \rho] \Vert_1  \nonumber \\
&= \Idiamond(\mathcal{M}_{\mathbf{p}_A}), \nonumber
\end{align}
where we used the representation of $\Idiamond(\mathcal{M}_{\mathbf{p}_A})$ according to Eq.~\eqref{steering_interpretation_incompatibility} in the last line. We employ a similar approach to show that $\mathrm{S}(\Vec{\sigma}_{\mathbf{p}_A}) \geq \mathrm{N}(\mathbf{q}_{\mathbf{p}})$ in Lemma~\ref{Lem_steering_vs_nonlocality} in Appendix~\ref{Append_steering_bound_nonlocality}.
\end{proof}
Note that hierarchies related to that in Eq.\,\eqref{full_hierarchy} have also been established, at least partly, for weight- and robustness-based resource quantifiers~\cite{PhysRevLett.126.220404,PhysRevA.93.052112}. The connection between incompatibility, steering, and nonlocality has been studied by Cavalcanti et al.~\cite{PhysRevA.93.052112} extensively for for weight-\, and robustness-based quantifiers while Designolle et al.~\cite{PhysRevLett.126.220404} discussed the relation between coherence and incompatibility and, for a single \ac{POVM}, between the informativeness and the coherence in terms of the generalized robustness. 

The hierarchy\,\eqref{full_hierarchy} in Theorem~\ref{thrm_hierarchy} gives insights how resources like the incompatibility limit steering and nonlocal correlations quantitatively. On the other hand, every detection of these quantum correlations gives a lower bound to the measurement resources. In particular, the violation of every appropriately normalized steering or Bell inequality, in the nonlocal game formulation~\cite{Cleve,Arajo2020}, can lower bound these measurement resources. We show in Appendix~\ref{Append_steering_nonlocality_dual} that $\mathrm{S}(\Vec{\sigma}_{\mathbf{p}_A})$ is the maximal possible steering inequality violation given by
\begin{align}
\mathrm{S}(\Vec{\sigma}_{\mathbf{p}_A}) = \underset{G_{a \vert x},\ell}{\max} \sum_{a,x} p_A(x) \mathrm{Tr}[\sigma_{a \vert x} G_{a \vert x}] - \ell , \label{Steering_inequality}
\end{align}
where $\ell = \underset{\Vec{\tau} \in \mathrm{LHS}}{\max} \sum\limits_{a,x} p_A(x) \mathrm{Tr}[\tau_{a \vert x} G_{a \vert x}]$ is the classical bound obeyed by all unsteerable assemblages $\Vec{\tau} \in \mathrm{LHS}$ and the $G_{a \vert x}$ are positive semidefinite matrices s.t. $\lVert G_{a \vert x} \rVert_{\infty} \leq 1$, where $\lVert \,\cdot\, \rVert_{\infty}$ is the spectral norm. \\
\indent Moreover, the nonlocality $\mathrm{N}(\mathbf{q}_{\mathbf{p}})$ can be reformulated as the violation of a Bell inequality given by 
\begin{align}
\mathrm{N}(\mathbf{q}_{\mathbf{p}}) = \underset{C_{ab \vert xy},\ell}{\max}\sum_{a,b,x,y} p(x,y) C_{ab\vert xy} q(a,b \vert x,y) - \ell , \label{Bell_inequality}
\end{align}
where $\ell = \underset{\mathbf{t} \in \mathrm{LHV}}{\max} \sum\limits_{a,b,x,y} p(x,y) C_{ab\vert xy} t(a,b \vert x,y)$ is the local bound obeyed by all local correlations $\mathbf{t} \in \mathrm{LHV}$ and $C_{ab\vert xy}$ are Bell coefficients s.t. $0 \leq C_{ab\vert xy} \leq 1 $. \\
\indent It is worth to highlight that the hierarchy~\eqref{full_hierarchy} is reminiscent of the resource hierarchy for quantum states formulated by Streltsov et al.~\cite{Purity_resource}. For quantum states, it holds that
\begin{align}
\mathrm{P}(\rho) \geq \mathrm{C}(\rho) \geq \mathrm{D}(\rho) \geq \mathrm{E}(\rho),    \label{state_hierarchy}
\end{align}
where $\mathrm{P}(\rho)$, $\mathrm{C}(\rho)$, $\mathrm{D}(\rho)$, and $\mathrm{E}(\rho)$ denote the quantum state's purity, coherence with respect to product bases, discord, and entanglement, respectively, using the same geometric quantifier. Comparing both hierarchies, it becomes clear that the informativeness of measurements is in some sense the analogue to a state's purity, as both quantify the deviation from their respective uninformative element. We also observe that coherence is an important resource for states as well as measurements, which allows for more complex phenomena such as entanglement and incompatibility. Incompatibility and entanglement both play a similar role in their respective hierarchies, as both are the smallest known resource that is necessary for steering and nonlocality. Interestingly, incompatibility and entanglement also share similarities in their respective resource breaking maps~\cite{Heinosaari2015}. 

Moreover, we show in Appendix \ref{Append_entanglement_bound_steering} that the entanglement $\ETD(\rho)$ as defined in Eq.\,\eqref{entanglement} also upper bounds the steerability $\mathrm{S}(\Vec{\sigma}_{\mathbf{p}_A}) \leq \ETD(\rho)$. This leads to the conclusion that the nonlocality $\mathrm{N}(\mathbf{q}_\mathbf{p})$ and the steerability $\mathrm{S}(\Vec{\sigma}_{\mathbf{p}_A})$ are upper bounded by the smallest of the used resources to obtain $\mathbf{q}_\mathbf{p}$, respectively $\Vec{\sigma}_{\mathbf{p}_A}$.
\begin{Corollary}
Let $\mathcal{M}_{\mathbf{p}_A}$, $\mathcal{N}_{\mathbf{p}_B}$ be any \acp{WMA} and $\rho$ any bipartite quantum state of appropriate dimensions. Let $\Vec{\sigma}_{\mathbf{p}_A}$ be a state assemblage obtained via $\sigma_{a \vert x} = \mathrm{Tr}_1[(M_{a \vert x} \otimes \mathds{1}) \rho]$ and let $\mathbf{q}_{\mathbf{p}} = (\mathbf{q},\mathbf{p})$ be a probability distribution obtained via $q(a,b\vert x,y) = \mathrm{Tr}[N_{b \vert y} \sigma_{a \vert x} ]$ and $p(x,y) = p_A(x)p_B(y)$. The following inequalities hold
\begin{align}
\mathrm{N}(\mathbf{q}_{\mathbf{p}})  &\leq  \min \lbrace \ETD(\rho), \Idiamond(\mathcal{M}_{\mathbf{p}_A}), \Idiamond(\mathcal{N}_{\mathbf{p}_B})  \rbrace, \\
\mathrm{S}(\Vec{\sigma}_{\mathbf{p}_A}) &\leq  \min \lbrace \ETD(\rho), \Idiamond(\mathcal{M}_{\mathbf{p}_A})   \rbrace.
\end{align}  
\end{Corollary}
\indent An example illustrating the hierarchy~\eqref{full_hierarchy} is given by considering the respective resources of the \ac{CGLMP} measurements~\cite{PhysRevLett.88.040404,PhysRevLett.97.170409} applied to the maximally entangled state $\vert \Phi^+ \rangle = \dfrac{1}{\sqrt{d}} \sum_{i = 0}^{d-1} \lvert ii \rangle$. In the \ac{CGLMP} scenario, Alice and Bob perform two equally likely projective measurements in dimension $d$, given by $\lbrace M_{a \vert x} = \lvert a_x \rangle \langle a_x \rvert \rbrace, \lbrace N_{b \vert y} =  \lvert b_y \rangle \langle b_y \rvert \rbrace$, where \begin{align}
\lvert a_x \rangle = \dfrac{1}{\sqrt{d}} \sum_{q = 0}^{d-1} \exp{[\dfrac{2\pi i}{d}q(a-\alpha_x)]} \lvert q \rangle,  \label{CGLMP_Alice}  
\end{align}
for Alice's measurements and
\begin{align}
\lvert b_y \rangle = \dfrac{1}{\sqrt{d}} \sum_{q = 0}^{d-1} \exp{[-\dfrac{2\pi i}{d}q(b-\beta_y)]} \lvert q \rangle,    \label{CGLMP_Bob}
\end{align}
for Bob's measurements, with $\alpha_x = (x-1/2)/2$, $\beta_y = y/2$, and $a,b = 0, \cdots, d-1$ for $x,y = 1,2$. We visualize our results in Figure~\ref{CGLMP_Plot}.
\begin{figure}
\centering
\includegraphics[scale = 0.35]{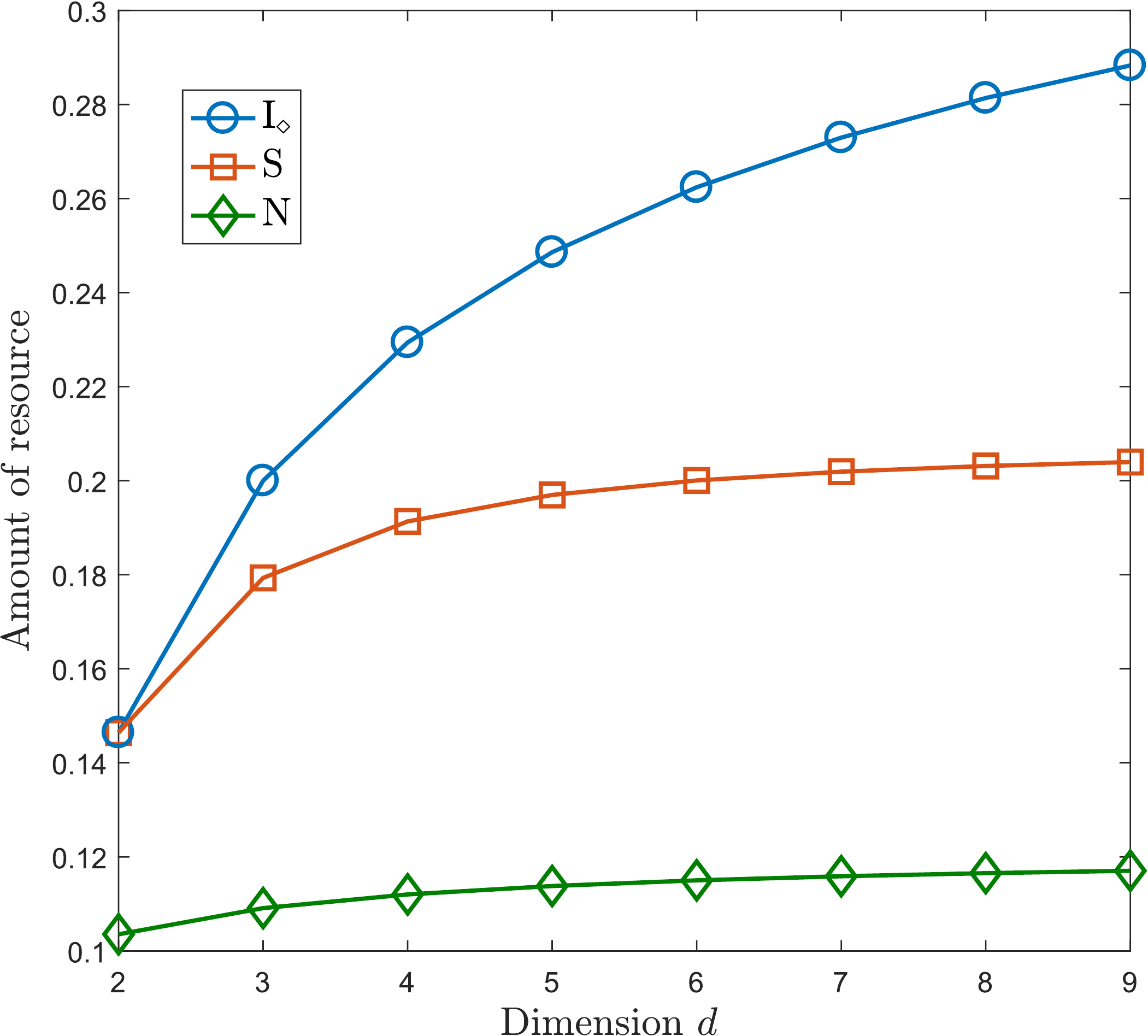}
 \caption{Comparison of the measurement resources. The incompatibility $\Idiamond$, steerability $\mathrm{S}$ (both from Alice), and nonlocality $\mathrm{N}$ are shown here for the \ac{CGLMP} measurements, for different dimensions $d$ and $m=2$ settings. The informativeness $\IFdiamond$ and the coherence $\Cdiamond$, which are not shown here, coincide in this particular case. More specifically, $\IFdiamond = \Cdiamond = 1-\dfrac{1}{d}$ as we show in section \ref{section_tightness}. The hierarchy $\IFdiamond \geq \Cdiamond \geq \Idiamond \geq \mathrm{S} \geq \mathrm{N}$ is clearly obeyed. While the nonlocality and the steerability converge quickly for growing $d$, the incompatibility increases further. The numerical methods used to obtain the results are explained in section~\ref{section_SDP}.}
 \label{CGLMP_Plot}
\end{figure}

\section{\texorpdfstring{\ac{SDP}}{SDP} formulations}
\label{section_SDP}
To study the hierarchy from Theorem~\ref{thrm_hierarchy} and the resources in more detail, an efficient method to numerically compute the respective resource quantifiers is needed. This can be done by formulating the quantifiers in terms of an \ac{SDP}, which also allows us to study the quantifiers analytically by exploiting duality theory. The computation of the general quantifier $\Rdiamond(\mathcal{M}_\mathbf{p})$ from Eq.~\eqref{diamond_monotone} can be stated as the following optimization problem:
\begin{align}
 &\underline{\text{Primal problem (general):}} \label{SDP_general_primal} \\   
&\mathrm{given:} \ \ \ \mathcal{M}_\mathbf{p} \nonumber \\ 
&\underset{Z_x, \mathcal{F}}{\mathrm{minimize}} \ \sum_x p(x) \lVert \mathrm{Tr}_1[Z_x] \rVert_{\infty} \nonumber \\
&\text{subject to:} \nonumber \\
&Z_x \geq J(M_x) - J(F_x), \ Z_x \geq 0 \ \forall \ x, \ \mathcal{F} \in \mathscr{F}, \nonumber 
\end{align}
where the $Z_x$ are positive semidefinite matrices, $J(M_x)$ is the Choi–Jamio\l kowski-matrix (see Eq.~\eqref{Choi–Jamiołkowski_state}) associated to setting $x$ of the assemblage $\mathcal{M}$, and $\mathcal{F}$ are the elements of the set of free assemblages $\mathscr{F}$. The formulation of the optimization in Eq.~\eqref{SDP_general_primal} mainly relies on the \ac{SDP} formulation of the diamond distance due to Watrous~\cite{0901.4709}. 

This compact representation of $\Rdiamond(\mathcal{M}_\mathbf{p})$ can be brought into an explicit \ac{SDP} formulation whenever the set $\mathscr{F}$ admits an \ac{SDP} formulation as we show in Appendix~\ref{Append_explicit_SDP} for the resources considered in this work. Every \ac{SDP} comes with a dual formulation which under some mild conditions (Slater's condition see, e.g. \cite{boyd_vandenberghe_2004}) returns the same optimal value as the primal problem. This condition is always satisfied for the \ac{SDP} \eqref{SDP_general_primal}. Hence, $\Rdiamond(\mathcal{M}_\mathbf{p})$ can also be written as optimal value of the optimization problem:
\begin{align}
&\underline{\text{Dual problem (general):}} \label{SDP_general_dual} \\    
&\mathrm{given:} \ \ \ \mathcal{M}_\mathbf{p} \nonumber \\ 
&\underset{C_{a \vert x}, \rho_x, \mathcal{F}}{\mathrm{maximize}} \ \sum_{a,x} p(x) \mathrm{Tr}[M_{a \vert x} C_{a \vert x}] - \underset{\mathcal{F} \in \mathscr{F}}{\max} \sum_{a,x} p(x) \mathrm{Tr}[F_{a \vert x} C_{a \vert x}] \nonumber \\
&\text{subject to:} \nonumber \\
&0 \leq C_{a \vert x} \leq \rho_x \ \forall \ a,x, \ \rho_x \geq 0, \mathrm{Tr}[\rho_x] = 1 \ \forall \ x, \nonumber
\end{align}
where the $C_{a \vert x}$, $\rho_x$ are positive semidefinite matrices and $\mathcal{F}$ are the elements of the set of free assemblages $\mathscr{F}$. Note that the dual formulation in Eq.~\eqref{SDP_general_dual} is in direct correspondence to the steering and Bell inequality formulations in Eq.~\eqref{Steering_inequality} and Eq.~\eqref{Bell_inequality}, as we maximize the difference of the resource value and the classical bound. The matrices $C_{a\vert x}$ describe a hyperplane in the assemblage space, while the states $\rho_x$ fix the scale (i.e. $\Rdiamond(\mathcal{M}_\mathbf{p}) \leq 1$ for any $\mathcal{M}_\mathbf{p}$) of the dual program. \\
\indent Since $\Rdiamond(\mathcal{M}_\mathbf{p})$ can be formulated as an \ac{SDP}, it is efficiently computable (in the Hilbert space dimension $d$) and one can resort to standard toolboxes for its computation~\cite{cvx,gb08,sdpt3,mosek}. We want to remark that it is also possible to use a variation of the \ac{SDP}~\eqref{SDP_general_dual} to obtain the optimal setting distribution $\mathbf{p}$ instead of fixing one in advance. This can be seen by introducing $C'_{a \vert x} = p(x) C_{a \vert x}$ and adjusting the constraints accordingly. See Appendix \ref{section_optimal_input_distribution} for an example where optimizing over $\mathbf{p}$ leads to an advantage over the uniform distribution for the incompatibility $\Idiamond(\mathcal{M}_\mathbf{p})$, even when only two measurement settings are considered. 

Even though \acp{SDP} are mainly used for numerical optimization, the underlying structure of an \ac{SDP} also offers a method to obtain analytical upper and lower bounds or even exact analytical expressions for $\Rdiamond(\mathcal{M}_\mathbf{p})$ depending on the complexity of the considered resource. More precisely, every feasible (but possibly sub-optimal) solution of the primal problem corresponds to an upper bound on $\Rdiamond(\mathcal{M}_\mathbf{p})$, while every feasible solution of the dual problem results in a lower bound. If we find feasible solutions of the primal and dual that result in the same value, we can conclude that this value is exactly $\Rdiamond(\mathcal{M}_\mathbf{p})$. We make use of this approach to derive bounds on the 
incompatibility $\Idiamond(\mathcal{M}_\mathbf{p})$ for any assemblage $\mathcal{M}$ weighted with a uniform distribution $\mathbf{p}$ in Theorem~\ref{thrm_bounds_incomp} and to identify cases in which the hierarchy in Theorem~\ref{thrm_hierarchy} is tight in section~\ref{section_tightness}.
\begin{theorem}
 \label{thrm_bounds_incomp}
 Given any \ac{WMA} $\mathcal{M}_\mathbf{p}$ consisting of $m$ \acp{POVM} in dimension $d$, with uniformly distributed measurement settings, i.e., $p(x) = 1/m \ \forall \ x$. The incompatibility $\Idiamond(\mathcal{M}_\mathbf{p})$ is upper and lower bounded by
 \begin{subequations}
  \begin{align}
\Idiamond(\mathcal{M}_\mathbf{p}) \geq &\dfrac{1}{md} \sum_{a,x} \mathrm{Tr}[M_{a \vert x}^2]-\dfrac{1}{m} \Big (\max\limits_{a,x} \lVert M_{a \vert x} \rVert_{\infty} + (m-1) \max_{a,a',x,x' \neq x} \lVert M_{a \vert x}^{1/2}  M_{a' \vert x'}^{1/2} \rVert_{\infty}\Big),   \\
\Idiamond(\mathcal{M}_\mathbf{p}) \leq &\dfrac{m-1}{(d+1)m^2} \sum_x \lVert d \mathds{1} - \sum_a \mathrm{Tr}[M_{a \vert x}] M_{a \vert x}  \rVert_{\infty}. 
\end{align}
 \end{subequations}
\end{theorem}
\begin{proof}
The proof relies on finding feasible solutions of the primal (upper bound) and dual problem (lower bound) in Eq.\,\eqref{SDP_general_primal} and Eq.\,\eqref{SDP_general_dual} for the specific set of \ac{JM} measurements $\mathscr{F}_{\mathrm{JM}}$. For the primal, we choose 
\begin{align}
Z_x = (1-\eta) \sum_a \lvert a \rangle \langle a \rvert \otimes \dfrac{d-\mathrm{Tr}[M_{a \vert x}]}{d}M_{a \vert x}^T,    
\end{align}
 where $\eta \in [0,1]$ is the largest number such that $\mathcal{F}$ obtained from 
\begin{align}
F_{a \vert x} = \eta M_{a \vert x} + (1- \eta) \mathrm{Tr}[M_{a \vert x}] \dfrac{\mathds{1}}{d}    \label{depolarizing_robustness}
\end{align}
is \ac{JM}. The coefficient $\eta$ is known as the depolarizing robustness of the assemblage $\mathcal{M}$.
Now by design, $\mathcal{F}$ is \ac{JM} and $Z_x \geq 0$. The remaining constraint, $Z_x \geq  \sum_a \lvert a \rangle \langle a \rvert \otimes (M_{a \vert x} - F_{a \vert x})^T$, can be verified by direct computation. It follows that 
$\Idiamond(\mathcal{M}_\mathbf{p}) \leq \dfrac{1-\eta}{md} \sum_x \lVert d \mathds{1} - \sum_a \mathrm{Tr}[M_{a \vert x} ]M_{a \vert x} \rVert_{\infty}$,
where we used that $\mathbf{p}$ is uniformly distributed. Finally, the upper bound follows from \cite{Designolle2019}, where it was found that $\eta^{\text{low}} = \dfrac{1}{m}(1+\dfrac{m-1}{d+1})$ is a lower bound to the depolarizing robustness and therefore always leads to jointly measurable measurements for general measurement assemblages $\mathcal{M}$ with $m$ measurements of dimension $d$. \\
\indent To obtain the lower bound from the dual problem, we rewrite the objective function as \\
$\sum_{a,x} p(x) \mathrm{Tr}[M_{a \vert x} C_{a \vert x}] -\mathrm{Tr}[L]$, where $L$ is a matrix such that $L \geq \sum_{a,x} p(x)  v(a \vert x, \lambda) C_{a \vert x} \ \forall \ \lambda$. Note that such an $L$ always exists, which can be verified by multiplying both sides of the inequality with the \ac{POVM} effect $G_{\lambda}$ before summing over all $\lambda$ and taking the trace. We choose as feasible solution $C_{a \vert x} = \dfrac{M_{a \vert x}}{d}$, $\rho_x = \sum_a C_{a \vert x} = \dfrac{\mathds{1}}{d}$, and $L = l \mathds{1}$ with some free parameter $l$. Clearly, in this way all constraints are satisfied for some appropriately chosen parameter $l$ which still needs to be determined. We obtain that the incompatibility is lower bounded such that $\dfrac{1}{md} \sum_{a,x} \mathrm{Tr}[M_{a \vert x}^2]-dl \leq \Idiamond(\mathcal{M}_\mathbf{p})$. The constraint to find $l$ is now given by 
$l \mathds{1} \geq \sum_{a,x} \dfrac{1}{md}  v(a \vert x, \lambda) M_{a \vert x} \ \forall \ \lambda$, which means $l$ is the spectral norm  $ \lVert \dfrac{1}{md}  \sum_{a,x} v(a \vert x, \lambda) M_{a \vert x} \rVert_{\infty}$ maximized over the deterministic distributions $ \lbrace v(a \vert x,\lambda) \rbrace$. 
This means to find a valid $l$, we need to find an $l$ such that 
\begin{align}
l \geq \dfrac{T}{md}, \label{tighter_bounds_incomp}
\end{align}
where $T \coloneqq \lVert \sum_{a,x} v^*(a \vert x,\lambda) M_{a \vert x} \rVert_{\infty}$ and $ \lbrace v^*(a \vert x,\lambda) \rbrace_{a,x}$ are the deterministic probability distributions that maximize the right-hand side of Eq~\eqref{tighter_bounds_incomp}, respectively the spectral norm in the definition of $T$. \\
\indent Using the results in\,\cite{10.2307/24715730} it follows that $l = \dfrac{1}{md}(\max\limits_{a,x} \lVert M_{a \vert x} \rVert_{\infty}  + (m-1) \max\limits_{x,x' \neq x, a,a'} \lVert M_{a \vert x}^{1/2}  M_{a' \vert x'}^{1/2} \rVert_{\infty})$ is a valid choice, from which the lower bound on $\Idiamond(\mathcal{M}_\mathbf{p})$ follows. 
\end{proof}
While the bounds in Theorem~\ref{thrm_bounds_incomp} look complicated, we highlight that they become much simpler in the case of rank-$1$ projective measurements and especially for measurements based on \ac{MUB}.
Two orthonormal bases $\lbrace \vert v_a \rangle \rbrace_{0 \leq a \leq d-1}$ and $\lbrace \vert w_b \rangle \rbrace_{0 \leq b \leq d-1}$ are MUB if  
\begin{align}
\vert \langle v_a \vert w_b \rangle \vert = \dfrac{1}{\sqrt{d}} \ \forall \ a,b. \label{Overlap_MUB}
\end{align}
The set of projectors onto the vectors of a basis form a measurement $ \mathcal{M} = \lbrace M_a = \lvert v_a \rangle \langle v_a \rvert \rbrace$. An \ac{MUB} measurement assemblage is a set of measurements where the condition~\eqref{Overlap_MUB} holds for any two projections from different bases. \ac{MUB} measurement assemblages find many applications in quantum information \cite{DURT2010} and are natural candidates for highly incompatible measurements as studied in \cite{PhysRevLett.122.050402,Designolle2019,PhysRevA.96.022110}. It is known that in every dimension $d \geq 2$ there exist at least $m=p^r+1$ \ac{MUB}, where $p^r$ is the smallest prime power factor of $d$ \cite{10.1007/978-3-540-24633-6_10}. While it is in general an open problem how many \ac{MUB} really exist in a given dimension $d$, explicit constructions for $m=d+1$ \ac{MUB} are known when $d$ is a prime-power. \\
\indent Regarding the simplifications for rank-$1$ projective measurements we observe from the proof of Theorem~\ref{thrm_bounds_incomp} that it holds
\begin{align}
\dfrac{1}{md} \sum_{a,x} \mathrm{Tr}[M_{a \vert x}^2]-\dfrac{T}{m} \leq \Idiamond(\mathcal{M}_\mathbf{p}) \leq \dfrac{1-\eta}{md} \sum_x \lVert d \mathds{1} - \sum_a \mathrm{Tr}[M_{a \vert x} ]M_{a \vert x} \rVert_{\infty},    
\end{align}
where $T \coloneqq \lVert \sum_{a,x} v^*(a \vert x,\lambda) M_{a \vert x} \rVert_{\infty}$ and $\eta \in [0,1]$ is the largest number such that the assemblage $\mathcal{F}$ containing the measurements $F_{a \vert x} = \eta M_{a \vert x} + (1- \eta) \mathrm{Tr}[M_{a \vert x}] \dfrac{\mathds{1}}{d}$ is \ac{JM}. Due to the fact that $\mathrm{Tr}[M_{a \vert x}] = \mathrm{Tr}[M_{a \vert x}^{2}] = 1$ for rank-$1$ projections, we obtain the following corollary. 
 \begin{Corollary}
 \label{Corollary_projective_measurements}
 The incompatibility $\Idiamond(\mathcal{M}_\mathbf{p})$ of any measurement assemblage $\mathcal{M}$ consisting of $m$ rank-$1$ projective measurements, weighted with a uniformly distributed $p(x) = \dfrac{1}{m} \ \forall \ x$, is bounded as
 \begin{align}
\label{Bounds_incompatibility_projective}    
1 - \dfrac{T}{m} \leq \Idiamond(\mathcal{M}_\mathbf{p}) \leq (1-\eta) \dfrac{d-1}{d},
\end{align}
with $T \coloneqq \lVert \sum_{a,x} v^*(a \vert x,\lambda) M_{a \vert x} \rVert_{\infty}$ and the depolarizing robustness $\eta$ defined via Eq.~\eqref{depolarizing_robustness}.
\end{Corollary}
Using the overlap relation in Eq.\,\eqref{Overlap_MUB} and the same lower bound on $\eta$ as in Theorem~\ref{thrm_bounds_incomp}, it follows for a uniformly weighted \ac{MUB} measurement assemblage that 
\begin{align}
\label{Bounds_incompatibility_MUB}    
1 - \dfrac{1}{m}\big(1+\dfrac{(m-1)}{\sqrt{d}}\big) \leq \Idiamond(\mathcal{M}_\mathbf{p}) \leq \dfrac{(d-1)(m-1)}{(d+1)m}.
\end{align}
Some asymptotic behaviours for the incompatibility of \ac{MUB} measurement assemblages can be observed. In the case of large dimensions $d$ for a fixed number of measurements, the incompatibility approaches $\Idiamond(\mathcal{M}_\mathbf{p}) \approx 1 - \dfrac{1}{m}$ (as the upper and lower bound collapse onto each other) i.e. it asymptotically approaches the value $1$ for large $d$ and $m$. To get an impression of the quality of the bounds in Eq.~\eqref{Bounds_incompatibility_MUB} we investigate a specific construction of \ac{MUB} in prime dimensions $d$ based on the Heisenberg-Weyl operators
\begin{align}
\hat{X} = \sum_{k=0}^{d-1} \lvert k+1 \rangle \langle k \rvert, \ \hat{Z} = \sum_{k=0}^{d-1} \omega^k \lvert k \rangle \langle k \rvert,    
\end{align}
where $\lbrace \lvert k \rangle \rbrace_{0 \leq k \leq d-1}$ is the computational basis and $\omega = \exp{\big(\dfrac{2\pi i}{d}\big)}$ is a root of unity. In prime dimensions $d$, the eigenbases of the $d+1$ operators $\hat{X},\hat{Z},\hat{X} \hat{Z}, \hat{X} \hat{Z}^2, \cdots, \hat{X}\hat{Z}^{d-1}$ are mutually unbiased \cite{bandyopadhyay2002new}. We use these eigenbases to form sets of projective POVMs. Note that it matters which subset of eigenbases we choose. For example, the set of measurements associated with the eigenbases of $\mathcal{M}^{(1)} = \lbrace \hat{X}, \hat{Z}, \hat{X}\hat{Z} \rbrace$ can possibly have a different incompatibility than the measurements associated with the set $\mathcal{M}^{(2)} = \lbrace  \hat{X}\hat{Z}^{d-3},\hat{X}\hat{Z}^{d-2},\hat{X}\hat{Z}^{d-1} \rbrace.$ This is indeed the case for \ac{MUB} measurement assemblages in dimension $d=5$ and $m=3$ settings. We find that $\Idiamond(\mathcal{M}^{(1)}_\mathbf{p}) = 0.3750$, while $\Idiamond(\mathcal{M}^{(2)}_\mathbf{p}) = 0.3685$. This shows that different \ac{MUB} are operationally inequivalent, which has also been demonstrated for the depolarizing robustness \cite{PhysRevLett.122.050402}.
For the values in Table~\ref{Table1}, we used the assignment of \ac{MUB} according to the \ac{WMA} $\mathcal{M}^{(1)}_\mathbf{p}$, i.e., we take the first $m$ eigenbases. \\ 
\begin{table}
\centering
\scalebox{0.90}{
\begin{tabular}{| p{12mm} | p{12mm} | p{12mm} | p{12mm} | p{12mm} |} 
    $m$\textbackslash$d$ & $2$ & $3$& $5$ & $7$ \\ \hline
    $2$ & $0.1667$ \newline $\textcolor{blue}{0.1464}$ \newline $0.1464$ & $0.2500$ \newline $\textcolor{blue}{0.2113}$ \newline $0.2113$ & $0.3333$ \newline $\textcolor{blue}{0.2764}$ \newline $0.2764$ & $0.3750$ \newline $\textcolor{blue}{0.3110}$ \newline $0.3110$ \\ \hline 
    $3$ &$0.2222$ \newline $\textcolor{blue}{0.2113}$ \newline $0.1953$ & $0.3333$ \newline $\textcolor{blue}{0.2876}$  \newline $0.2818$ & $0.4444$ \newline $\textcolor{blue}{0.3750}$ \newline $0.3685$ & $0.5000$ \newline $\textcolor{blue}{0.4154}$ \newline $0.4147$ \\ \hline
    $4$ &  & $0.3750$ \newline $\textcolor{blue}{0.3455}$ \newline $0.3170$ & $0.5000$ \newline $\textcolor{blue}{0.4307}$  \newline $0.4146$ & $0.5625$ \newline $\textcolor{blue}{0.4724}$ \newline $0.4665$ \\ \hline
    $5$ &  &  & $0.5333$ \newline $\textcolor{blue}{0.4657}$ \newline $0.4422$ & $0.6000$ \newline $\textcolor{blue}{0.5040}$ \newline $0.4976$ \\ \hline
    $6$ &  &  & $0.5556$ \newline $\textcolor{blue}{0.4910}$ \newline $0.4607$ & $0.6250$ \newline $\textcolor{blue}{0.5257}$ \newline $0.5184$ \\ \hline
    $7$ &  &  &  & $0.6429$ \newline $\textcolor{blue}{0.5413}$ \newline $0.5332$ \\ \hline
    $8$ &  &  &  & $0.6563$ \newline $\textcolor{blue}{0.5728}$ \newline $0.5443$ \\ \hline
    \end{tabular}}
    \caption{\label{Table1} Incompatibility $\Idiamond(\mathcal{M}_\mathbf{p})$ of MUB in prime dimensions $d$ with $m$ settings. In each cell, the first number is the upper bound on the incompatibility, the second number is the actual incompatibility which can be computed via the \acp{SDP} \eqref{SDP_general_primal} and \eqref{SDP_general_dual}, marked in blue, and the third number is the lower bound on the incompatibility. The bounds are obtained from Eq. \eqref{Bounds_incompatibility_MUB}. Note that the lower bound is tight for $m =2$ measurements. Furthermore, it is shown in the text, that the incompatibilities for $m=2$, $m=d$, and $m=d+1$ measurements can be obtained analytically.}
\end{table} 
\indent As one can see in Table~\ref{Table1}, the upper and lower bounds combined give a good idea on how incompatible this implementation of MUB is in practical scenarios. The lower bound can be tightened significantly by using the bound from Corollary \ref{Corollary_projective_measurements} directly. Note that this requires an optimization over all $N_{\mathrm{det}} = o^m$ deterministic assignments $\lbrace v(a \vert x, \lambda) \rbrace$, where $o$ is the number of measurement outcomes for each of the $m$ settings. Surprisingly, the tightened lower bound coincides with the numerical values for the incompatibility $\Idiamond(\mathcal{M}_\mathbf{p})$ for all $m,d$ in Table~\ref{Table1} up to the fourth digit.
While we were not able to show that the lower bound from Corollary \ref{Corollary_projective_measurements} is tight for \ac{MUB} measurement assemblages in general, we are able to identify important cases where this is indeed the case. 

More specifically, it was shown by Designolle et al. \cite{PhysRevLett.122.050402} that $\eta = \dfrac{dT-m}{dm-m}$ is the depolarising robustness for the \textit{standard construction} of \ac{MUB} measurement assemblages in prime power dimensions given in \cite{Wootters1989} for $m=2$, $m=d$, and $m=d+1$ measurements. It is important to highlight that the construction used above, based on the Heisenberg-Weyl operators, is an equivalent reformulation of this construction for prime dimensions \cite{bandyopadhyay2002new}. From Eq.\,\eqref{Bounds_incompatibility_projective}, it follows directly that $\Idiamond(\mathcal{M}_\mathbf{p}) = 1-\dfrac{T}{m}$, since the upper and lower bound coincide for these cases. Note that while this result holds only for this special construction, it was conjectured \cite{PhysRevLett.122.050402} that $\eta = \dfrac{dT-m}{dm-m}$ holds for all constructions of \ac{MUB} and $m=2$, $m=d$, and $m=d+1$. 
Note further that the bounds in Eq.\,\eqref{Bounds_incompatibility_MUB} can also be used to study cases where it is not known whether \ac{MUB} exist. For instance, if there exists a set of $m=4$ \ac{MUB} in $d=6$, the \ac{WMA} needs to have an incompatibility in between $0.4438 \leq \Idiamond(\mathcal{M}_\mathbf{p}) \leq 0.5357$.  \\
\indent To conclude this section, we want to emphasize that analogous discussions to obtain bounds on the resource quantifier $\Rdiamond(\mathcal{M}_\mathbf{p})$ can be made for any \ac{QRT} with a free set $\mathcal{F}$ that can be described by \ac{SDP} constraints. For instance, we show in Appendix~\ref{Append_inform_and_coherence} that the informativeness $\IFdiamond(\mathcal{M}_\mathbf{p})$ of rank-$1$ projective measurements is given by $\IFdiamond(\mathcal{M}_\mathbf{p}) = 1 - \dfrac{1}{d}$ for any probability distribution $\mathbf{p}$. Note that since the set $\mathscr{F}_{\mathrm{UI}}$ of \ac{UI} assemblages (see Eq.\,\eqref{informativeness_free}) has a much simpler structure than the set of \ac{JM} measurements, it is also easier to obtain exact expressions. 

\section{Tightness of the Hierarchy}
\label{section_tightness}
It is particularly interesting to study the optimal conversion of one resource to another, i.e., to study for which measurements (and states) the bounds in Eq.\,\eqref{full_hierarchy} are tight. Obviously, for \ac{UI} measurements it holds $\IFdiamond(\mathcal{M}_\mathbf{p}) = 0$ and all bounds are trivially tight. We study nontrivial cases of resource equivalences where $\IFdiamond(\mathcal{M}_\mathbf{p}) = \Cdiamond(\mathcal{M}_\mathbf{p})$ and $\Idiamond(\mathcal{M}_\mathbf{p}) = \mathrm{S}(\Vec{\sigma}_{\mathbf{p}})$ holds. We start with the latter. \\
\indent Incompatibility and steerability are known to be deeply connected and equivalences have been reported for robustness and weight-based quantifiers \cite{PhysRevLett.115.230402,PhysRevA.93.052112}. We consider again the situation of uniformly distributed measurements, i.e., $p(x) = 1/m$. Let $\rho = \lvert \Phi^+ \rangle \langle \Phi^+ \rvert$ be the maximally entangled state, where $\lvert \Phi^+ \rangle = \dfrac{1}{\sqrt{d}}\sum_{i = 0}^{d-1} \vert ii \rangle$. It is readily verified that 
\begin{align}
\sigma_{a \vert x} = \mathrm{Tr}_1[(M_{a \vert x} \otimes \mathds{1}) \rho] = \dfrac{M_{a \vert x}^T}{d},   \label{transposed_assemlage} 
\end{align}
 where the transposition is with respect to the computational basis. Using the state assemblage $\Vec{\sigma} = \lbrace \sigma_{a \vert x} \rbrace$ obtained via Eq.\,\eqref{transposed_assemlage} is the standard approach to map incompatibility problems to steering problems and proves also to be useful here. In section~\ref{section_SDP}, we showed that for the construction of \ac{MUB} in \cite{Wootters1989} and $m=2$, $m=d$, $m=d+1$ measurements $\Idiamond(\mathcal{M}_\mathbf{p}) = 1-\dfrac{T}{m}$ holds,
 where $T = \lVert \sum_{a,x} v^*(a \vert x, \lambda) M_{a \vert x} \rVert_{\infty}$. It follows that $1-\dfrac{T}{m} \geq \mathrm{S}(\Vec{\sigma}_{\mathbf{p}})$. Using the state assemblage $\Vec{\sigma}$ obtained from Eq.\,\eqref{transposed_assemlage}, it is possible to show that this bound is indeed fulfilled. To show this, we employ the steering inequality formulation of  $\mathrm{S}(\Vec{\sigma}_{\mathbf{p}})$ as discussed in Eq.~\eqref{Steering_inequality}, which we repeat here for convenience:
\begin{align}
\mathrm{S}(\Vec{\sigma}_{\mathbf{p}}) = \underset{G_{a,x},\ell}{\max} \sum_{a,x} p(x) \mathrm{Tr}[\sigma_{a \vert x} G_{a \vert x}] - \ell, \end{align}
where $\ell = \underset{\Vec{\tau} \in \mathrm{\ac{LHS}}}{\max} \sum\limits_{a,x} p(x) \mathrm{Tr}[\tau_{a \vert x} G_{a \vert x}]$ is the classical bound obeyed by all unsteerable assemblages $\Vec{\tau} \in \mathrm{\ac{LHS}}$ and the $G_{a \vert x}$ are positive semidefinite matrices such that $\lVert  G_{a \vert x} \rVert_{\infty} \leq 1$. We want to emphasize that $\ell$ can equivalently be defined in a way that it has to satisfy $\ell \mathds{1} \geq  \sum_{a,x} p(x) v(a \vert x, \lambda) G_{a \vert x}$ for all $\lambda$. By multiplying both sides of the inequalities with the hidden states $\sigma_{\lambda}$ and taking the trace trace afterwards it follows $\ell \geq \underset{\Vec{\tau} \in \mathrm{LHS}}{\max} \sum\limits_{a,x} p(x) \mathrm{Tr}[\tau_{a \vert x} G_{a \vert x}]$ and the equality follows from the fact that we maximize over $\ell$. Now, by choosing $G_{a \vert x} = M_{a \vert x}^T$ and $\ell = \dfrac{T}{m}$, clearly all constraints are fulfilled and the steerability is lower bounded by $\mathrm{S}(\Vec{\sigma}_{\mathbf{p}}) \geq 1 -\dfrac{T}{m}$, which coincides with the upper bound. Therefore, it follows that $\Idiamond(\mathcal{M}_\mathbf{p})  = \mathrm{S}(\Vec{\sigma}_{\mathbf{p}})$. While this result is only valid for $m=2$, $m=d$, and $m=d+1$ and the special construction of \ac{MUB} in \cite{Wootters1989} we conjecture that the equivalence between incompatibility and steerability holds for general constructions of \ac{MUB} and $2 \leq m \leq d+1$. \\
\indent We searched numerically for other cases with an equality between incompatibility and steerability. However, apart from the case of generic qubit projective measurements we were not able to identify any other scenarios. Note that this finding deviates from the observations for consistent weight and robustness quantifiers studied by Cavalcanti et al.~\cite{PhysRevA.93.052112}, where an equivalence between incompatibility and steerability was found for all assemblages. This difference is not artificial, as it remains even if we include the consistency constraint for the steerability below Eq.~\eqref{Steering_quantifier_distance}. \\
 \indent The second equivalence of resources we want to discuss is that between the informativeness and the coherence of assemblages. More precisely, we discuss when $\IFdiamond(\mathcal{M}_\mathbf{p}) = \Cdiamond(\mathcal{M}_\mathbf{p})$ holds. Interestingly, this equivalence is achieved by \acp{WMA} $\mathcal{M}_\mathbf{p}$ that are mutually unbiased to the set of projective measurements onto the incoherent basis $\lbrace \vert i \rangle \langle i \rvert \rbrace$. To see this, we note first that $\IFdiamond(\mathcal{M}_\mathbf{p}) = 1-\dfrac{1}{d}$ holds for all rank-$1$ projective measurements as shown in section~\ref{section_SDP}. From there it follows that the coherence of \ac{MUB} is bounded by $1-\dfrac{1}{d} \geq \Cdiamond(\mathcal{M}_\mathbf{p})$. To show that this bound can be achieved, we use the dual formulation of $\Cdiamond(\mathcal{M}_\mathbf{p})$. More specifically, $\Cdiamond(\mathcal{M}_\mathbf{p})$ is given by
 \begin{align}
  \Cdiamond(\mathcal{M}_\mathbf{p}) = \underset{C_{a \vert x},\ell_{x,i}}{\mathrm{max}} \sum_{a,x} p(x) \mathrm{Tr}[M_{a \vert x} C_{a \vert x}] - \sum_{x,i} \ell_{x,i},  
 \end{align}
 where the $\ell_{x,i}$ are scalars such that $\ell_{x,i} \geq p(x) \mathrm{Tr}[C_{a \vert x} \lvert i \rangle \langle i \rvert] \ \forall \ a,x,i,$ and the $C_{a \vert x}$ are matrices such that  $0 \leq C_{a \vert x} \leq \rho_x \ \forall \ a,x$, where the $\rho_x$ are quantum states. The optimal solutions of the dual problem are $C_{a \vert x} = \dfrac{M_{a \vert x}}{d}$ and $\ell_{x,i} = \dfrac{p(x)}{d^2}$. Clearly, these choices are feasible, which can be verified by using that $\mathrm{Tr}[\vert i \rangle \langle i \vert M_{a \vert x}] = \dfrac{1}{d}$, due to the unbiasedness of $\mathcal{M}$ and the incoherent basis $\lbrace \vert i \rangle \langle i \rvert \rbrace$. Further, they are optimal since they lead to $\Cdiamond(\mathcal{M}_\mathbf{p}) \geq 1-\dfrac{1}{d}$, which coincides with the upper bound. \\
 \indent The fact that measurements that are mutually unbiased to the incoherent basis maximize the coherence is very similar to the situation for quantum states \cite{Purity_resource}. There, for a fixed spectrum, the coherence is maximized by states that have an eigendecomposition in a mutually unbiased basis with respect to the incoherent basis. Note that the measurements within $\mathcal{M}$ do not need to be \ac{MUB} measurement assemblages themselves, as long as they are mutually unbiased to the incoherent bases. Indeed, we show in Appendix~\ref{Append_inform_and_coherence} that the \ac{CGLMP} measurements defined via Eq.\,\eqref{CGLMP_Alice} and Eq.\,\eqref{CGLMP_Bob} also maximize the coherence in the sense that $\IFdiamond(\mathcal{M}_\mathbf{p}) = \Cdiamond(\mathcal{M}_\mathbf{p}) = 1-\dfrac{1}{d}$. Note that it is known that the maximal coherence of a single \ac{POVM} in terms of the generalized robustness can be achieved by measurements in the Fourier basis of the incoherent basis \cite{Oszmaniec2019} .  \\
 \indent Let us briefly comment on the other two inequalities of the hierarchy in Eq.~\eqref{full_hierarchy}. The remaining two inequalities are $\Cdiamond(\mathcal{M}_\mathbf{p}) \geq \Idiamond(\mathcal{M}_\mathbf{p})$ and $\mathrm{S}(\Vec{\sigma}_{\mathbf{p}_A}) \geq \mathrm{N}(\mathbf{q}_{\mathbf{p}})$.  While the relation between steering and nonlocality is notoriously hard to study, even for two-qubit states, the connection between coherence and incompatibility has only recently gained some attention \cite{PhysRevA.105.012205,PhysRevLett.126.220404}. Our numerical search suggests that both bounds are true inequalities in non-trivial scenarios. However, future research is needed to come to a conclusion. 

\section{Conclusion and Outlook}

Quantifying quantum advantages plays an important role in modern quantum information
theory, particularly in the framework of \acp{QRT}. The quantification of measurement resources has developed historically in a different direction than the quantification of state resources, as it was unclear how a distance-based approach can be applied to measurements and especially sets of measurements in a meaningful way. Instead, weight-based and especially robustness-based quantifiers were predominantly used for quantifying measurement resources so far. \\
\indent In the present work, we have solved this problem by introducing the general notion of distance-based resource quantification for sets of measurements. 
We have studied which prerequisites are necessary for a function to be a proper distance between measurement assemblages and showed that every such distance induces a resource monotone for any convex \ac{QRT}. We have proposed one particular quantifier, based on the diamond norm, with a clear operational meaning in terms of the optimal single-shot distinguishability of different measurement assemblages. \\
\indent On the basis of this particular quantifier, we have established a hierarchy of measurement resources in Theorem~\ref{thrm_hierarchy} and showed that recently introduced steering \cite{PhysRevA.97.022338} and nonlocality quantifiers \cite{PhysRevA.97.022111} fit naturally into this hierarchy. Furthermore, we have shown that our quantifier can be studied numerically and analytically in terms of \acp{SDP}. We have used this insight to establish analytical upper and lower bounds on the incompatibility of any measurement assemblage in Theorem~\ref{thrm_bounds_incomp}. 
Noteworthy, by focussing on rank-$1$ projective measurements, we have shown that the bounds on the incompatibility in Corollary~\ref{Corollary_projective_measurements} are tight for particular \ac{MUB} measurement assemblages, which play a special role in the established measurement hierarchy. More precisely, we showed in section~\ref{section_tightness} that the incompatibility of \ac{MUB} measurement assemblages attains the same value as the steerability of the state assemblages obtained from performing these measurements on one part of a maximally entangled state. Furthermore, we showed that measurements that are mutually unbiased to the incoherent basis maximize the coherence among all rank-$1$ projective measurements. \\
\indent
It would be interesting to see which insights can be obtained when distance-based quantifiers like the one presented here are studied for other resource theories like projective-simulability \cite{PhysRevLett.119.190501} or the \ac{QRT} of imaginarity \cite{PhysRevLett.126.090401} applied to measurements. 
Furthermore, distance-based quantifiers should also be compared to possible entropic resource quantifiers of measurement assemblages. 
So far, entropic quantifiers have only been considered very recently \cite{PhysRevA.106.022401} for a single \ac{POVM}. 
With the definition of a distance between measurement assemblages, it is also possible to study the continuity of functions of measurement resources, which could be of independent interest for robust self-testing \cite{Supic2020selftestingof} or measurement tomography \cite{PhysRevLett.83.3573}.

\begin{acknowledgments}
We thank Lennart Bittel, Federico Grasselli, and Gl\'aucia Murta for helpful discussions.
This research was partially supported
by the EU H2020 QuantERA ERA-NET Cofund in
Quantum Technologies project QuICHE, and by the Federal Ministry of Education and Research (BMBF). 
MK's work is supported by the Deutsche Forschungsgemeinschaft (DFG, German Research Foundation) -- project number 441423094. 
\\[1em]
\end{acknowledgments}

\appendix

\section{Proof of Theorem \ref{thrm1}} 
\label{Append_thrm1}
Here, we show that the function $\Ddiamond(\mathcal{M}_\mathbf{p}, \mathcal{N}_\mathbf{p})$ defined in Eq.~\eqref{assemblage_distance} is a jointly-convex distance function between the two \acp{WMA} $\mathcal{M}_\mathbf{p}$ and $\mathcal{N}_\mathbf{p}$. In the following, we use that a measure-and-prepare channel (see Eq.~\eqref{measure_prepare}) corresponding to setting $x$ of the assemblage $\mathcal{M}$ applied to the first subsystem of a bipartite state is given by 
\begin{align}
 (\Lambda_{\mathcal{M}_x} \otimes \mathds{1})(\rho) = \sum_a (\lvert a \rangle \langle a \rvert \otimes \mathrm{Tr}_1[(M_{a \vert x} \otimes \mathds{1}) \rho]).  
\end{align}
Furthermore, we introduce the quantities $\sigma_{a \vert x}(\rho) = \mathrm{Tr}_1[(M_{a \vert x} \otimes \mathds{1}) \rho]$ and $\tau_{a \vert x}(\rho) = \mathrm{Tr}_1[(N_{a \vert x} \otimes \mathds{1}) \rho]$.

\begin{theorem1}
The function $\Ddiamond(\mathcal{M}_\mathbf{p}, \mathcal{N}_\mathbf{p})$ is a distance function between the \acp{WMA} $\mathcal{M}_\mathbf{p}$ and $\mathcal{N}_\mathbf{p}$, i.e., it fulfils all the conditions stated in Definition~\ref{Def_Distance}. Moreover, $\Ddiamond(\mathcal{M}_\mathbf{p}, \mathcal{N}_\mathbf{p})$ is jointly-convex.
\end{theorem1}
\begin{proof}
We start by writing $\Ddiamond(\mathcal{M}_\mathbf{p}, \mathcal{N}_\mathbf{p})$ in a more convenient form.
More precisely, we use in the following that the triangle inequality for the trace norm $\lVert \cdot \rVert_1$ results in an equality due to the support on different subspaces of the terms with different $a$ within the sum over the outcomes $a$. Furthermore, we used the multiplicity of the trace norm under tensor products and the fact that $\lVert \lvert a \rangle \langle a \rvert \rVert_1 = 1 \ \forall \ a$. It follows that
\begin{widetext}
\begin{align}
\Ddiamond(\mathcal{M}_\mathbf{p},\mathcal{N}_\mathbf{p}) \label{Append_steering_interpretation} &=\dfrac{1}{2} \sum_{x} p(x) \max\limits_{\rho} \lVert \sum_a \lvert a \rangle \langle a \rvert \otimes [\sigma_{a \vert x}(\rho) - \tau_{a \vert x}(\rho)] \rVert_1  \\ 
&=\dfrac{1}{2} \sum_{x} p(x) \max\limits_{\rho} \sum_a \lVert  \lvert a \rangle \langle a \rvert \otimes [\sigma_{a \vert x}(\rho) - \tau_{a \vert x}(\rho)] \rVert_1 \nonumber \\
&= \dfrac{1}{2} \sum_{x} p(x) \max\limits_{\rho} \sum_a \lVert  \sigma_{a \vert x}(\rho) - \tau_{a \vert x}(\rho) \rVert_1. \nonumber
\end{align} 
We use the form of $\Ddiamond(\mathcal{M}_\mathbf{p},\mathcal{N}_\mathbf{p})$ according to Eq.~\eqref{Append_steering_interpretation} in the following. \\
\indent Clearly $\Ddiamond(\mathcal{M}_\mathbf{p}, \mathcal{N}_\mathbf{p})$ is a non-negative function with $\Ddiamond(\mathcal{M}_\mathbf{p}, \mathcal{N}_\mathbf{p}) = 0$ if and only if $\mathcal{M} = \mathcal{N}$. The triangle inequality is obeyed due to the linearity of $\Ddiamond(\mathcal{M}_\mathbf{p}, \mathcal{N}_\mathbf{p})$ in the trace norm. The monotonicity under quantum channel $\Lambda^{\dagger}$ follows from direct calculation,
\begin{align}
\Ddiamond(\mathcal{M}_\mathbf{p}, \mathcal{N}_\mathbf{p}) 
&= \dfrac{1}{2} \sum_x p(x) \max\limits_{\rho} \sum_a \lVert  \mathrm{Tr}_1[((M_{a \vert x}-N_{a \vert x}) \otimes \mathds{1}) \rho]  \rVert_1  \\ 
&\geq \dfrac{1}{2} \sum_x p(x) \max\limits_{\rho'} \sum_a \lVert  \mathrm{Tr}_1[((M_{a \vert x}-N_{a \vert x}) \otimes \mathds{1}) (\Lambda \otimes \mathds{1})(\rho')] \rVert_1 \nonumber \\
&= \dfrac{1}{2} \sum_x p(x) \max\limits_{\rho'} \sum_a \lVert  \mathrm{Tr}_1[(\Lambda^{\dagger}(M_{a \vert x}-N_{a \vert x}) \otimes \mathds{1}) \rho'] \rVert_1 \nonumber \\
&= \Ddiamond(\Lambda^{\dagger}(\mathcal{M})_\mathbf{p}, \Lambda^{\dagger}(\mathcal{N})_\mathbf{p}), \nonumber
\end{align}
where we introduced in the second line a new quantum state $\rho'$ (acting on a possibly different Hilbert space) and a \ac{CPT} map $\Lambda$ acting on the first subsystem of $\rho'$. The resulting state $(\Lambda \otimes \mathds{1})(\rho')$ is clearly already included in the optimization of the first line, hence the inequality. In the third line, we used the fact that we can swap the evolution under the \ac{CPT} map $\Lambda$ onto the POVMs by introducing the adjoint channel $\Lambda^{\dagger}$. This results exactly in the definition of $\Ddiamond(\Lambda^{\dagger}(\mathcal{M})_\mathbf{p}, \Lambda^{\dagger}(\mathcal{N})_\mathbf{p})$, from which the monotonicity under quantum channel $\Lambda^{\dagger}$ follows. \\
\indent The monotonicity under classical simulations $\xi$ can be shown in a similar fashion,
\begin{align}
\mathrm{D}(\xi(\mathcal{M}_\mathbf{p})_\mathbf{q}, \xi(\mathcal{N}_\mathbf{p})_\mathbf{q}) &= \dfrac{1}{2} \sum_y q(y) \max\limits_{\rho} \sum_b \lVert  \mathrm{Tr}_1[((M'_{b \vert y}-N'_{b \vert y}) \otimes \mathds{1}) \rho] \rVert_1  \\
&= \dfrac{1}{2} \sum_y q(y) \max\limits_{\rho} \sum_b \lVert \sum_x p(x \vert y) \sum_a q(b \vert y,x,a) \nonumber \mathrm{Tr}_1[((M_{a \vert x} - N_{a \vert x}) \otimes \mathds{1}) \rho]  \rVert_1 \nonumber \\
&\leq  \dfrac{1}{2} \sum_y q(y) \max\limits_{\rho} \sum_b  \sum_x p(x \vert y) \sum_a q(b \vert y,x,a) \lVert \mathrm{Tr}_1[((M_{a \vert x} - N_{a \vert x}) \otimes \mathds{1}) \rho]  \rVert_1 \nonumber \\
&= \dfrac{1}{2} \sum_y q(y) \max\limits_{\rho}   \sum_x p(x \vert y) \sum_a \lVert \mathrm{Tr}_1[((M_{a \vert x} - N_{a \vert x}) \otimes \mathds{1}) \rho]  \rVert_1 \nonumber \\
&\leq \dfrac{1}{2} \sum_{x,y} q(y) p(x \vert y) \max\limits_{\rho}  \sum_a \lVert \mathrm{Tr}_1[((M_{a \vert x} - N_{a \vert x}) \otimes \mathds{1}) \rho]  \rVert_1 \nonumber \\
&= \dfrac{1}{2} \sum_{x} p(x) \max\limits_{\rho}  \sum_a \lVert \mathrm{Tr}_1[((M_{a \vert x} - N_{a \vert x}) \otimes \mathds{1}) \rho]  \rVert_1 \nonumber \\
&= \Ddiamond(\mathcal{M}_\mathbf{p}, \mathcal{N}_\mathbf{p}), \nonumber
\end{align}
\end{widetext}
where we used the following properties. In the first line, we used the definition of $\mathrm{D}(\xi(\mathcal{M}_\mathbf{p})_\mathbf{q}, \xi(\mathcal{N}_\mathbf{p})_\mathbf{q})$ by introducing the assemblages $\mathcal{M}'_\mathbf{q}$ and $\mathcal{N}'_\mathbf{q}$ with measurement outcomes $b$ for the settings $y$, associated with the probability distribution $\mathbf{q}$. \\
\indent In the second line, we use that $ M'_{b \vert y} = \sum_x p(x \vert y) \sum_a q(b \vert y,x,a) M_{a \vert x}$ and the analogous expression for $N'_{b \vert y}$. In the third line, we used the triangle inequality. In the fourth line, we performed the sum over $b$. In the fifth line, we interchanged the maximization with the sum over $x$, which leads to more degrees of freedom since we can now choose a different $\rho$ for each $x$. Finally, in the sixth line, we used that $\sum_y q(y) p(x \vert y) = p(x)$, which leads exactly to the definition of $\Ddiamond(\mathcal{M}_\mathbf{p}, \mathcal{N}_\mathbf{p})$ from which the monotonicity under classical simulations $\xi(\mathcal{M}_\mathbf{p})_\mathbf{q}$ follows. \\
\indent The joint-convexity of $\Ddiamond(\mathcal{M}_\mathbf{p}, \mathcal{N}_\mathbf{p})$ can be seen by first considering the joint-convexity of the diamond distance $\Ddiamond(\Lambda_{\mathcal{M}_x}, \Lambda_{\mathcal{N}_x})$ between measure and prepare channels $\Lambda_{\mathcal{M}_x}$ and $\Lambda_{\mathcal{N}_x}$. The joint-convexity of a norm induced distance follows from the triangle inequality and the absolute homogeneity of the norm, which implies its convexity. It now follows that
\begin{align}
&\Ddiamond(\eta \mathcal{M}^{(1)}_\mathbf{p} + (1-\eta) \mathcal{M}^{(2)}_\mathbf{p}, \eta \mathcal{F}^{(1)}_\mathbf{p} + (1-\eta) \mathcal{F}^{(2)}_\mathbf{p}) \\
&= \sum_x p(x) \Ddiamond(\eta \Lambda_{\mathcal{M}^{(1)}_x} + (1-\eta) \Lambda_{\mathcal{M}^{(2)}_x},\eta \Lambda_{\mathcal{F}^{(1)}_x} + (1-\eta) \Lambda_{\mathcal{F}^{(2)}_x}) \nonumber \\
&\leq  \sum_x p(x) [\eta \Ddiamond(\Lambda_{\mathcal{M}^{(1)}_x},\Lambda_{\mathcal{F}^{(1)}_x}) + (1-\eta) \Ddiamond(\Lambda_{\mathcal{M}^{(2)}_x},\Lambda_{\mathcal{F}^{(2)}_x})] \nonumber \\
&=\eta \mathrm{D}(\mathcal{M}^{(1)}_\mathbf{p}, \mathcal{F}^{(1)}_\mathbf{p}) + (1-\eta) \mathrm{D}(\mathcal{M}^{(2)}_\mathbf{p}, \mathcal{F}^{(2)}_\mathbf{p}), \nonumber   
\end{align}
which concludes the proof.
\end{proof}

\section{Steerability as upper bound to nonlocality} 
\label{Append_steering_bound_nonlocality}

Here, we show that the steerability of a state assemblage $\Vec{\sigma}_{\mathbf{p}_A} $ upper bounds the nonlocality of any probability distribution $\mathbf{q}_{\mathbf{p}}$ obtained from it. This completes the proof of Theorem~\ref{thrm_hierarchy}.
\begin{Lemma}
\label{Lem_steering_vs_nonlocality}
 Let $\Vec{\sigma}_{\mathbf{p}_A} = (\Vec{\sigma},{\mathbf{p}_A})$ be any state assemblage weighted with the probability distribution $\mathbf{p}_A$, $\mathcal{N}_{\mathbf{p}_B}$ any \ac{WMA} of appropriate dimension and $\mathbf{q}_\mathbf{p} = (\mathbf{q},\mathbf{p})$ a probability distribution obtained via $q(a,b\vert x,y) = \mathrm{Tr}[N_{b \vert y} \sigma_{a \vert x} ]$ and $p(x,y) = p_A(x)p_B(y)$. Then, it holds that
 \begin{align}
  \mathrm{S}(\Vec{\sigma}_{\mathbf{p}_A} ) \geq \mathrm{N}(\mathbf{q}_\mathbf{p}).   
 \end{align}
\end{Lemma}
\begin{proof}
Let $\Vec{\tau}^{*}$ be the closest \ac{LHS} assemblage to $\Vec{\sigma}$ with respect to the quantifier $\mathrm{S}(\Vec{\sigma}_{\mathbf{p}_A} )$. We use the fact that unsteerable assemblages always lead to local probability distributions. It follows that
\begin{align}
\mathrm{N}(\mathbf{q}_\mathbf{p}) &\leq \dfrac{1}{2} \sum\limits_{a,b,x,y} p(x,y) \lvert \mathrm{Tr}[N_{b \vert y}(\sigma_{a \vert x}-\tau^*_{a \vert x})] \rvert \\
&\leq \dfrac{1}{2} \sum_{a,b,x,y} p(x,y) \mathrm{Tr}[N_{b \vert y} \lvert \sigma_{a \vert x} - \tau^*_{a \vert x} \rvert] \nonumber \\
&= \dfrac{1}{2} \sum_{a,x} p_A(x) \mathrm{Tr}[ \lvert \sigma_{a \vert x} - \tau^*_{a \vert x} \rvert] \nonumber \\
&= \dfrac{1}{2} \sum_{a,x} p_A(x) \lVert \sigma_{a \vert x} - \tau^*_{a \vert x} \rVert_1 \nonumber = \mathrm{S}(\Vec{\sigma}_{\mathbf{p}_A} ) \nonumber,
\end{align}
where we used in the first line the definition in Eq.~\eqref{nonlocality_distance} of the nonlocality $\mathrm{N}(\mathbf{q}_\mathbf{p})$ and the fact that any measurement on the closest LHS assemblage $\Vec{\tau}^{*}$ (with respect to $\Vec{\sigma}$) results in a local probability distribution. In the second line, we used some basic property of the absolute value and the fact that we can always decompose the difference of two Hermitian matrices like $ \sigma_{a \vert x} - \tau^*_{a \vert x} = T_{a \vert x}- S_{a \vert x}$, where $T_{a \vert x}$ and $S_{a \vert x}$ are positive operators with orthogonal support. It follows that $\vert  \sigma_{a \vert x} - \tau^*_{a \vert x} \vert = T_{a \vert x}+S_{a \vert x}$, where $\vert X \vert = \sqrt{X^{\dagger} X}$. Finally, we used in the third line that $\sum_b N_{b \vert y} = \mathds{1}_d \ \forall \ y$, $\sum_y p(x,y) = p_A(x)$, and the definition of $\mathrm{S}(\Vec{\sigma}_{\mathbf{p}_A})$. Therefore, it follows that the steerability is an upper bound to the nonlocality. 
\end{proof}

\section{Dual formulation of steerability and nonlocality} 
\label{Append_steering_nonlocality_dual}

Here, we show that the steerability $\mathrm{S}(\Vec{\sigma}_{\mathbf{p}_A} )$ and the nonlocality $\mathrm{N}(\mathbf{q}_\mathbf{p})$ can be understood as optimal steering, respectively Bell inequality violation.
\begin{theorem}
Let $\mathrm{S}(\Vec{\sigma}_{\mathbf{p}_A} )$ be the steerability of the state assemblage $\Vec{\sigma}_{\mathbf{p}_A} $. The steerability $\mathrm{S}(\Vec{\sigma}_{\mathbf{p}_A} )$ can be reformulated as the violation of an optimized steering inequality given by
\begin{align}
\mathrm{S}(\Vec{\sigma}_{\mathbf{p}_A} ) = \underset{G_{a,x}, \ell}{\max} \sum_{a,x} p_A(x) \mathrm{Tr}[\sigma_{a \vert x} G_{a \vert x}] - \ell,
\label{append_steering_ineq}
\end{align}
where $\ell=\underset{\Vec{\tau} \in \mathrm{LHS}}{\max} \sum\limits_{a,x} p_A(x) \mathrm{Tr}[\tau_{a \vert x} G_{a \vert x}]$ is the classical bound obeyed by all unsteerable assemblages $\Vec{\tau} \in \mathrm{LHS}$ and the $G_{a \vert x}$ are positive semidefinite matrices s.t. $\lVert G_{a \vert x} \rVert_{\infty} \leq 1$.\\ 
\indent Moreover, the nonlocality $\mathrm{N}(\mathbf{q}_\mathbf{p})$ of any probability distribution $\mathbf{q}_\mathbf{p}$ can be reformulated as the violation of an optimized Bell inequality 
\begin{align}
\mathrm{N}(\mathbf{q}_\mathbf{p}) = \underset{C_{ab \vert xy},\ell'}{\max}\sum_{a,b,x,y} p(x,y) C_{ab\vert xy} q(a,b \vert x,y) - \ell' , 
\end{align}
where $\ell' = \underset{\mathbf{t} \in \mathrm{LHV}}{\max} \sum\limits_{a,b,x,y} p(x,y) C_{ab\vert xy} t(a,b \vert x,y)$ is the local bound obeyed by all local correlations $\mathbf{t} \in \mathrm{LHV}$ and $C_{ab\vert xy}$ are Bell coefficients s.t. $0 \leq C_{ab\vert xy} \leq 1 $.
\end{theorem}
\begin{proof}
The proof relies on the dual formulations of $\mathrm{S}(\Vec{\sigma}_{\mathbf{p}_A} ),$ which can be written in terms of an \ac{SDP}, and $\mathrm{N}(\mathbf{q}_\mathbf{p})$ which can be written as a linear program. We start with the nonlocality $\mathrm{N}(\mathbf{q}_\mathbf{p})$ by stating the optimization for an optimal Bell inequality violation given the distribution $\mathbf{q}$ and by showing that it is dual to $\mathrm{N}(\mathbf{q}_\mathbf{p})$. Note that all of the following optimization problems require the knowledge of the deterministic probability distributions of the corresponding problem, which for the nonlocality $\mathrm{N}(\mathbf{q}_\mathbf{p})$ are denoted by $\lbrace v_A(a \vert x, \lambda) v_B(b \vert y, \lambda) \rbrace$. Since these are fixed for a given problem and are trivially accessible, we will not treat them as input variables.  
\begin{align}
&\underline{\text{Dual problem (nonlocality):}} \label{optimal_Bell_dual} \\
&\mathrm{given}: \ \mathbf{q}_\mathbf{p}  \nonumber \\
&\underset{C_{ab \vert xy}, \ell'}{\mathrm{maximize}} \sum_{a,b,x,y} p(x,y) C_{ab \vert xy} q(a,b \vert x,y) -\ell' \nonumber \\
&\text{subject to:} \nonumber \\
&\ell' \geq \sum_{a,b,x,y} p(x,y) C_{ab \vert xy} v_A(a \vert x, \lambda) v_B(b \vert y, \lambda) \ \forall \ \lambda, \nonumber \\
&0 \leq C_{ab \vert xy} \leq 1 \ \forall \ a,b,x,y, \nonumber 
\end{align}
where $C_{ab \vert xy}$ are the Bell coefficients of the Bell inequality and $\ell'$ is the local bound. Note that $\ell' = \underset{\mathbf{t} \in \mathrm{LHV}}{\max} \sum\limits_{a,b,x,y} p(x,y) C_{ab\vert xy} t(a,b \vert x,y)$ follows directly from the first constraint. Remember that $\mathbf{t}$ admits an \ac{LHV} decomposition according to Eq.\,\eqref{LHV_model}. The equality can be seen by multiplying the constraints $\ell' \geq \sum_{a,b,x,y} p(x,y) C_{ab \vert xy} v_A(a \vert x, \lambda) v_B(b \vert y, \lambda) \ \forall \ \lambda$ with the probabilities $\pi(\lambda)$ before summing all the constraints together. This leads to the bound $\ell' \geq \underset{\mathbf{t} \in \mathrm{LHV}}{\max} \sum\limits_{a,b,x,y} p(x,y) C_{ab\vert xy} t(a,b \vert x,y)$, with $t(a,b \vert x,y)  = \sum\limits_{\lambda} \pi(\lambda) v_A(a \vert x, \lambda) v_B(b \vert y, \lambda)$. The equality follows from the fact that we maximize the objective function. \\
\indent Now, we show that the optimal value of the optimization in Eq.\,\eqref{optimal_Bell_dual} is equal to $\mathrm{N}(\mathbf{q}_\mathbf{p})$ by deriving the primal program. Note that this generally requires dealing with inequality constraints, which can be done by generalizing the method of Lagrange multipliers to using the Karush–Kuhn–Tucker conditions (see e.g.\,\cite{boyd_vandenberghe_2004}). However, since we are interested in formulating dual formulations of convex optimization problems, we can rely on simpler but less general conditions for the equivalence of the primal and the dual problem, which we come back to down below. 
We start by stating the Lagrangian of the problem:
\begin{align}
\mathcal{L} &= \sum_{a,b,x,y} p(x,y) C_{ab \vert xy} q(a,b \vert x,y) -\ell'  \\
&+\sum_{\lambda} \pi(\lambda) \big(\ell'-\sum_{a,b,x,y} p(x,y) C_{ab \vert xy} v_A(a \vert x, \lambda) v_B(b \vert y, \lambda) \big) \nonumber \\
&+\sum_{a,b,x,y} A_{ab \vert xy} (1-C_{ab \vert xy}) + \sum_{a,b,x,y} B_{ab \vert xy}C_{ab \vert xy}, \nonumber 
\end{align}
where we introduced the Lagrange parameters $\pi(\lambda)$, $A_{ab \vert xy}$, and $B_{ab \vert xy}$ to make the constraints explicit. Note that $\sum_{a,b,x,y} p(x,y) C_{ab \vert xy} q(a,b \vert x,y) -\ell' \leq \mathcal{L}$ for any feasible $\lbrace C_{ab \vert xy} \rbrace, \ell'$ in Eq.~\eqref{optimal_Bell_dual}, as long as all the $\pi(\lambda),A_{ab \vert xy},$ and $B_{ab \vert xy}$ are non-negative coefficients.
We obtain the dual function (which is here actually the dual function of the dual problem) by taking the supremum of the Lagrangian over the given (here the dual) variables. More precisely, the dual function is given by 
\begin{align}
&G(\lbrace \pi(\lambda) \rbrace, \lbrace A_{ab \vert xy} \rbrace, \lbrace B_{ab \vert xy} \rbrace) =\underset{C_{ab \vert xy}, \ell'}{\sup} \Big \lbrace \ell'(-1+\sum_{\lambda} \pi(\lambda)) + \sum_{a,b,x,y} A_{ab \vert xy}  \\
&+\sum_{a,b,x,y} C_{ab \vert xy}\Big(p(x,y)q(a,b \vert x,y) \nonumber -p(x,y) \sum_{\lambda} \pi(\lambda) v_A(a \vert x, \lambda) v_B(b \vert y, \lambda) - A_{ab \vert xy} + B_{ab \vert xy}\Big) 
\Big \rbrace. \nonumber  
\end{align}
The dual function is unbounded from above, unless certain constraints (the primal constraints) are met. This can for instance be seen by realizing that unless $(-1+\sum_{\lambda} \pi(\lambda)) = 0$, it is always possible to make the term $\ell'(-1+\sum_{\lambda} \pi(\lambda))$ arbitrarily large, since $\ell'$ is now treated as an unconstrained variable. We obtain the primal program by minimizing the dual function under these constraints. The primal program is given by
\begin{align}
&\underline{\text{Primal problem (nonlocality):}} \label{optimal_Bell_primal} \\
&\mathrm{given}:  \ \mathbf{q}_\mathbf{p} \nonumber \\
&\underset{A_{ab \vert xy}, B_{ab \vert xy}, \pi(\lambda)} {\mathrm{minimize}} \sum_{a,b,x,y} A_{ab \vert xy} \nonumber \\
&\text{subject to:} \nonumber \\
&A_{ab \vert xy} - B_{ab \vert xy} = p(x,y) \big(q(a,b \vert x,y) - \sum_{\lambda} \pi(\lambda) v_A(a \vert x, \lambda) v_B(b \vert y, \lambda)\big), \nonumber \\
&\sum_{\lambda} \pi(\lambda) = 1, \ \pi(\lambda) \geq 0 \ \forall \ \lambda, \ A_{ab \vert xy}, B_{ab \vert xy} \geq 0 \ \forall \ a,b,x,y. \nonumber 
\end{align}
Now, by definition of the $\ell_1$-distance between two normalized probability distributions, the optimal value of Eq.~(\ref{optimal_Bell_primal}) is exactly $\mathrm{N}(\mathbf{q}_\mathbf{p})$ (see, for instance, Remark $4.3$ in \cite{Levin2009-gp}). Since we are dealing with a linear program, there is no duality gap between the primal and the dual formulation. Hence, $\mathrm{N}(\mathbf{q}_\mathbf{p})$ describes the maximal Bell violation possible with the probability distribution $\mathbf{q}_\mathbf{p}$. 

Next, we need to show that $\mathrm{S}(\Vec{\sigma}_{\mathbf{p}_A} )$ corresponds to the optimal steering inequality violation. The procedure is the same as before for the nonlocality. However, this time we start from the primal problem.
\begin{align}
&\underline{\text{Primal problem (steerability):}}  \label{Steering_primal} \\
&\mathrm{given}: \ \Vec{\sigma}_{\mathbf{p}_A}  \nonumber \\
&\underset{\sigma_{\lambda}}{\mathrm{minimize}} \ \dfrac{1}{2}\sum_{a,x} p_A(x) \lVert \sigma_{a \vert x} - \sum_{\lambda} v(a \vert x, \lambda) \sigma_{\lambda} \rVert_1 \nonumber \\
&\text{subject to:} \nonumber \\
&\mathrm{Tr}[\sum_{\lambda} \sigma_{\lambda}] = 1, \sigma_{\lambda} \geq 0 \ \forall \ \lambda. \nonumber 
\end{align}
First, we need to rewrite the trace norm explicitly in \ac{SDP} form. We use the following formulation of the trace norm (see, e.g. \cite{Nemirovski}): $\lVert Z \rVert_1 = \min\limits_Y \Big \lbrace \dfrac{\mathrm{Tr}[Y_1]}{2}+\dfrac{\mathrm{Tr}[Y_2]}{2} : \begin{bmatrix} Y_1 & Z \\ Z^{\dagger} & Y_2 \end{bmatrix} \geq 0 \Big \rbrace$. This leads to the primal problem in explicit \ac{SDP} form
\begin{align}
&\underline{\text{Primal problem (steerability):}}  \\
&\mathrm{given}: \ \Vec{\sigma}_{\mathbf{p}_A}  \nonumber \\
&\underset{\sigma_{\lambda}, U_{a \vert x}, W_{a \vert x}}{\mathrm{minimize}} \ \dfrac{1}{4}\sum_{a,x} p_A(x) \mathrm{Tr}[U_{a \vert x} + W_{a \vert x}] \nonumber \\
&\text{subject to:} \nonumber \\
&\begin{bmatrix}
U_{a \vert x} & \sigma_{a \vert x} - \sum_{\lambda} v(a \vert x, \lambda) \sigma_{\lambda}\\
\sigma_{a \vert x} - \sum_{\lambda} v(a \vert x, \lambda) \sigma_{\lambda} & W_{a \vert x}
\end{bmatrix} \geq 0  \ \forall \ a,x, \nonumber \\
&\mathrm{Tr}[\sum_{\lambda} \sigma_{\lambda}] = 1, \sigma_{\lambda} \geq 0 \ \forall \ \lambda, \nonumber 
\end{align}
where we introduced the Hermitian matrices $U_{a \vert x}, W_{a \vert x}$. This allows us to state the Lagrangian 
\begin{align}
\mathcal{L} &= \dfrac{1}{4}\sum_{a,x} p_A(x) \mathrm{Tr}[U_{a \vert x} + W_{a \vert x}] + \ell'(1-\mathrm{Tr}[\sum_{\lambda} \sigma_{\lambda}]) \\
&-\sum_{a,x} \Big (\mathrm{Tr}[H_{a \vert x}^{11}U_{a \vert x}] + \mathrm{Tr}[H_{a \vert x}^{22}W_{a \vert x}] + \mathrm{Tr}[2H_{a \vert x}^{12}(\sigma_{a \vert x} - \sum_{\lambda} v(a \vert x, \lambda) \sigma_{\lambda})] \Big), \nonumber
\end{align}
where $\ell'$ is a scalar and $H_{a \vert x}^{11},H_{a \vert x}^{12}$, and $H_{a \vert x}^{22}$ are block-matrices such that 
\begin{align}
H_{a \vert x} = \begin{bmatrix}
 H_{a \vert x}^{11} & H_{a \vert x}^{12} \\
 H_{a \vert x}^{12} & H_{a \vert x}^{22}
\end{bmatrix} \geq 0 \ \forall \ a,x.   
\end{align}
Note that $\mathrm{S}(\Vec{\sigma}_{\mathbf{p}_A} ) \geq \mathcal{L}$ for any feasible point of the dual problem in Eq.~\eqref{Steering_primal}. Analogous to the optimization for the nonlocality before, we can now formulate the dual function $G(\lbrace H_{a \vert x} \rbrace, \ell') = \underset{\sigma_{\lambda} \geq 0, U_{a \vert x}, W_{a \vert x}}{\inf} \mathcal{L}$. We obtain the following:
\begin{align}
G(\lbrace H_{a \vert x} \rbrace, \ell') &= \underset{\sigma_{\lambda} \geq 0, U_{a \vert x}, W_{a \vert x}}{\inf} \Big\lbrace \sum_{a,x} \mathrm{Tr}[U_{a \vert x} (\dfrac{1}{4}p_A(x)\mathds{1} - H_{a \vert x}^{11})] + \mathrm{Tr}[W_{a \vert x} (\dfrac{1}{4}p_A(x)\mathds{1} - H_{a \vert x}^{22})] \\
&+ \sum_{\lambda} \mathrm{Tr} [\sigma_{\lambda} (-\ell' \mathds{1}+\sum_{a,x} v(a \vert x, \lambda) 2 H_{a \vert x}^{12})] - \sum_{a,x} 2 \mathrm{Tr} [H^{12}_{a \vert x} \sigma_{a \vert x}] + \ell' \nonumber
\Big\rbrace
\end{align}

By identifying the conditions (the dual constraints) that make the dual function bounded we obtain the following dual program
\begin{align}
&\underline{\text{Dual problem (steerability):}}  \\
&\mathrm{given}: \ \Vec{\sigma}_{\mathbf{p}_A} \nonumber \\
&\underset{H_{a \vert x}, \ell'}{\mathrm{maximize}} \ - \sum_{a,x} \mathrm{Tr}[\sigma_{a \vert x} 2 H_{a \vert x}^{12}] + \ell' \nonumber \\
&\text{subject to:} \nonumber \\
&-\ell' \mathds{1} \geq - \sum_{a,x} v(a \vert x, \lambda) 2 H_{a \vert x}^{12} \ \forall \ \lambda, \nonumber \\
&\begin{bmatrix}
\dfrac{1}{4} p_A(x) \mathds{1} & H_{a \vert x}^{12} \\
 H_{a \vert x}^{12} & \dfrac{1}{4} p_A(x) \mathds{1}
\end{bmatrix} \geq 0 \ \forall \ a,x. \nonumber 
\end{align}
We can rewrite the dual problem in a more convenient form. By identifying the \ac{SDP} formulation of the spectral norm (see, e.g. \cite{Nemirovski}) $\lVert Z \rVert_{\infty} = \min\limits_t \Big \lbrace t : \begin{bmatrix} t \mathds{1} & Z \\ Z^{\dagger} & t \mathds{1} \end{bmatrix} \geq 0 \Big \rbrace $ and substituting $\ell' =-\tilde{\ell}$ and $H_{a \vert x}^{12} = - p_A(x)\dfrac{G'_{a \vert x}}{2}$, we arrive at
\begin{align}
&\underline{\text{Dual problem (steerability):}}  \\
&\mathrm{given}: \ \Vec{\sigma}_{\mathbf{p}_A}  \nonumber \\
&\underset{G'_{a \vert x}, \tilde{\ell}}{\mathrm{maximize}} \ \sum_{a,x} p_A(x) \mathrm{Tr}[G'_{a \vert x} \sigma_{a \vert x} ] - \tilde{\ell} \nonumber \\
&\text{subject to:} \nonumber \\
&\tilde{\ell} \mathds{1} \geq  \sum_{a,x} p_A(x) v(a \vert x, \lambda) G'_{a \vert x} \ \forall \ \lambda, \nonumber \\
& -\dfrac{1}{2}\mathds{1} \leq G'_{a \vert x} \leq \dfrac{1}{2} \mathds{1}. \nonumber
\end{align}
To finally arrive at the dual formulation equivalent to the statement in Eq.~\eqref{append_steering_ineq} we shift the variables such that $G_{a \vert x} = G'_{a \vert x}+\dfrac{1}{2}\mathds{1}$ and $\tilde{\ell} = \ell-\dfrac{1}{2}$. This leads to
\begin{align}
&\underline{\text{Dual problem (steerability):}} \label{steering_dual_sdp}  \\
&\mathrm{given}:  \ \Vec{\sigma}_{\mathbf{p}_A}  \nonumber \\
&\underset{G_{a \vert x}, \ell}{\mathrm{maximize}} \ \sum_{a,x} p_A(x) \mathrm{Tr}[G_{a \vert x} \sigma_{a \vert x} ] - \ell \nonumber \\
&\text{subject to:} \nonumber \\
&\ell \mathds{1} \geq  \sum_{a,x} p_A(x) v(a \vert x, \lambda) G_{a \vert x} \ \forall \ \lambda, \nonumber \\
& 0 \leq G_{a \vert x} \leq \mathds{1}. \nonumber
\end{align}
Note that it follows again directly that the classical bound $\ell$ fulfills $\ell = \underset{\Vec{\tau} \in \mathrm{LHS}}{\max} \sum\limits_{a,x} p_A(x) \mathrm{Tr}[\tau_{a \vert x} G_{a \vert x}]$, where $\Vec{\tau}$ admits an \ac{LHS} as defined in Eq.~\eqref{LHS_def}. 

As last step of the proof, we note that we can always find a strictly feasible point in the SDP corresponding to Eq.~\eqref{steering_dual_sdp} by choosing the $G_{a \vert x}$ proportional to the identity and $\ell$ sufficiently large. Hence there is no duality gap due to Slater's theorem (see e.g.\,\cite{boyd_vandenberghe_2004}). Therefore, $\mathrm{S}(\Vec{\sigma}_{\mathbf{p}_A} )$ can be written as optimized steering inequality, which concludes the proof.
\end{proof}

\section{Entanglement as upper bound for the steerability}
\label{Append_entanglement_bound_steering}

Here, we show that the geometric entanglement $\ETD(\rho)$ defined in Eq.~\eqref{entanglement} as 
\begin{align}
\ETD(\rho) = \min\limits_{\rho_{S} \in\mathrm{Sep}(\mathcal{H} \otimes \mathcal{H})} \TD(\rho,\rho_S), 
\end{align}
where $\mathrm{Sep}(\mathcal{H} \otimes \mathcal{H})$ with $\mathcal{H} \cong \mathds{C}^d$ is the set of separable states, upper bounds the steerability $\mathrm{S}(\Vec{\sigma}_\mathbf{p})$ of the assemblage $\Vec{\sigma}_\mathbf{p} $. More precisely, when $\Vec{\sigma}_\mathbf{p} $ is obtained by performing $d$-dimensional measurements form any \ac{WMA} $\mathcal{M}_\mathbf{p}$ onto any state $\rho \in \mathcal{S}(\mathcal{H} \otimes \mathcal{H})$ via $\sigma_{a \vert x} = \mathrm{Tr}_1[(M_{a \vert x} \otimes \mathds{1}) \rho]$ we show that it follows $\mathrm{S}(\Vec{\sigma}_\mathbf{p} ) \leq \ETD(\rho)$.
Let $\rho_S^*$ be the closest separable state with respect to the given state $\rho$. It holds,
\begin{align}
\mathrm{S}(\Vec{\sigma}_\mathbf{p} ) &\leq \dfrac{1}{2} \sum_{a,x} p(x) \Vert \mathrm{Tr}_1[(M_{a \vert x} \otimes \mathds{1}) (\rho-\rho^*_S)] \Vert_1 \\
&= \dfrac{1}{2} \sum_{a,x} p(x) \max\limits_{\lbrace -\mathds{1} \leq O_{a \vert x} \leq \mathds{1} \rbrace_{a,x}} \mathrm{Tr}[O_{a \vert x} \mathrm{Tr}_1[(M_{a \vert x} \otimes \mathds{1}) (\rho-\rho^*_S)]]  \nonumber \\
&= \dfrac{1}{2} \sum_{a,x} p(x) \max\limits_{\lbrace -\mathds{1} \leq O_{a \vert x} \leq \mathds{1} \rbrace_{a,x}} \mathrm{Tr}[(M_{a \vert x} \otimes O_{a \vert x}) (\rho-\rho^*_S)] \nonumber \\
&= \dfrac{1}{2} \sum_{x} p(x) \max\limits_{\lbrace -\mathds{1} \leq O_{a \vert x} \leq \mathds{1} \rbrace_{a,x}} \sum_a  \mathrm{Tr}[(M_{a \vert x} \otimes O_{a \vert x}) (\rho-\rho^*_S)]   \nonumber \\
&= \dfrac{1}{2} \sum_{x} p(x) \max\limits_{\lbrace -\mathds{1} \leq O_{a \vert x} \leq \mathds{1} \rbrace_{a,x}}   \mathrm{Tr}[\sum_a(M_{a \vert x} \otimes O_{a \vert x}) (\rho-\rho^*_S)]   \nonumber \\
&\leq \dfrac{1}{2} \sum_{x} p(x) \max\limits_{\lbrace -\mathds{1} \leq O_{a \vert x} \leq \mathds{1} \rbrace_{a,x}}   \lvert \mathrm{Tr}[\sum_a(M_{a \vert x} \otimes O_{a \vert x}) (\rho-\rho^*_S)] \rvert   \nonumber \\
&\leq \dfrac{1}{2} \sum_x p(x) \max\limits_{\lbrace -\mathds{1} \leq O_{a \vert x} \leq \mathds{1} \rbrace_{a,x}} \Vert \sum_a(M_{a \vert x} \otimes O_{a \vert x}) \Vert_{\infty} \Vert \rho-\rho^*_S \Vert_1 \nonumber \\
&\leq \dfrac{1}{2} \sum_x p(x) \Vert \sum_a(M_{a \vert x} \otimes \mathds{1}) \Vert_{\infty} \Vert \rho-\rho^*_S \Vert_1 \nonumber \\
&= \dfrac{1}{2} \sum_x p(x) \Vert \rho - \rho^*_S \Vert_1 = \dfrac{1}{2} \Vert \rho - \rho^*_S \Vert_1 = \ETD(\rho), \nonumber
\end{align}
where we used in the first line the definition of the steerability $\mathrm{S}(\Vec{\sigma}_\mathbf{p})$ and the fact that separable states $\rho_S$ cannot lead to steering. In the second line, we used the variational characterization of the trace norm by introducing the optimization variables $\lbrace O_{a \vert x} \rbrace_{a,x}$. In the third line, we used some basic property of the trace and the partial trace. Next, we used in the fourth line, that we can interchange the sum over $a$ and the maximization. In the sixth line, we upper bounded the trace by its absolute value. In the seventh line, we used the Hölder inequality. Finally, in the last line we used the fact that $ -\sum_a(M_{a \vert x} \otimes \mathds{1}) \leq \sum_a(M_{a \vert x} \otimes O_{a \vert x}) \leq \sum_a(M_{a \vert x} \otimes \mathds{1}) $ in the positive semidefinite sense. This lets us find as an upper bound $\Vert \sum_a M_{a \vert x} \otimes \mathds{1} \Vert_{\infty} = 1$, due to the completeness relation of the \acp{POVM}. Therefore, the entanglement $\ETD(\rho)$ limits the steerability $\mathrm{S}(\Vec{\sigma}_\mathbf{p} ) \leq \ETD(\rho)$.
\section{\texorpdfstring{\ac{SDP}}{SDP} formulations of incompatibility}
\label{Append_explicit_SDP}
Here, we give detailed information about the \ac{SDP} formulations in Eq.~\eqref{SDP_general_primal} and Eq.~\eqref{SDP_general_dual}. As an example, we explicitly derive the primal and the dual formulation of the incompatibility quantifier $\Idiamond(\mathcal{M}_\mathbf{p})$. More specifically, we show that the incompatibility $\Idiamond(\mathcal{M}_\mathbf{p})$ is the optimal value of the following two \acp{SDP}.
\begin{align}
&\underline{\text{Primal problem (incompatibility):}} \label{Incomp_SDP_primal} \\
&\mathrm{given:} \ \mathcal{M}_\mathbf{p} \nonumber \\
&\underset{a_x, Z_x, G_{\lambda}}{\mathrm{minimize}} \sum_x p(x) a_x \nonumber \\
&\text{subject to:} \nonumber \\
 &a_x \mathds{1} - \mathrm{Tr}_1[Z_x] \geq 0 \ \forall \ a,x, \nonumber \\
 &Z_x \geq  \sum_a \vert a \rangle \langle a \vert \otimes (M_{a \vert x} - F_{a \vert x})^T \ \forall \ x, \nonumber \\
 &F_{a \vert x} = \sum_{\lambda} v(a \vert x, \lambda) G_{\lambda} \ \forall \ x,a, \ G_{\lambda} \geq 0 \ \forall \ \lambda, \sum_{\lambda} G_{\lambda} = \mathds{1}, \nonumber \\
&Z_x \geq 0, \ a_x \geq 0 \ \forall \ x, \nonumber 
\end{align}
\begin{align}
&\underline{\text{Dual problem (incompatibility):}} \label{Incomp_SDP_dual} \\
&\mathrm{given}: \ \mathcal{M}_\mathbf{p} \nonumber  \\
&\underset{C_{a \vert x}, \rho_x,L}{\mathrm{maximize}}  \ \ \ \sum_{a,x} p(x) \mathrm{Tr}[M_{a \vert x} C_{a \vert x}] - \mathrm{Tr}[L] \nonumber \\
&\text{subject to:} \nonumber \\
& L \geq \sum_{a,x} p(x) v(a \vert x, \lambda) C_{a \vert x} \ \forall \ \lambda, \nonumber \\
&0 \leq C_{a \vert x} \leq \rho_x \ \forall \ a,x, \ \rho_x \geq 0, \mathrm{Tr}[\rho_x] = 1 \ \forall \ x, \nonumber 
\end{align}
where $\lbrace v(a  \vert x,\lambda) \rbrace$ are the deterministic probability distributions. The optimization variables of the primal problem are the positive coefficients $a_x$, and the positive semidefinite matrices $Z_x$ and $G_{\lambda}$. The optimization variables of the dual problem are the positive semidefinite matrices $C_{a \vert x}$, $\rho_x$, and $L$. \\
\indent The formulation of the primal problem heavily relies on the \ac{SDP} formulation of the diamond norm due to Watrous \cite{0901.4709}, see also \cite{PhysRevLett.122.190405}. Let us recall that the Choi–Jamio\l kowski-matrix of a measure-and-prepare channel (see Eq.~\eqref{Choi–Jamiołkowski_state}) corresponding to one \ac{POVM} $\mathcal{M}_x = \lbrace M_{a \vert x} \rbrace_a$ is given by
\begin{align}
J(\mathcal{M}_x) = \sum_a \vert a \rangle \langle a \vert \otimes M_{a \vert x}^T,   
\end{align}
where the transpose is with respect to the computational basis. The diamond distance between the quantum channels $\Lambda_{\mathcal{M}_x}$ and $\Lambda_{\mathcal{F}_x}$ can now be computed as
\begin{align}
&\mathrm{given}: \ \ \ J(\mathcal{M}_x), J(\mathcal{F}_x) \\
&\underset{Z_x}{\mathrm{minimize}}  \ \ \ \lVert \mathrm{Tr}_1[Z_x] \rVert_{\infty} \nonumber \\
& \text{subject to:} \ Z_x \geq J(\mathcal{M}_x)-J(\mathcal{F}_x), \ Z_x \geq 0.  \nonumber 
\end{align}
Using this form of the diamond norm, the primal problem in Eq.~\eqref{SDP_general_primal} follows by summing over the settings $x$ weighted with probabilities $p(x)$ and by explicitly minimizing over the Choi–Jamio\l kowski-matrices $J(\mathcal{F}_x)$, where $\mathcal{F}_x$ is the \ac{POVM} corresponding to setting $x$ of the free assemblages $\mathcal{F}$. \\
\indent To arrive at the specific \ac{SDP} for the incompatibility in Eq.~\eqref{Incomp_SDP_primal} from the general formulation in Eq.\,\eqref{SDP_general_primal} we first note that the spectral norm $ \lVert \mathrm{Tr}_1[Z_x] \rVert_{\infty}$ of a positive semidefinite matrix $\mathrm{Tr}_1[Z_x]$ can be written as the minimal value $a_x$ such that $a_x \mathds{1} \geq \mathrm{Tr}_1[Z_x]$ holds. Next, we write out the Choi–Jamio\l kowski-matrices corresponding to the channels $\Lambda_{\mathcal{M}_x}$ and $\Lambda_{\mathcal{F}_x}$ in terms of the \ac{POVM} elements $M_{a \vert x}$ and $F_{a \vert x}$. Finally, we constraint the $F_{a \vert x}$ explicitly to be \ac{JM} i.e., that it holds $F_{a \vert x} = \sum_{\lambda} v(a \vert x, \lambda) G_{\lambda} \ \forall \ x,a$, where $\lbrace G_{\lambda} \rbrace$ is the \ac{POVM} simulating $\mathcal{F}$. \\
\indent To derive the dual formulation in Eq.\,\eqref{Incomp_SDP_dual}, we formulate the Lagrangian of the primal problem by incorporating the constraints explicitly. The Lagrangian is given by
\begin{align}
&\mathcal{L}  = \sum_x p(x) a_x + \sum_x \mathrm{Tr}[H_x (\mathrm{Tr}_1[Z_x]-a_x\mathds{1})] \\
&+ \sum_x \mathrm{Tr}[C_x \big (\sum_a \lvert a \rangle \langle a \rvert \otimes [M_{a \vert x} -\sum_{\lambda} v(a \vert x, \lambda)  G_{\lambda}]^T-Z_x \big )]
+ \mathrm{Tr}[L(\sum_{\lambda} G_{\lambda} - \mathds{1})], \nonumber   
\end{align}
where the $H_x$ are $d$-dimensional positive semidefinite matrices, the $C_x$ are $(o \times d)$-dimensional positive semidefinite matrices, where $o$ is the number of measurement outcomes, and $L$ is a $d$-dimensional Hermitian matrix. Note that for every feasible point in Eq.\,\eqref{Incomp_SDP_primal}, it holds  $\Idiamond(\mathcal{M}_\mathbf{p}) \geq \mathcal{L}$. \\
\indent By using the property $\mathrm{Tr}[\mathrm{Tr}_1(Z_x) H_x] = \mathrm{Tr}[Z_x(\mathds{1} \otimes H_x)]$, we can formally state the dual function $G(\lbrace H_x \rbrace,\lbrace C_x \rbrace,L)$, which is obtained by taking the infimum of the Lagrangian $\mathcal{L}$ over the variables of the primal problem. More precisely, 
\begin{align}
&G(\lbrace H_x \rbrace,\lbrace C_x \rbrace,L) =\\  
&\underset{a_x,Z_x,G_{\lambda} \geq 0}{\inf} \Big\lbrace \sum_x a_x(p(x)-\mathrm{Tr}[H_x]) -\mathrm{Tr}[L]  + \sum_x \mathrm{Tr}[Z_x(\mathds{1} \otimes H_x-C_x)] \nonumber \\
&+ \sum_{\lambda} \mathrm{Tr}[G_{\lambda}(L-\sum_{a,x} v(a \vert x, \lambda) \mathrm{Tr}_1[(\vert a \rangle \langle a \vert \otimes \mathds{1}) C_x]^T)] \nonumber + \sum_x \mathrm{Tr}[C_x(\sum_a \vert a \rangle \langle a \vert \otimes M_{a \vert x}^T)] \nonumber \Big\rbrace.  \nonumber  
\end{align}
It is clear that $G(\lbrace H_x \rbrace,\lbrace C_x \rbrace,L)$ is unbounded from below, unless certain conditions (the dual constraints) are met. For instance, if $p(x)-\mathrm{Tr}[H_x] < 0$ holds for some $x$, the corresponding term $a_x(p(x)-\mathrm{Tr}[H_x])$ can be made arbitrarily small by increasing $a_x$, which is only constrained to be non-negative. The corresponding dual program is obtained by maximizing the dual function $G$ over the dual variables $\lbrace C_x \rbrace,\lbrace D_x \rbrace,L$ under the dual constraints. This leads to the optimal lower bound to the primal problem. We obtain
\begin{align}
&\underline{\text{Dual problem (incompatibility):}} \label{SDP_dual_diamond_program} \\
&\mathrm{given}: \ \mathcal{M}_\mathbf{p}  \nonumber \\
&\underset{C_{x}, H_x,L}{\mathrm{maximize}}  \ \ \ \sum_{x}  \mathrm{Tr}[C_x ( \sum_a \lvert a \rangle \langle a \rvert \otimes M_{a \vert x}^T)] - \mathrm{Tr}[L] \nonumber \\
&\text{subject to:} \nonumber \\
&L \geq \sum_{a,x}  v(a \vert x, \lambda) \mathrm{Tr}_1[(\vert a \rangle \langle a \vert \otimes \mathds{1}) C_x]^T \ \forall \ \lambda, \nonumber \\
&H_x \geq 0, \ C_x \geq 0, \ p(x) \geq \mathrm{Tr}[H_x] \ \forall \ x, \nonumber \\
&\mathds{1} \otimes H_x - C_x \geq 0, \ \forall \ x, \ L = L^{\dagger}, \nonumber 
\end{align}
which is formally the dual program to the primal formulation of Eq.\,\eqref{Incomp_SDP_primal}. However, we can rewrite the program in Eq.\,\eqref{SDP_dual_diamond_program} in a more useful form.

First, we can get rid of all transposes by using $\mathrm{Tr}[A B^T] = \mathrm{Tr}[A^T B]$, i.e., by swapping the transposition on the optimization variables $C_x$. Since the transpose $C_x^T$ is already included in the optimization, we can simply ignore it. Second, we rewrite the first term of the objective function as $\mathrm{Tr}[C_x ( \sum_a \lvert a \rangle \langle a \rvert \otimes M_{a \vert x})] = \sum_a \mathrm{Tr}[M_{a \vert x} (\mathrm{Tr}_1[(\lvert a \rangle \langle a \rvert \otimes \mathds{1}) C_x])]$. This shows that only the block-diagonal entries of $C_x$ are important. Note that the same observation holds for the constraints $C_x$ is involved in. It is therefore no loss of generalization to assume $C_x$ as block diagonal. We denote $ \mathrm{Tr}_1[(\lvert a \rangle \langle a \rvert \otimes \mathds{1}) C_x] = C_{a \vert x}'$. \\
\indent With this, we arrive at 
\begin{align}
&\underline{\text{Dual problem (incompatibility):}} \\
&\mathrm{given}: \ \mathcal{M}_\mathbf{p} \nonumber \\
&\underset{C_{a \vert x}', H_x, L}{\mathrm{maximize}}  \ \ \ \sum_{a,x}  \mathrm{Tr}[M_{a \vert x} C_{a \vert x}'] - \mathrm{Tr}[L] \nonumber \\
&\text{subject to:} \nonumber \\
& L \geq \sum_{a,x}  v(a \vert x, \lambda) C_{a \vert x}' \ \forall \ \lambda, \nonumber \\
&H_x \geq 0, p(x) \geq \mathrm{Tr}[H_x] \ \forall \ x, H_x \geq C_{a \vert x}' \geq 0 \ \forall \ a,x, L = L^{\dagger}. \nonumber 
\end{align}
Next, we note that it is always possible (without loss of optimality) to chose $ p(x) = \mathrm{Tr}[H_x] $, since $H_x$ constraints other variables only from above. We rewrite $H_x = p(x) \rho_x$, introducing the variables $\rho_x \geq 0$ with $\mathrm{Tr}[\rho_x] = 1$. Finally, we rewrite $C_{a \vert x}'$ such that $C_{a \vert x}' = p(x) C_{a \vert x}$ which leads to
\begin{align}
&\underline{\text{Dual problem (incompatibility):}} \label{Incomp_SDP_dual_final} \\ 
&\mathrm{given}: \ \mathcal{M}_\mathbf{p} \nonumber  \\
&\underset{C_{a \vert x}, \rho_x, L}{\mathrm{maximize}}  \ \ \ \sum_{a,x} p(x) \mathrm{Tr}[M_{a \vert x} C_{a \vert x}] - \mathrm{Tr}[L] \nonumber \\
&\text{subject to:} \nonumber \\
& L \geq \sum_{a,x} p(x) v(a \vert x, \lambda) C_{a \vert x} \ \forall \ \lambda, \nonumber \\
&0 \leq p(x)C_{a \vert x} \leq p(x)\rho_x \ \forall \ a,x, \ \rho_x \geq 0, \mathrm{Tr}[\rho_x] = 1 \ \forall \ x, \nonumber     
\end{align}
which is equivalent to the \ac{SDP}~\eqref{Incomp_SDP_dual}. Note that by virtue of the first constraint in Eq.~\eqref{Incomp_SDP_dual_final}, it follows that 
\begin{align}
\mathrm{Tr}[L] = \underset{\mathcal{F} \in \mathscr{F}_{\mathrm{JM}}}{\max} \sum_{a,x} p(x) \mathrm{Tr}[F_{a \vert x} C_{a \vert x}].    
\end{align}
 The bound $\mathrm{Tr}[L] \geq \underset{\mathcal{F} \in \mathscr{F}_{\mathrm{JM}}}{\max} \sum_{a,x} p(x) \mathrm{Tr}[F_{a \vert x} C_{a \vert x}]$ can be seen by multiplying all the inequalities $L \geq \sum_{a,x} p(x) v(a \vert x, \lambda) C_{a \vert x} \ \forall \ \lambda$ with $G_{\lambda}$, then summing over all $\lambda $ and taking the trace. The equality follows from the fact, that a strict inequality would contradict with the maximization of the objective function. \\
 \indent As a final step in the proof, we need to show that there is no duality gap between the primal and the dual program. This follows from Slater's theorem (see e.g. \cite{boyd_vandenberghe_2004}) since it always possible to find a strictly feasible point in either the primal or the dual problem. This can be seen directly for the dual program in Eq.~\eqref{Incomp_SDP_dual}, as we can chose all $C_{a \vert x}$ to be proportional to the identity and adjust the $\rho_x$ and $L$ accordingly.

 \section{Optimal input distribution}
\label{section_optimal_input_distribution}

Here, we give an example where an optimization over the input distribution $\mathbf{p} = \lbrace p(x) \rbrace$ for the settings $x$ is relevant to optimize the available resources. In particular, we show that for the incompatibility $\Idiamond(\mathcal{M}_\mathbf{p})$ (see Eq.~\eqref{Incompatibility}) of a measurement assemblage with only two settings, the optimal incompatibility is not always achieved for a uniform distribution $p(1) = p(2) = \dfrac{1}{2}$. The idea is to introduce noise in only one of the measurement settings, here for $x=2$. Let us consider an \ac{MUB} measurement assemblage $\mathcal{N}$ containing $m=2$ \acp{POVM}, constructed in the same way as the assemblages considered in Table \ref{Table1}. From the \ac{MUB} measurement assemblage $\mathcal{N}$ we obtain the measurement assemblage $\mathcal{M}$ via
\begin{align}
M_{a \vert 1} = N_{a \vert 1} \ \forall \ a, \label{noisy_assemblage} \\
M_{a \vert 2} = \mu N_{a \vert 2} +(1-\mu) \mathrm{Tr}[N_{a \vert 2}]\dfrac{\mathds{1}}{d} \ \forall \ a, \nonumber
\end{align}
where $\mu \in [0,1]$ is a depolarizing noise parameter for the second measurement. In the following, we analyze how to choose the probability distribution $\mathbf{p} = \lbrace p(x) \rbrace$ such that the incompatibility $\Idiamond(\mathcal{M}_\mathbf{p})$ is maximized for the given assemblage $\mathcal{M}$. As mentioned in section \ref{section_SDP}, the \ac{SDP} \eqref{SDP_general_dual} can be rewritten such that it includes a maximization over the input distribution $\mathbf{p} = \lbrace p(x) \rbrace$. We illustrate our results in Figure \ref{Optimal_distribution_plot} for the optimal setting probabilities $p(x)$ of the assemblage $\mathcal{M}$ in dimension $d$ with noise parameter $\mu$.

\begin{figure}
\centering
\includegraphics[scale = 0.38]{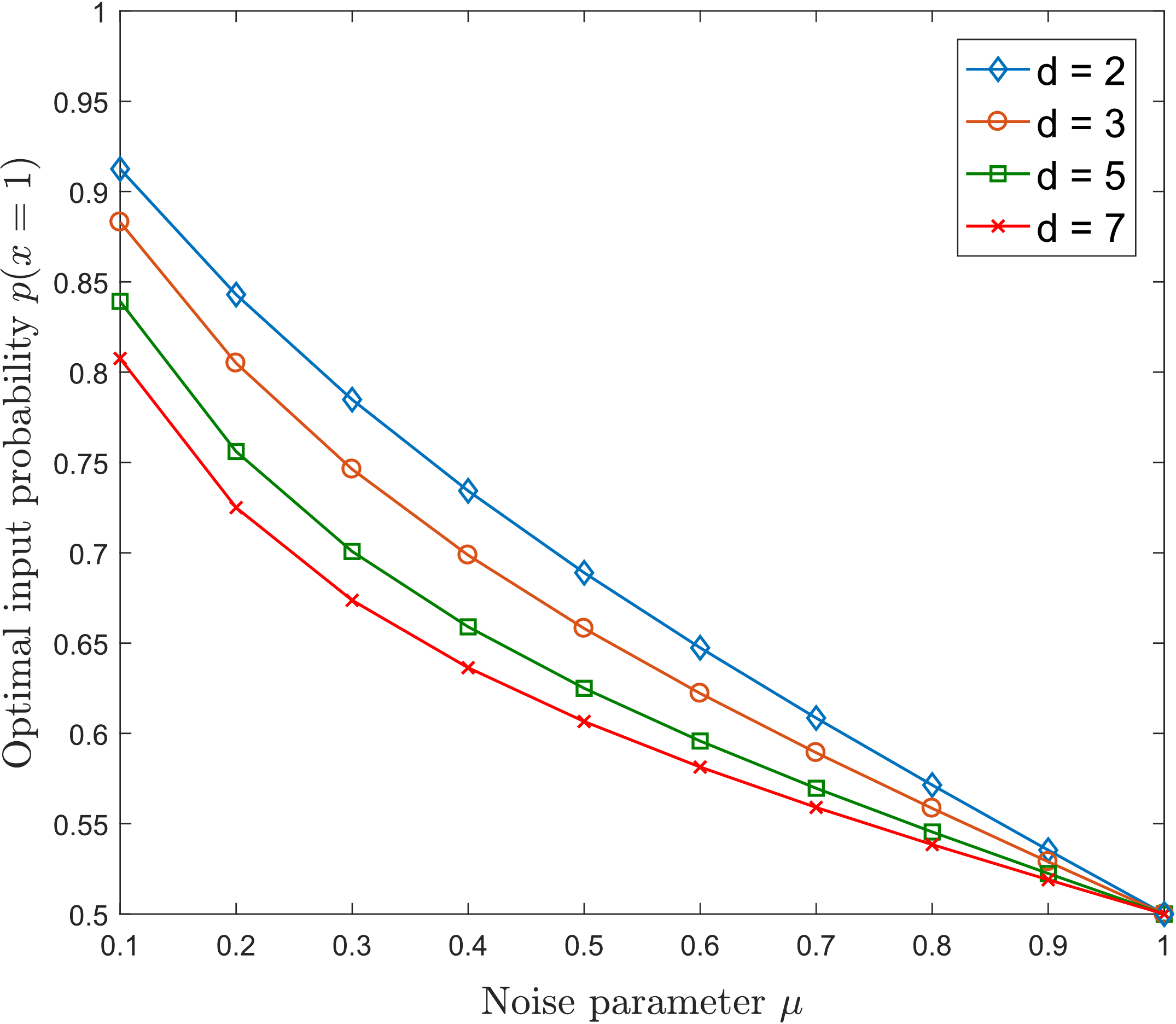}
 \caption{The optimal input probability $p(x = 1)$ for the first (noise free) measurement setting of the assemblage $\mathcal{M}$ in Eq.\,\eqref{noisy_assemblage} depending on the dimension $d$ and the noise parameter $\mu$. The plot shows the optimal probability $p(x=1)$ which maximizes the incompatibility $\Idiamond(\mathcal{M}_\mathbf{p})$. It can be seen that a uniform distribution is only optimal in the absence of noise (i.e. $\mu = 1$). Especially for high noise regimes (e.g. $\mu = 0.1$) a strong bias towards the noise free measurement can be seen. However, this strong bias decreases with increasing dimension~$d$.} 
 \label{Optimal_distribution_plot}
\end{figure}
As one can see, even for only two measurements, strong biases towards one setting can be necessary in order to maximize the incompatibility $\Idiamond(\mathcal{M}_\mathbf{p})$. We want to remark that except for the noise free case, i.e. $\mu = 1$, the optimized input distribution leads to a strictly larger incompatibility than with a uniform distribution. Note that in this particular example, the advantage is weak as can be seen in Figure \ref{Optimal_distribution_comparison_plot}.
\begin{figure}[ht]
\centering
\includegraphics[scale = 0.42]{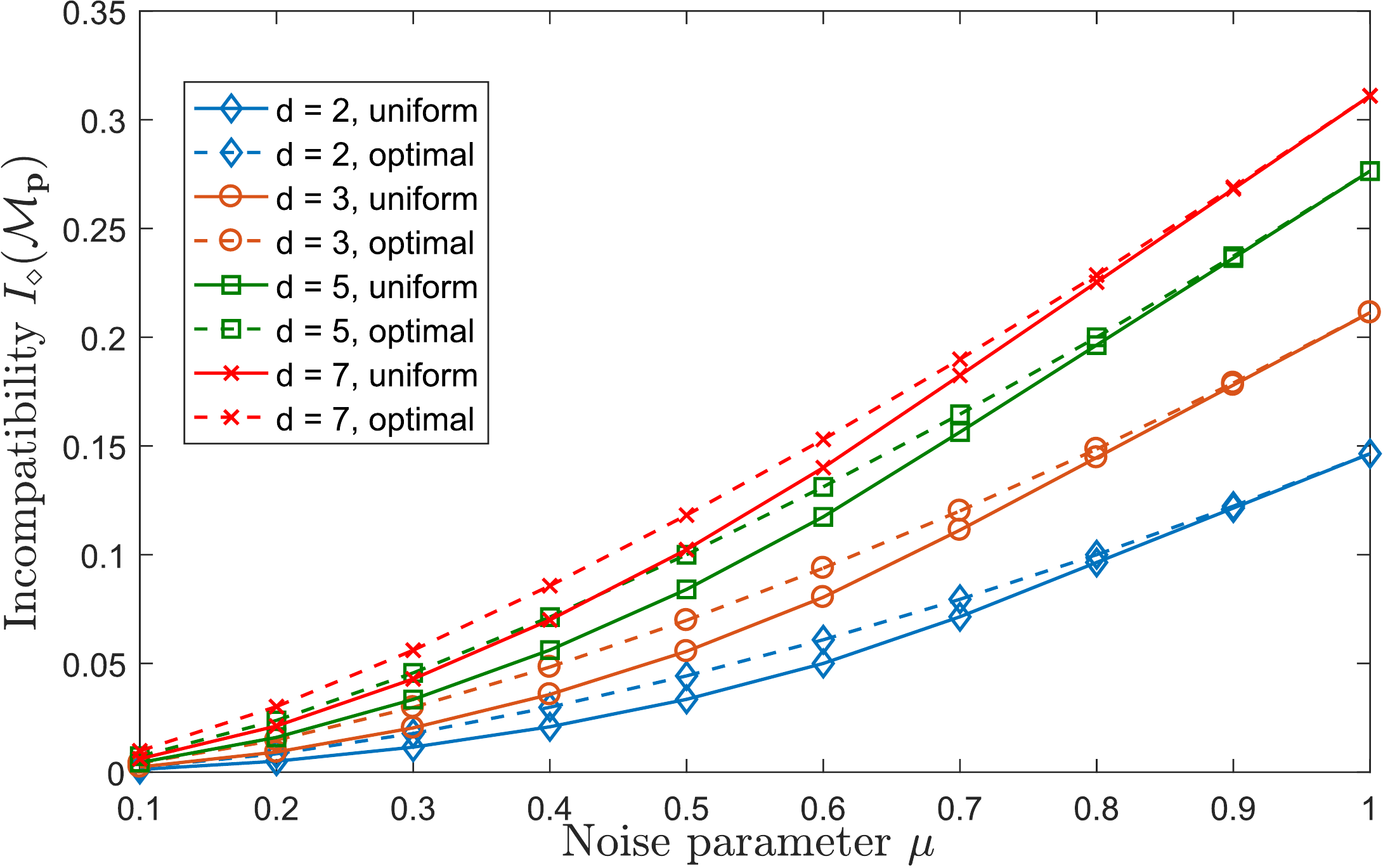}
 \caption{Comparison of the incompatibility $\Idiamond(\mathcal{M}_{\mathbf{p}})$ between the optimal input distribution (dashed lines) and the uniform distribution (solid lines) depending on the noise parameter $\mu$ for a given dimension $d$. It can be seen that the optimized input distribution outperforms the uniform distribution for the measurement assemblage described in Eq.~\eqref{noisy_assemblage}. For low noise regime ($\mu$ close to $1$) the solid and the dashed lines approach each other, as the uniform distribution is optimal for $\mu = 1$.} 
 \label{Optimal_distribution_comparison_plot}
\end{figure}
However, an optimization over the distribution $\mathbf{p}$ can lead to a strong increase in incompatibility for $m \geq 3$, by essentially neglecting weakly incompatibly subsets of measurements. On a qualitative basis, this effect can be explained in terms of the informativeness $\IFdiamond(\mathcal{M}_\mathbf{p})$ (see Eq. \eqref{Informativeness}). For large noise (e.g. $\mu = 0.1$) the distribution is strongly biased towards the noise-free, hence more informative measurement.
 
\section{\texorpdfstring{\ac{SDP}}{SDP} formulation of informativeness and coherence}
\label{Append_inform_and_coherence}
 
 \indent While the calculations in Appendix~\ref{Append_explicit_SDP} are specific to the quantifier $\Idiamond(\mathcal{M}_\mathbf{p})$ and the \ac{QRT} of incompatibility, analogous considerations can be made for any resource that has a free set $\mathscr{F}$ that admits a formulation as an \ac{SDP}. In order to not repeat almost the same calculation as above, we simply state the corresponding \ac{SDP} formulations for the coherence and the informativeness in the following. We start with the latter. The informativeness $\IFdiamond(\mathcal{M}_\mathbf{p})$ is given as the optimal value of the following \acp{SDP}. 
\begin{align}
&\underline{\text{Primal problem (informativeness):}} \label{Info_SDP_primal} \\
&\mathrm{given:} \ \ \ \mathcal{M}_\mathbf{p} \nonumber \\
&\underset{a_x, Z_x, q(a \vert x)}{\mathrm{minimize}} \sum_x p(x) a_x \nonumber \\
&\text{subject to:} \nonumber \\
 &a_x \mathds{1} - \mathrm{Tr}_1[Z_x] \geq 0 \ \forall \ a,x, \nonumber \\
 &Z_x \geq  \sum_a \vert a \rangle \langle a \vert \otimes (M_{a \vert x} - F_{a \vert x})^T \ \forall \ x, \nonumber \\
 &F_{a \vert x} = q(a \vert x) \mathds{1}, \ \ q(a \vert x) \geq 0 \ \forall \ a,x, \sum_{a} q(a \vert x) = 1, \forall \ a, \nonumber \\
&Z_x \geq 0, a_x \geq 0 \ \forall \ x, \nonumber \\
\mbox{} \nonumber 
\end{align}
\begin{align}
&\underline{\text{Dual problem (informativeness):}} \label{Info_SDP_dual} \\
&\mathrm{given}: \ \ \ \mathcal{M}_\mathbf{p} \nonumber  \\
&\underset{C_{a \vert x}, \rho_x, \ell_x}{\mathrm{maximize}}  \ \ \ \sum_{a,x} p(x) \mathrm{Tr}[M_{a \vert x} C_{a \vert x}] - \sum_x \ell_x \nonumber \\
&\text{subject to:} \nonumber \\
& \ell_x \geq p(x) \mathrm{Tr}[C_{a \vert x}] \ \forall \ a,x, \nonumber \\
&0 \leq C_{a \vert x} \leq \rho_x \ \forall \ a,x, \ \ \rho_x \geq 0, \mathrm{Tr}[\rho_x] = 1 \ \forall \ x. \nonumber 
\end{align}
The optimization variables of the primal problem are the positive coefficients $a_x$, the positive semidefinite matrices $Z_x$ and the probabilities $q(a \vert x)$. The optimization variables of the dual problem are the positive semidefinite matrices $C_{a \vert x}$, $\rho_x$, and the scalars $\ell_x$. Note that it follows directly from the first constraint of the dual that 
\begin{align}
\sum_x \ell_x = \underset{\mathcal{F} \in \mathscr{F}_{\mathrm{UI}}}{\max} \sum_{a,x} p(x) \mathrm{Tr}[F_{a \vert x} C_{a \vert x}].    
\end{align}
This can be seen by realizing that $\ell_x \geq p(x) \mathrm{Tr}[C_{a \vert x}] \ \forall a,x$ implies $\ell_x q(a \vert x) \geq p(x) \mathrm{Tr}[C_{a \vert x} q(a \vert x) \mathds{1}] \\ \forall \ a,x$, where we multiplied both sides with the conditional probabilities $q(a \vert x)$. We identify $q(a \vert x) \mathds{1} = F_{a \vert x}$ due to the definition of the \ac{UI} measurements in Eq.\,\eqref{informativeness_free}. Finally, we sum both sides over $a$ and $x$. The equality follows from the fact that we maximize the objective function. 

Like for the incompatibility, the \ac{SDP} formulations of the informativeness $\IFdiamond(\mathcal{M}_\mathbf{p})$ allow us to gain additional insight on the informativeness of \ac{WMA} $\mathcal{M}_\mathbf{p}$. In particular, we show in the following that the informativeness $\IFdiamond(\mathcal{M}_\mathbf{p})$ of any set of rank$-1$ projective measurements is given by $\IFdiamond(\mathcal{M}_\mathbf{p}) = 1 - \dfrac{1}{d}$ for any probability distribution $p(x)$.

This follows by choosing $C_{a \vert x} = \dfrac{M_{a \vert x}}{d}$, $\rho_x = \sum_a \dfrac{M_{a \vert x}}{d} = \dfrac{\mathds{1}}{d}$, and $\ell_x = \dfrac{p(x)}{d}$ as feasible solution for the dual problem in Eq.\,\eqref{Info_SDP_dual}. It follows by direct calculation that $\IFdiamond(\mathcal{M}_\mathbf{p}) \geq 1 -\dfrac{1}{d}$. As feasible solution for the primal problem in Eq.\,\eqref{Info_SDP_primal} we chose $q(a \vert x) = \dfrac{1}{d},$ $a_x = \dfrac{1}{d},$ and $Z_x = (1-\dfrac{1}{d}) \sum_a \vert a \rangle \langle a \vert \otimes M_{a \vert x}$. This leads by direct computation to the upper bound $\IFdiamond(\mathcal{M}_\mathbf{p}) \leq 1 -\dfrac{1}{d}$, which equals the lower bound.

\indent \textit{State discrimination.}\textemdash 
Note that in the particular case of rank-$1$ projective measurements, the informativeness $\IFdiamond(\mathcal{M}_\mathbf{p})$ can be understood as state discrimination game. More precisely, we consider the following task. Given the (known) states $\lbrace \rho_{a \vert x} \rbrace$ that are distributed with probability $p(a \vert x) = \dfrac{1}{d}$ given that we chose the setting $x$ with probability $p(x)$. That is, with probability $p(a,x) = p(a \vert x) p(x)$ we receive the state $\rho_{a \vert x}$. The goal of the game is now to identify the label $a$ correctly with as high probability as possible, given that we have access to the measurement assemblage $\mathcal{M}$. We compare this to the situation where we do not perform a measurement and simply guess the label $a$. 
The term $\sum_{a,x} p(x) \mathrm{Tr}[M_{a \vert x} C_{a \vert x}] = \sum_{a,x} \dfrac{1}{d} p(x) \mathrm{Tr}[M_{a \vert x} \rho_{a \vert x}]$ can be seen as average probability to correctly guess the label $a$ of the states $\rho_{a \vert x}$ that are send out with probability $\dfrac{1}{d}$ given that the setting $x$ has been chosen with probability $p(x)$. The value $\IFdiamond(\mathcal{M}_\mathbf{p}) = 1 - \dfrac{1}{d}$ describes now the difference of the optimal average probability, achieved with the assemblage $\mathcal{M}_\mathbf{p}$ compared to randomly guessing the label $a$ for each setting $x$.  

\indent \textit{Coherence.}\textemdash Finally, the coherence $\Cdiamond(\mathcal{M}_\mathbf{p})$ can also be computed by \acp{SDP}.
In particular, $\Cdiamond(\mathcal{M}_\mathbf{p})$ is the optimal value of the following two \acp{SDP}.
\\
\begin{align}
&\underline{\text{Primal problem (coherence):}} \label{Coherence_SDP_primal} \\
&\mathrm{given:} \ \ \ \mathcal{M}_\mathbf{p} \nonumber \\
&\underset{a_x, Z_x,  \alpha_{i \vert (a,x)}}{\mathrm{minimize}} \sum_x p(x) a_x \nonumber \\
&\text{subject to:} \nonumber \\
 &a_x \mathds{1} - \mathrm{Tr}_1[Z_x] \geq 0 \ \forall \ a,x, \nonumber \\
 &Z_x \geq  \sum_a \vert a \rangle \langle a \vert \otimes (M_{a \vert x} - \sum_{i} \alpha_{i \vert (a,x)} \vert i \rangle \langle i \vert)^T \ \forall \ x, \nonumber \\
 &\alpha_{i \vert (a,x)} \geq 0 \ \forall \ i,a,x, \sum_{a} \alpha_{i \vert (a,x)} = 1, \forall \ i,x, \nonumber \\
&Z_x \geq 0, a_x \geq 0 \ \forall \ x, \nonumber 
\end{align}
\begin{align}
&\underline{\text{Dual problem (coherence):}} \label{Coherence_SDP_dual} \\
&\mathrm{given}: \ \ \ \mathcal{M}_\mathbf{p} \nonumber  \\
&\underset{C_{a \vert x}, \rho_x, \ell_{x,i}}{\mathrm{maximize}}  \ \ \ \sum_{a,x} p(x) \mathrm{Tr}[M_{a \vert x} C_{a \vert x}] - \sum_{x,i} \ell_{x,i} \nonumber \\
&\text{subject to:} \nonumber \\
& \ell_{x,i} \geq p(x) \mathrm{Tr}[C_{a \vert x} \lvert i \rangle \langle i \rvert] \ \forall \ a,x,i, \nonumber \\
&0 \leq C_{a \vert x} \leq \rho_x \ \forall \ a,x, \ \ \rho_x \geq 0, \mathrm{Tr}[\rho_x] = 1 \ \forall \ x. \nonumber 
\end{align}
The optimization variables of the primal problem are the positive coefficients $a_x$, the positive semidefinite matrices $Z_x$ and the coefficients $\alpha_{i \vert (a,x)}$. The optimization variables of the dual problem are the positive semidefinite matrices $C_{a \vert x}$, $\rho_x$, and the scalars $\ell_{x,i}$. With the same reasoning as with the previous resources, it can directly be seen that 
\begin{align}
\sum_{x,i} \ell_{x,i} =  \underset{\mathcal{F} \in \mathscr{F}_{\mathrm{IC}}}{\max} \sum_{a,x} p(x) \mathrm{Tr}[F_{a \vert x} C_{a \vert x}].
\end{align}
\indent We use these insights about the informativeness $\IFdiamond(\mathcal{M}_\mathbf{p})$ and the coherence $\Cdiamond(\mathcal{M}_\mathbf{p})$ to identify non-trivial cases for which it holds that  $\IFdiamond(\mathcal{M}_\mathbf{p}) = \Cdiamond(\mathcal{M}_\mathbf{p})$ in the following. We start by considering assemblages $\mathcal{M}_\mathbf{p}$ where every \ac{POVM} is a rank-$1$ projective measurement that is mutually unbiased to the incoherent basis. More formally, it has to hold 
\begin{align}
\mathrm{Tr}[\vert i \rangle \langle i \vert M_{a \vert x}] = \dfrac{1}{d} \ \forall \ i,a,x.    
\label{coherence_condition}
\end{align}
Note that this holds true for appropriately chosen assemblages $\mathcal{M}$ of \ac{MUB} measurement assemblages that are also mutually unbiased to the incoherent basis. However, it is actually not necessary that the measurements within the assemblage are \ac{MUB} themselves. All what is needed is that Eq.~\eqref{coherence_condition} holds true for an assemblage $\mathcal{M}$ of rank-$1$ projections. For instance, the \ac{CGLMP} (see Eq.~\eqref{CGLMP_Alice} and Eq.~\eqref{CGLMP_Bob}) measurements are also a valid choices. Under the condition in Eq.~\eqref{coherence_condition} it is easy to see that the choices $C_{a \vert x} = \dfrac{M_{a \vert x}}{d}, \rho_x = \sum_a \dfrac{M_{a \vert x}}{d} = \dfrac{\mathds{1}}{d}$, and $\ell_{x,i} = \dfrac{p(x)}{d^2}$ for the dual problem in Eq.~\eqref{Coherence_SDP_dual} lead to $\Cdiamond(\mathcal{M}_\mathbf{p}) \geq 1-\dfrac{1}{d}$ for any probability distribution $p(x)$. Since $\IFdiamond(\mathcal{M}_\mathbf{p}) = 1-\dfrac{1}{d}$ and $\IFdiamond(\mathcal{M}_\mathbf{p}) \geq \Cdiamond(\mathcal{M}_\mathbf{p})$ for any assemblage of rank-$1$ projections it has to hold $\IFdiamond(\mathcal{M}_\mathbf{p}) = \Cdiamond(\mathcal{M}_\mathbf{p})$ whenever the condition in Eq.\,\eqref{coherence_condition} is fulfilled. 

\section{More distances}

In the main text, we defined general distances between measurement assemblages in Definition~\ref{Def_Distance}. However, so far we only focused on one particular distance. Here, we introduce more examples of distances for measurements and discuss their basic properties. \\
\indent We start by introducing the Schatten $\mathrm{p}-$norm functions $\DP(\mathcal{M}_\mathbf{p},\mathcal{N}_\mathbf{p})$ for $\mathrm{p} \in [1, \infty)$, defined as
\begin{align}
\DP(\mathcal{M}_\mathbf{p},\mathcal{N}_\mathbf{p}) =  \sum_{a,x} p(x)   \dfrac{1}{2} \lVert M_{a \vert x} - N_{a\vert x} \rVert_{\mathrm{p}},
\end{align}
where $\lVert X \rVert_{\mathrm{p}} = (\mathrm{Tr}[\lvert X \rvert^{\mathrm{p}}])^{1/\mathrm{p}}$ is the Schatten $\mathrm{p}-$norm of $X$. Note that the cases $\mathrm{p}=1$ and $\mathrm{p}=\infty$ correspond to the trace norm, respectively the spectral norm. While the functions $\DP(\mathcal{M}_\mathbf{p},\mathcal{N}_\mathbf{p})$ will generally not fulfil the monotonicity under Hilbert-Schmidt adjoint channels $\Lambda^{\dagger}$ according to Definition~\ref{Def_Distance}, we will show in the following that the $\mathrm{p}=\infty$ case corresponds to a proper distance. Note that for $\mathrm{p}=1$, the monotonicity under quantum channel $\Lambda^{\dagger}$ is not fulfilled, which can be seen by considering trivial extensions of the form $\Lambda^{\dagger}(M_{a \vert x}) = \mathds{1} \otimes M_{a \vert x}$. Nevertheless, we also define the induced functions 
\begin{align}
\RP(\mathcal{M}_\mathbf{p}) = \min\limits_{\mathcal{F} \in \mathscr{F}}  \DP(\mathcal{M}_\mathbf{p},\mathcal{F}_\mathbf{p})\label{Schatten_functions}
\end{align}
We formulate the following theorem to show that $\Dinfty(\mathcal{M}_\mathbf{p},\mathcal{N}_\mathbf{p})$ is a distance between measurement assemblages. 
\begin{theorem}
The function $\Dinfty(\mathcal{M}_\mathbf{p}, \mathcal{N}_\mathbf{p})$ is a distance function between the~\acp{WMA} $\mathcal{M}_\mathbf{p}$ and $\mathcal{N}_\mathbf{p}$, i.e., it fulfils all the conditions stated in Definition~\ref{Def_Distance}. Moreover, $\Dinfty(\mathcal{M}_\mathbf{p}, \mathcal{N}_\mathbf{p})$ is jointly-convex.
\end{theorem}
\begin{proof}
The proof can be split into several parts. Note first, that since $\Dinfty(\mathcal{M}_\mathbf{p}, \mathcal{N}_\mathbf{p})$ is a weighted sum of spectral norms, it follows that $\Dinfty(\mathcal{M}_\mathbf{p}, \mathcal{N}_\mathbf{p}) \geq 0$ with equality holding if and only if $\mathcal{M} = \mathcal{N}$. Note further that the symmetry and triangle inequality condition in Definition~\ref{Def_Distance} are fulfilled trivially. \\ 
\begin{widetext}
\indent For the monotonicity under quantum channel, we consider a more general data-processing type inequality for the $\infty-$ distance between two POVM elements. Namely, for $\Lambda^{\dagger}(M_{a\vert x})$ and $\Lambda^{\dagger}(N_{a \vert x})$, where $\Lambda^{\dagger}$ is a unital completely positive map, it follows
\begin{align}
\lVert \Lambda^{\dagger}(M_{a \vert x}) - \Lambda^{\dagger}(N_{a \vert x}) \rVert_{\infty} 
&= \max\limits_{\rho} \lvert \mathrm{Tr}[(\Lambda^{\dagger}(M_{a \vert x}) - \Lambda^{\dagger}(N_{a \vert x})) \rho] \rvert \label{Data-processing_operator_norm}  \\
&= \max\limits_{\rho} \lvert \mathrm{Tr}[(M_{a \vert x} - N_{a \vert x}) \Lambda(\rho)] \rvert \nonumber \\
&\leq \max\limits_{\rho'} \lvert \mathrm{Tr}[(M_{a \vert x} - N_{a \vert x}) \rho'] \rvert \nonumber \\
&= \Vert M_{a \vert x} - N_{a \vert x} \Vert_{\infty} , \nonumber
\end{align}
where we used the dual representation of the Schatten-$\infty$ norm and the fact that maximum is always achieved for a density matrix $\rho$ (more specifically the projector onto the eigenvalue of largest absolute value of $\Lambda^{\dagger}(M_{a \vert x}) - \Lambda^{\dagger}(N_{a \vert x})$). Furthermore, we used that the adjoint of the unital completely positive map $\Lambda^{\dagger}$ is a \ac{CPT} map $\Lambda$ which maps density matrices onto density matrices and therefore shrinks the state-space one optimizes over. Since $\Dinfty(\mathcal{M}_\mathbf{p}, \mathcal{N}_\mathbf{p})$ is a sum of norms $ \Vert M_{a \vert x} - N_{a \vert x} \Vert_{\infty}$, it follows that $\Dinfty(\mathcal{M}_\mathbf{p}, \mathcal{N}_\mathbf{p}) \geq \Dinfty(\Lambda^{\dagger}(\mathcal{M})_\mathbf{p}, \Lambda^{\dagger}(\mathcal{N})_\mathbf{p})$. \\
\indent The monotonicity under classical simulations $\xi(\mathcal{M}_\mathbf{p})_\mathbf{q}$ follows by direct computation,
\begin{align}
\Dinfty(\xi(\mathcal{M}_\mathbf{p})_\mathbf{q}, \xi(\mathcal{N}_\mathbf{p})_\mathbf{q})
&= \dfrac{1}{2} \sum_{b,y} q(y) \lVert \sum_{x,a} p(x \vert y) q(b \vert y,x,a)[M_{a \vert x} - N_{a \vert x}] \rVert_{\infty}  \\
&\leq \dfrac{1}{2} \sum_{b,y} q(y) \sum_{x,a} p(x \vert y) q(b \vert y,x,a) \lVert M_{a \vert x} - N_{a \vert x} \rVert_{\infty} \nonumber \\
&= \dfrac{1}{2} \sum_{y,x,a} q(y) p(x \vert y)  \lVert M_{a \vert x} - N_{a \vert x} \rVert_{\infty} \nonumber \\
&= \dfrac{1}{2} \sum_{x,a} p(x) \lVert M_{a \vert x} - N_{a \vert x} \rVert_{\infty} = \Dinfty(\mathcal{M}_\mathbf{p}, \mathcal{N}_\mathbf{p}), \nonumber 
\end{align}
\end{widetext}
where we used the following properties. \\
\indent In the first line, we used the definition of $\Dinfty(\xi(\mathcal{M}_\mathbf{p})_\mathbf{q}, \xi(\mathcal{N}_\mathbf{p})_\mathbf{q})$ by introducing the assemblages $\mathcal{M}'_\mathbf{q}= \xi(\mathcal{M}_\mathbf{p})_\mathbf{q}$ and $\mathcal{N}'_\mathbf{q}=\xi(\mathcal{N}_\mathbf{p})_\mathbf{q}$ where we inserted $ M'_{b \vert y} = \sum_x p(x \vert y) \sum_a q(b \vert y,x,a) M_{a \vert x}$ and $N'_{b \vert y} = \sum_x p(x \vert y) \sum_a q(b \vert y,x,a) N_{a \vert x}$ directly. In the second line, we used the triangle inequality. In the third line, we performed the sum over $b$. Finally, in the fourth line, we used that $\sum_y q(y) p(x \vert y) = p(x)$, which leads exactly to the definition of $\Dinfty(\mathcal{M}_\mathbf{p}, \mathcal{N}_\mathbf{p})$ from which the monotonicity under classical simulations $\xi$ follows. Therefore, $\Dinfty(\mathcal{M}_\mathbf{p}, \mathcal{N}_\mathbf{p})$ is a distance between measurement assemblages according to Definition~\ref{Def_Distance}. \\
\indent The proof of the joint-convexity of $\Dinfty(\mathcal{M}_\mathbf{p}, \mathcal{N}_\mathbf{p})$ follows exactly the same lines as the proof of the joint-convexity of $\Ddiamond(\mathcal{M}_\mathbf{p}, \mathcal{N}_\mathbf{p})$  in Theorem \ref{thrm1} and can be adapted from there.
\end{proof}
Even though they are not resource monotones generally, the functions $\mathrm{R}_{\mathrm{p}}(\mathcal{M}_\mathbf{p})$ in Eq.\,\eqref{Schatten_functions} can be used to bound the resource quantifier $\Rdiamond(\mathcal{M}_\mathbf{p})$ defined in Eq.\,\eqref{diamond_monotone}. More specifically,
we derive in the following the bounds on the diamond distance based quantifier $\Rdiamond(\mathcal{M}_\mathbf{p})$ given by
\begin{align}
\dfrac{1}{d}\mathrm{R}_{\infty}(\mathcal{M}_\mathbf{p}) \leq \dfrac{1}{d}\mathrm{R}_{1}(\mathcal{M}_\mathbf{p}) &\leq \Rdiamond(\mathcal{M}_\mathbf{p}) \leq \mathrm{R}_{\infty}(\mathcal{M}_\mathbf{p}) \leq \mathrm{R}_{1}(\mathcal{M}_\mathbf{p}),
\end{align}
where $d$ is the dimension of the Hilbert space $\mathcal{H}$ the POVMs from $\mathcal{M}$ act on. Note that due to the monotonicity of Schatten norms, it holds $\lVert X \rVert_{\mathrm{p}} \leq \lVert X \rVert_{\mathrm{p}'}$ for $\mathrm{p} \geq \mathrm{p}'$ from which the bounds $\dfrac{1}{d}\mathrm{R}_{\infty}(\mathcal{M}_\mathbf{p}) \leq \dfrac{1}{d}\mathrm{R}_{1}(\mathcal{M}_\mathbf{p})$ and $\mathrm{R}_{\infty}(\mathcal{M}_\mathbf{p}) \leq \mathrm{R}_{1}(\mathcal{M}_\mathbf{p})$ follow directly. \\
\begin{widetext}
\indent The bound $\Rdiamond(\mathcal{M}_\mathbf{p}) \leq \mathrm{R}_{\infty}(\mathcal{M}_\mathbf{p})$ follows from
\begin{align}
\Rdiamond(\mathcal{M}_\mathbf{p})
&= \min\limits_{\mathcal{F} \in \mathscr{F}} \dfrac{1}{2} \sum_x p(x) \max\limits_{\rho} \sum_a \lVert  \mathrm{Tr}_1[(M_{a \vert x} \otimes \mathds{1}) \rho] - \mathrm{Tr}_1[(F_{a \vert x} \otimes \mathds{1}) \rho] \rVert_1  \\
&\leq \min\limits_{\mathcal{F} \in \mathscr{F}} \dfrac{1}{2} \sum_x p(x) \max\limits_{\rho} \sum_a \lVert  (M_{a \vert x} \otimes \mathds{1}) \rho - (F_{a \vert x} \otimes \mathds{1}) \rho \rVert_1 \nonumber \\
&\leq \min\limits_{\mathcal{F} \in \mathscr{F}} \dfrac{1}{2} \sum_{a,x} p(x) \max\limits_{\rho}  \lVert  (M_{a \vert x} \otimes \mathds{1})  - (F_{a \vert x} \otimes \mathds{1})  \rVert_{\infty} \lVert \rho \rVert_1 \nonumber \\
&=  \min\limits_{\mathcal{F} \in \mathscr{F}} \dfrac{1}{2} \sum_{a,x} p(x)   \Vert  (M_{a \vert x} \otimes \mathds{1})  - (F_{a \vert x} \otimes \mathds{1})  \Vert_{\infty} = \mathrm{R}_{\infty}(\mathcal{M}_\mathbf{p}), \nonumber
\end{align} 
where we used the definition of $\Rdiamond(\mathcal{M}_\mathbf{p})$ in the first line and the monotonicity of the trace norm under partial trace in the second line. In the third line, we used the Hölder inequality and in the last line identified the definition of $\mathrm{R}_{\infty}(\mathcal{M}_\mathbf{p})$. 

The last remaining bound $\dfrac{1}{d}\mathrm{R}_{1}(\mathcal{M}_\mathbf{p}) \leq \Rdiamond(\mathcal{M}_\mathbf{p})$ can directly be obtained by using $\rho = \vert \Phi^+ \rangle \langle \Phi^+ \vert $ within the optimization of the diamond norm. Here, $\vert \Phi^+ \rangle = \dfrac{1}{\sqrt{d}}\sum_{i = 0}^{d-1} \vert ii \rangle $ is the maximally entangled state, where (as before) $d$ is the dimension the POVMs of the measurement assemblage $\mathcal{M}$ act on. It follows
\begin{align}
\Rdiamond(\mathcal{M}_\mathbf{p})
&= \min\limits_{\mathcal{F} \in \mathscr{F}} \dfrac{1}{2} \sum_x p(x) \max\limits_{\rho} \sum_a \lVert  \mathrm{Tr}_1[(M_{a \vert x} \otimes \mathds{1}) \rho] - \mathrm{Tr}_1[(F_{a \vert x} \otimes \mathds{1}) \rho] \rVert_1  \\
&\geq \min\limits_{\mathcal{F} \in \mathscr{F}} \dfrac{1}{2} \sum_x p(x) \sum_a \lVert  \mathrm{Tr}_1[(M_{a \vert x} \otimes \mathds{1}) \vert \Phi^+ \rangle \langle \Phi^+ \vert] - \mathrm{Tr}_1[(F_{a \vert x} \otimes \mathds{1}) \vert \Phi^+ \rangle \langle \Phi^+ \vert] \rVert_1 \nonumber \\
&= \min\limits_{\mathcal{F} \in \mathscr{F}} \dfrac{1}{2} \sum_x p(x) \sum_a \dfrac{1}{d} \lVert (M_{a \vert x}-F_{a \vert x})^T \rVert_1 \nonumber \\
&= \min\limits_{\mathcal{F} \in \mathscr{F}} \dfrac{1}{2} \sum_x p(x) \sum_a \dfrac{1}{d} \lVert M_{a \vert x}-F_{a \vert x} \rVert_1 = \dfrac{1}{d} \mathrm{R}_1(\mathcal{M}_\mathbf{p}), \nonumber
\end{align}
where we used in the first line the definition of $\Rdiamond(\mathcal{M}_\mathbf{p})$ and in the second line that $\rho =  \vert \Phi^+ \rangle \langle \Phi^+ \vert$ is a feasible point within the maximization over the quantum states within the diamond norm. In the third line, we used that $\mathrm{Tr}_1[(M_{a \vert x} \otimes \mathds{1}) \vert \Phi^+ \rangle \langle \Phi^+ \vert] = \dfrac{1}{d}M_{a \vert x}^T$, where the transposition is with respect to the computational basis. Finally, we can use that a transposition does not change the singular values. 
\end{widetext} \\
\indent The monotone $\mathrm{R}_{\infty}(\mathcal{M}_\mathbf{p})$ in particular is not only a valuable tool to bound the diamond distance $\Rdiamond(\mathcal{M}_\mathbf{p})$ but is also interesting in itself. More specifically, we show in the following that $\mathrm{R}_{\infty}(\mathcal{M}_\mathbf{p})$ obeys also a measurement hierarchy similar to that in Eq.\,\eqref{full_hierarchy}. Let $\IFinfty(\mathcal{M}_\mathbf{p}),\Cinfty(\mathcal{M}_\mathbf{p}),$ and $\Iinfty(\mathcal{M}_\mathbf{p})$ be the informativeness, coherence, and incompatibility of ${M}_\mathbf{p}$ as measured by the distance $\mathrm{R}_{\infty}(\mathcal{M}_\mathbf{p})$ in Eq.\,\eqref{Schatten_functions} with respect to the free sets $\mathscr{F}_{\mathrm{UI}}$, $\mathscr{F}_{\mathrm{IC}}$, and $\mathscr{F}_{\mathrm{JM}}$. It follows directly from $\mathscr{F}_{\mathrm{UI}} \subset \mathscr{F}_{\mathrm{IC}} \subset \mathscr{F}_{\mathrm{JM}}$ that the hierarchy
\begin{align}
\IFinfty(\mathcal{M}_{\mathbf{p}_A}) \geq \Cinfty(\mathcal{M}_{\mathbf{p}_A}) \geq \Iinfty(\mathcal{M}_{\mathbf{p}_A}),
\end{align}
holds. Moreover, it can be shown that $\Iinfty(\mathcal{M}_{\mathbf{p}_A}) \geq \mathrm{S}(\Vec{\sigma}_{\mathbf{p}_A} )$ (remember that we already showed that $\mathrm{S}(\Vec{\sigma}_{\mathbf{p}_A} ) \geq \mathrm{N}(\mathbf{q}_\mathbf{p})$). This follows from the direct computation for any quantum state $\rho$ of appropriate dimension and the closest \ac{JM} assemblage $\mathcal{F}^*$ to $\mathcal{M}$ (with respect to the monotone $\Iinfty(\mathcal{M}_\mathbf{p})$):
\begin{align}
\mathrm{S}(\Vec{\sigma}_{\mathbf{p}_A} ) &\leq \dfrac{1}{2} \sum_{a,x} p_A(x)  \lVert \mathrm{Tr}_1[((M_{a \vert x} - F^*_{a\vert x}) \otimes \mathds{1} )\rho] \rVert_{1} \\
&\leq \dfrac{1}{2} \sum_{a,x} p_A(x)  \lVert ((M_{a \vert x} - F^*_{a\vert x}) \otimes \mathds{1} )\rho \rVert_{1} \nonumber \\
&\leq \dfrac{1}{2} \sum_{a,x} p_A(x)  \lVert (M_{a \vert x} - F^*_{a\vert x}) \otimes \mathds{1} \rVert_{\infty} \lVert \rho \rVert_{1} \nonumber \\ &= \dfrac{1}{2} \sum_{a,x} p_A(x)  \Vert M_{a \vert x} - F^*_{a\vert x} \Vert_{\infty} = \Iinfty(\mathcal{M}_{\mathbf{p}_A}), \nonumber
\end{align}
where we first used that \ac{JM} measurements always lead to unsteerable assemblages. Second, we used that the trace norm is non-increasing under partial traces. Third, we used the Hölder inequality. It therefore follows, that the hierarchy 
\begin{align}
\IFinfty(\mathcal{M}_{\mathbf{p}_A}) \geq \Cinfty(\mathcal{M}_{\mathbf{p}_A}) \geq \Iinfty(\mathcal{M}_{\mathbf{p}_A}) \geq \mathrm{S}(\Vec{\sigma}_{\mathbf{p}_A} ) \geq \mathrm{N}(\mathbf{q}_\mathbf{p}),
\end{align}
holds. \\
\indent Another distance for measurement assemblages that can be considered is based on the $\ell_1$-distance between probability distributions. More specifically, the induced $\ell_1$-distance between two \acp{WMA} is given by
\begin{align}
\Dell(\mathcal{M}_\mathbf{p},\mathcal{N}_\mathbf{p}) =  \dfrac{1}{2} \sum_{x} p(x) \max_{\rho_A } \sum_a  \vert \mathrm{Tr}[(M_{a \vert x} - N_{a\vert x}) \rho_A] \vert, 
\end{align}
which is the $\ell_1$-distance of the probability distributions $\lbrace \mathrm{Tr}[M_{a \vert x} \rho_A] \rbrace$ and $\lbrace \mathrm{Tr}[N_{a \vert x} \rho_A] \rbrace$ maximized over all quantum states $\rho_A$. With the same methods as for the distances $\Dinfty(\mathcal{M}_\mathbf{p},\mathcal{N}_\mathbf{p})$ and $\Ddiamond(\mathcal{M}_\mathbf{p},\mathcal{N}_\mathbf{p})$ it can be shown that $\Dell(\mathcal{M}_\mathbf{p},\mathcal{N}_\mathbf{p})$ fulfills all the conditions in Definition\,\ref{Def_Distance}. Moreover, the joint-convexity of $\Dell(\mathcal{M}_\mathbf{p},\mathcal{N}_\mathbf{p})$ follows from
\begin{align}
&\Dell(\eta \mathcal{M}^{(1)}_\mathbf{p} + (1-\eta) \mathcal{M}^{(2)}_\mathbf{p},\eta \mathcal{N}^{(1)}_\mathbf{p} + (1-\eta) \mathcal{N}^{(2)}_\mathbf{p}) \\
&= \dfrac{1}{2} \sum_x p(x) \max\limits_{\rho_A} \sum_a \lvert \mathrm{Tr}[\eta (M_{a \vert x}^{(1)}-N_{a \vert x}^{(1)})+ (1-\eta)(M_{a \vert x}^{(2)}-N_{a \vert x}^{(2)})] \rvert \nonumber \\
&\leq \dfrac{1}{2} \sum_x p(x) \max\limits_{\rho_A} \sum_a \eta \lvert \mathrm{Tr}[(M_{a \vert x}^{(1)}-N_{a \vert x}^{(1)}) \rho_A] \rvert + (1-\eta) \lvert \mathrm{Tr}[(M_{a \vert x}^{(2)}-N_{a \vert x}^{(2)}) \rho_A] \rvert \nonumber \\
&\leq \dfrac{1}{2} \sum_x p(x) \max\limits_{\rho_A^{(1)},\rho_A^{(2)}} \sum_a \eta \lvert \mathrm{Tr}[(M_{a \vert x}^{(1)}-N_{a \vert x}^{(1)}) \rho_A^{(1)}] \rvert + (1-\eta) \lvert \mathrm{Tr}[(M_{a \vert x}^{(2)}-N_{a \vert x}^{(2)}) \rho_A^{(2)}] \rvert \nonumber \\
&= \eta \Dell(\mathcal{M}^{(1)}_\mathbf{p},\mathcal{N}^{(1)}_\mathbf{p}) + (1-\eta) \Dell(\mathcal{M}^{(2)}_\mathbf{p},\mathcal{N}^{(2)}_\mathbf{p}), \nonumber
\end{align}
where we used the homogeneity and the triangle inequality of the absolute value and the fact that the individual maximization over both terms within the sum over $a$ can only increase the value. 
Hence, $\Dell(\mathcal{M}_\mathbf{p},\mathcal{N}_\mathbf{p})$ is a jointly-convex distance function which induces the distance-based convex monotone
\begin{align}
\Rell(\mathcal{M}_\mathbf{p}) = \min\limits_{\mathcal{F} \in \mathscr{F}}  \Dell(\mathcal{M}_\mathbf{p},\mathcal{F}_\mathbf{p}).    
\end{align}
\indent Note that while it follows directly that $\Rell(\mathcal{M}_\mathbf{p})$ will naturally induce a hierarchy between the informativeness, coherence, and the incompatibility of a \ac{WMA} $\mathcal{M}_\mathbf{p}$, it is not clear whether there exist steering or nonlocality monotones that are in natural correspondence to it.
Note further that in the context of coherence of single \acp{POVM}, this kind of statistical measure has also been defined by Baek et al.~\cite{Baek2020}.

Even though it is not clear whether a complete hierarchy of measurement resources holds, the quantifier $\Rell(\mathcal{M}_\mathbf{p})$ is important, as it can be seen as limiting case of the quantifier $\Rdiamond(\mathcal{M}_\mathbf{p})$ when the maximization is performed only over product states. More formally, it holds
\begin{align}
\Rdiamond(\mathcal{M}_\mathbf{p}) &\geq  \min\limits_{\mathcal{F} \in F} \dfrac{1}{2} \sum_x p(x) \max\limits_{\rho = \rho_A \otimes \rho_B} \sum_a \lVert \mathrm{Tr}_1[((M_{a \vert x}-F_{a \vert x}) \otimes \mathds{1}) \rho] \rVert_1  \\
&= \min\limits_{\mathcal{F} \in F} \dfrac{1}{2} \sum_x p(x) \max\limits_{\rho = \rho_A \otimes \rho_B} \sum_a \lVert \mathrm{Tr}[(M_{a \vert x}-F_{a \vert x}) \rho_A] \rho_B \rVert_1 \nonumber \\
&= \min\limits_{\mathcal{F} \in F} \dfrac{1}{2} \sum_x p(x) \max\limits_{\rho = \rho_A \otimes \rho_B} \sum_a \lvert \mathrm{Tr}[(M_{a \vert x}-F_{a \vert x}) \rho_A] \rvert \lVert \rho_B \rVert_1 \nonumber \\
&= \min\limits_{\mathcal{F} \in F} \dfrac{1}{2} \sum_x p(x) \max\limits_{\rho_A} \sum_a \lvert \mathrm{Tr}[(M_{a \vert x}-F_{a \vert x}) \rho_A] \rvert \nonumber \\
&= \Rell(\mathcal{M}_\mathbf{p}), \nonumber
\end{align}
where we used in the first line that we maximize only over the set of product states. In the second line we used the definition of the partial trace and finally, we used that states $\rho_B$ are normalized in the $1-$norm and identified the last line with the definition of the induced $\ell_1$-distance quantifier $\Rell(\mathcal{M}_\mathbf{p})$. In Appendix \ref{DichotomicMeasurements}, we show that the quantifiers $\Rell(\mathcal{M}_\mathbf{p})$ and  $\mathrm{R}_{\infty}(\mathcal{M}_\mathbf{p})$ coincide with $\Rdiamond(\mathcal{M}_\mathbf{p})$ in the case of dichotomic measurement assemblages.

\section{Dichotomic measurements}
\label{DichotomicMeasurements}

Here, we show an additional property of $\Rdiamond(\mathcal{M}_\mathbf{p})$ that can be useful for the case where we consider measurement assemblages $\mathcal{M}$ with only two outcomes for each setting~$x$. We show that in this special case, the diamond distance quantifier $\Rdiamond(\mathcal{M}_\mathbf{p})$ is equivalent to $\mathrm{R}_{\infty}(\mathcal{M}_\mathbf{p})$ and $\Rell(\mathcal{M}_\mathbf{p})$. Consider the \acp{WMA} $\lbrace M_{1 \vert x}, \mathds{1}-M_{1 \vert x} \rbrace_x$ and $\lbrace F_{1 \vert x}, \mathds{1}-F_{1 \vert x} \rbrace_x$. Remember that we already showed previously that $\Rdiamond(\mathcal{M}_\mathbf{p}) \leq \mathrm{R}_{\infty}(\mathcal{M}_\mathbf{p})$, so we only need to show that in the case of dichotomic measurements it also holds $\Rdiamond(\mathcal{M}_\mathbf{p}) \geq \mathrm{R}_{\infty}(\mathcal{M}_\mathbf{p}) = \Rell(\mathcal{M}_\mathbf{p})$. This follows directly via
\begin{align}
\Rdiamond(\mathcal{M}_\mathbf{p}) \geq \Rell(\mathcal{M}_\mathbf{p}) &= 2 \dfrac{1}{2}\min\limits_{\mathcal{F} \in \mathscr{F}} \sum_{x} p(x) \max_{ \rho_A }  \lvert \mathrm{Tr}[(M_{a \vert x} - F_{a\vert x}) \rho_A] \rvert = \mathrm{R}_{\infty}(\mathcal{M}_\mathbf{p}),
\end{align}
where we used the bound $\Rdiamond(\mathcal{M}_\mathbf{p}) \geq \Rell(\mathcal{M}_\mathbf{p})$ and the fact that for dichotomic measurements both outcomes contribute equally towards $\Rell(\mathcal{M}_\mathbf{q})$. Finally, we used that this holds also true for $\mathrm{R}_{\infty}(\mathcal{M}_\mathbf{p})$. Alternatively, it is also enough to see that the same $\rho_A$ is optimal for both outcomes, which leads to the conclusion that $\mathrm{R}_{\infty}(\mathcal{M}_\mathbf{p}) = \Rell(\mathcal{M}_\mathbf{p})$. Note that the above result shows that entanglement does not offer an advantage in distinguishing two measure-and-prepare channels for dichotomic measurements by means of the diamond norm.

\newpage

\section*{Acronyms}
\begin{acronym}[CGLMP]\itemsep 1\baselineskip
\acro{AGF}{average gate fidelity}
\acro{AMA}{associated measurement assemblage}

\acro{BOG}{binned outcome generation}

\acro{CGLMP}{Collins-Gisin-Linden-Massar-Popescu}
\acro{CHSH}{Clauser-Horne-Shimony-Holt}
\acro{CP}{completely positive}
\acro{CPT}{completely positive and trace preserving}
\acro{CS}{compressed sensing} 

\acro{DFE}{direct fidelity estimation} 
\acro{DM}{dark matter}

\acro{GST}{gate set tomography}
\acro{GUE}{Gaussian unitary ensemble}

\acro{HOG}{heavy outcome generation}

\acro{JM}{jointly measurable}

\acro{LHS}{local hidden-state model}
\acro{LHV}{local hidden-variable model}
\acro{LOCC}{local operations and classical communication}

\acro{MBL}{many-body localization}
\acro{ML}{machine learning}
\acro{MLE}{maximum likelihood estimation}
\acro{MPO}{matrix product operator}
\acro{MPS}{matrix product state}
\acro{MUB}{mutually unbiased bases} 
\acro{MW}{micro wave}

\acro{NISQ}{noisy and intermediate scale quantum}

\acro{POVM}{positive operator valued measure}
\acro{PVM}{projector-valued measure}

\acro{QAOA}{quantum approximate optimization algorithm}
\acro{QML}{quantum machine learning}
\acro{QMT}{measurement tomography}
\acro{QPT}{quantum process tomography}
\acro{QRT}{quantum resource theory}
\acroplural{QRT}[QRTs]{Quantum resource theories}

\acro{RDM}{reduced density matrix}

\acro{SDP}{semidefinite program}
\acro{SFE}{shadow fidelity estimation}
\acro{SIC}{symmetric, informationally complete}
\acro{SPAM}{state preparation and measurement}

\acro{RB}{randomized benchmarking}
\acro{rf}{radio frequency}

\acro{TT}{tensor train}
\acro{TV}{total variation}

\acro{UI}{uninformative}

\acro{VQA}{variational quantum algorithm}

\acro{VQE}{variational quantum eigensolver}

\acro{WMA}{weighted measurement assemblage}

\acro{XEB}{cross-entropy benchmarking}

\end{acronym}

\section*{Code availability}

Matlab codes that accompany this work are available on \href{https://github.com/LucasTendick/DistanceBasedResourceQuantificationSetsMeasurements}{GitHub}. These contain in particular implementations of the \acp{SDP} presented here. 

\bibliographystyle{./myapsrev4-2}
\bibliography{bibliography_2}

\end{document}